%% file: main.tex
\documentclass[english,reprint,longbibliography]{revtex4-1}
\input{commands.tex}

\begin{document}
\bstctlcite{BSTcontrol} %

\title{Work, entropy production, and thermodynamics of information under protocol constraints}
\author{Artemy Kolchinsky}
\altaffiliation{artemyk@gmail.com}
\affiliation{Santa Fe Institute, Santa Fe, New Mexico}

\author{David H. Wolpert}
\altaffiliation{
Complexity Science Hub, Vienna; Arizona State University, Tempe, Arizona; 
\texttt{http://davidwolpert.weebly.com}}
\affiliation{Santa Fe Institute, Santa Fe, New Mexico}

\begin{abstract}
In many real-world situations, there are constraints on the ways in which a physical system can be manipulated.  
We investigate the entropy production (EP) and extractable work involved in bringing a system from some initial distribution
$\ps$ to some final distribution $\pf$, given that the set of master equations available to the driving protocol obeys some constraints. 
We first derive general bounds on EP and extractable work, 
as well as a decomposition of the nonequilibrium free energy into an ``accessible free energy'' (which can be extracted as work, given a set of constraints) and an ``inaccessible free energy'' (which must be dissipated as EP).
In a similar vein, we consider the thermodynamics of information in the presence of constraints, and decompose the information acquired in a measurement into ``accessible'' and ``inaccessible'' components.   This decomposition allows us 
to consider the thermodynamic efficiency of different measurements of the same system, given a set of constraints. 
 We use our framework to analyze protocols subject to symmetry,  
modularity, and coarse-grained constraints, and consider various examples including the Szilard box, the 2D %
Ising model, and a multi-particle flashing ratchet. %
\end{abstract}
\maketitle

\section{Introduction}

\subsection{Background}

One of the foundational issues in thermodynamics is quantifying how much work is required to transform %
a system between two thermodynamic states. 
 Recent
results in statistical physics have derived general bounds on work which hold even for transformations between nonequilibrium
states~\citep{takara_generalization_2010,parrondo2015thermodynamics}.
In particular, suppose one wishes to transform a system with initial distribution $p$
and energy function $\Hs$ to some final distribution $p'$ and energy
function $\Hf$. For an \emph{isothermal} process, during which the system remains in contact with
a single heat bath at inverse temperature $\beta$, the work extracted
during this transformation obeys %
\begin{equation}
\Wp \le \NFEs  -\NFEf,\label{eq:workbound}
\end{equation}
where $\NFEs :=\left\langle \HorE\right\rangle _{p}-S(p)/\beta$
is the \emph{(nonequilibrium) free energy} of distribution $p$ given energy function $E$~\citep{takara_generalization_2010,parrondo2015thermodynamics,esposito2011second}. %
This inequality comes from the second law of thermodynamics, which states that \emph{entropy production} (EP), the total increase of the entropy of the system and all coupled reservoirs, is non-negative. For an isothermal process that carries out the transformation $\ptpp$, EP is given by %
\begin{align}
\EPp = \beta[\NFEs-\NFEf-\Wp] \ge 0.
\label{eq:EPw}
\end{align}
\cref{eq:workbound}  follows from \cref{eq:EPw} by a simple rearrangement.

To extract work from a system, %
one must manipulate the system by applying a {driving protocol}. 
There are many different driving protocols that can be used to transform some initial distribution $p$ to some final distribution $p'$, which generally 
incur different amounts of  EP and work. 
Achieving the fundamental bounds set by the second
law, such as \cref{eq:workbound}, typically requires 
idealized protocols, which make use of arbitrary energy functions, infinite timescales, etc.  
In  many real-world scenarios, however, there are strong practical constraints on how one can manipulate a  system, and %
 such idealized protocols are unavailable.  

The goal of this paper is to derive stronger bounds on EP and work involves in carrying out %
the transformation
$p \rightarrow p'$, given constraints on the set of master equations available to the driving protocol. 
Ultimately, such stronger bounds on EP and work can provide new insights 
into various real-world thermodynamic processes and work-harvesting devices, ranging from biological
organisms to artificial engines. %
They %
can also cast new light on some %
well-studied scenarios in
statistical physics.

\begin{figure}[b]
\begin{centering}
\includegraphics[width=0.65\columnwidth]{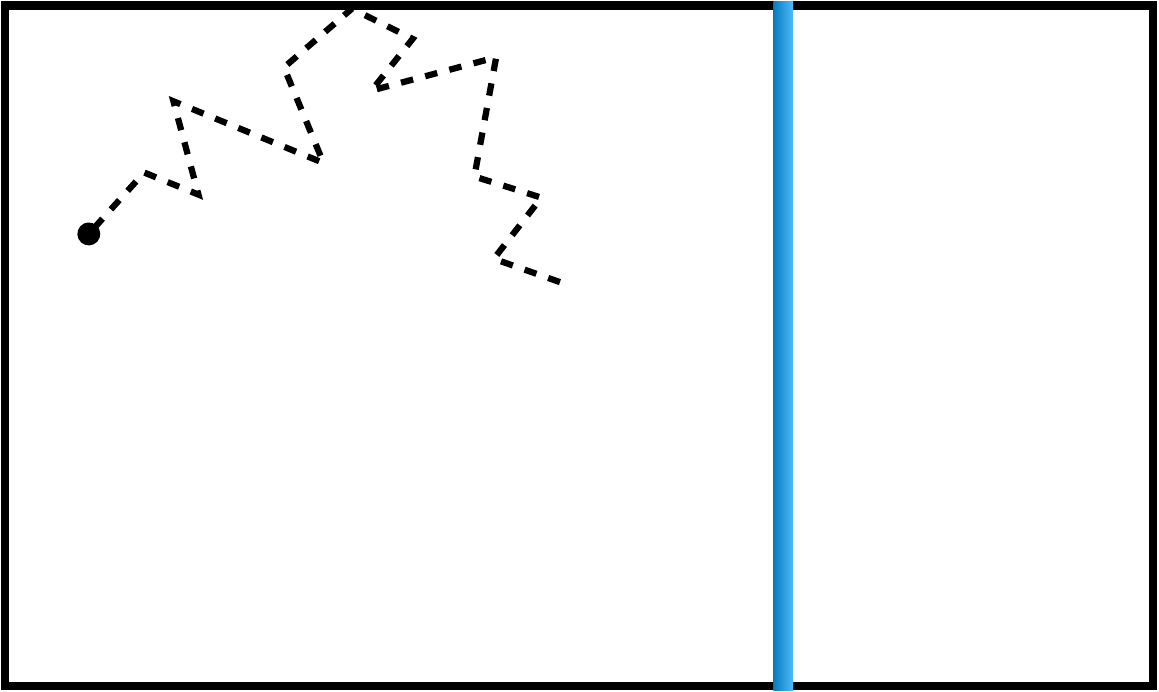}
\par\end{centering} 
\caption{A two-dimensional Szilard box with a single Brownian particle, where a vertical partition
(blue) can be positioned at different horizontal locations in
the box. We demonstrate that only information about the particle's
horizontal position, not its vertical position, can be used to extract work from the system.\label{fig:szilard0}}
\end{figure}

For example, consider a two-dimensional Szilard
box connected to a heat bath~\footnote{We use a Brownian model of the Szilard engine, which is similar to setups commonly employed in modern nonequilibrium statistical physics~\cite{berut2012experimental,roldan2014universal,koski2014experimental,shizume1995heat,gong2016stochastic,parrondo2015thermodynamics}, as shown in \cref{fig:szilard0}.  This model can be justified by imagining a box that contains a large colloidal 
particle, as well as a medium of small solvent particles to
which the vertical partition is permeable.  Note that this model differs from Szilard's original proposal~\cite{szilard1929entropieverminderung}, in which the box contains a single particle in a vacuum, which has been analyzed in~\cite{proesmans2015efficiency,hondou2007equation,bhat2017unusual}.}, which contains a single Brownian particle and a vertical partition, and suppose that the driving protocols can manipulate the horizontal position of this partition. 
Imagine that  the particle is initially located in the \emph{left half} of the box. 
How much work can be extracted by transforming this initial distribution to a uniform final distribution, assuming the system begins and ends with a uniform energy function? %
A simple application of \cref{eq:workbound} shows that the extractable work is upper bounded by $\lntwooverbeta$. This bound can be achieved by quickly moving the vertical partition to the middle of the box, and then slowly expanding it rightward.  Now imagine an alternative scenario, in which the particle is initially located in the \emph{top half} of the box.  
By \cref{eq:workbound}, the work that can be extracted by bringing this initial distribution to a uniform final distribution is again upper bounded by $\lntwooverbeta$. Intuitively, however, it seems that this bound should not be achievable, given the constrained set of available protocols (i.e., one can only manipulate the system by moving the  vertical partition left and right).  %
Our results will make this intuition rigorous for the two-dimensional Szilard box, as well as various other systems that can only be manipulated by a constrained set of driving protocols.

This phenomenon also occurs when the starting and ending distributions can depend on %
the outcome of a measurement of the system. This kind of setup, which was first used to analyze the thermodynamics of information in various kinds of Maxwellian demons, is sometimes called ``feedback control'' in the literature~\cite{sagawa2008second,parrondo2015thermodynamics}. %
Imagine that the state of some system $X$ is first measured
using some %
observation channel (conditional distribution) $q(m\vert x)$, producing measurement outcome $m$ with probability $p(m)=\sum_x p(x)q(m\vert x)$. The system then %
 undergoes a driving protocol which can depend on $m$. 
For simplicity, we assume that the system's energy function begins as $\Hs$ and ends as $\Hf$ for all measurement outcomes. 
Let $\pconds$ and $\pcondf$ indicate the system's initial and final conditional distributions given measurement outcome $m$, %
and let $\ps(x)=\sum_m p(m)\pconds(x\vert m)$ and $\pf(x')=\sum_m p(m)\pcondf(x'\vert m)$ indicate the system's initial and final marginal distributions (for simplicity, below we often use notation like $\ps$, instead of $\ps(x)$). 
We can then take  expectations of both sides of  \cref{eq:workbound} across measurement outcomes, thereby bounding the average 
extractable work  as~\footnote{As common in the literature, in \cref{eq:workbound-2a} we consider only the work that is extractable from the system after the measurement is made. We do not account for the possible work cost of making the measurement, nor any work exchanges that may be incurred by the measurement apparatus during the driving.} %
\begin{align}
\left\langle W\right\rangle  & \le\sum_m p(m) [\NFEs[\pconds]-\NFEf[\pcondf] ].\label{eq:workbound-2a}
\end{align}
By adding and subtracting $[S(p)-S(p')]/\beta$ on the right hand side, %
we can further rewrite \cref{eq:workbound-2a} in terms of the drop of the free energy in the marginal distribution, plus the loss of information between the measurement and the system over the course of the protocol,
\begin{align}
\left\langle W\right\rangle  \le\NFEs - \NFEf+[\IXMs-\IXMf]/\beta,\label{eq:workbound-2}
\end{align}
where $\IXMs$ and $\IXMf$ indicate the mutual information under the conditional distributions $\pconds$ and $\pcondf$ respectively.  
Comparing \cref{eq:workbound} and \cref{eq:workbound-2},  the bound on average extractable work increases with the drop of mutual
information. This is a classic result from the ``thermodynamics of information''~\cite{sagawa2008second,parrondo2015thermodynamics}, which shows that information about the state of a system  can be used to increase the work extracted from this system. 

Just like \cref{eq:workbound}, the bound in \cref{eq:workbound-2}
is typically saturated by idealized protocols, which have access
to arbitrary energy functions, infinite timescales, etc. As mentioned above, in the real-world there are
typically constraints on the available protocols, in which case the bound of \cref{eq:workbound-2} may not be achievable. 
For example, consider again  the Szilard box shown in   
 \cref{fig:szilard0}.   
Imagine measuring a bit of information about the location
of the particle and then using this information to extract work from the system while
driving it back to a uniform equilibrium distribution. In this case 
 $\IXMs=\ln2$ and  %
$\IXMf=0$, so if the system starts and ends with the uniform energy function,  \cref{eq:workbound-2} states that %
$\langle W \rangle \le \lntwooverbeta$. 
Intuitively, however, it seems that measuring the particle's horizontal position
should be useful for extracting work from the system, while measuring
the particle's vertical position should not be useful. The general
bound of \cref{eq:workbound-2} does not distinguish between
these two kinds of measurements.  In fact, this bound depends
only on the overall \emph{amount}
of information acquired by the measurement (as quantified by $\IXMs$), and is therefore completely insensitive to the \emph{content} of that information (i.e., the particular pattern of correlations quantified by $\IXMs$).

\subsection{Summary of results and roadmap}

In this paper we derive bounds on extractable work and
EP %
which arise when carrying out the transformation $\ptpp$ under constraints on the driving protocol. 
We consider a system coupled to a single heat bath which undergoes a driving protocol over some time interval $t\in[0,\ft]$ (where the units of time are arbitrary). A driving protocol is represented as a continuous-time master equation $\LL(t)$, where $\LL(t)$ refers to the \emph{(infinitesimal) generator} at time $t$.    
For example, a driving protocol could be a trajectory of time-dependent discrete-state rate matrices, or a trajectory of time-dependent Fokker-Planck operators for a continuous-state system.  %

We say that a driving protocol is \emph{constrained} if there is some restricted set of generators $\LLL$ such that  $L(t)\in \LLL$ at all times $t\in[0,\ft]$. 
As discussed below, the particular choice of $\LLL$ depends on 
the specific constraints being considered. For example, $\LLL$ might represent a set of generators that are invariant under some particular symmetry group (e.g., representing the dynamics of a set of indistinguishable particles, or a spin system on a lattice with symmetries).

Our analysis proceeds at three different ``levels'' of generality, which we summarize in the following subsections.

\subsubsection*{Level 1: General mathematical framework}

In the first level of analysis, presented in \cref{sec:info-geom-framework,sec:thermo-of-info}, we provide a general mathematical framework for deriving bounds on EP and work for constrained driving protocols. %

To develop our framework, given some some set of allowed generators $\LLL$, we consider an associated operator %
operator $\klopBase$ over distributions which 
 satisfies two conditions: 
it obeys the so-called \emph{Pythagorean identity} from information geometry, and it commutes with the  dynamics generated by elements of $\LLL$ (\cref{eq:pyth,eq:comm0} below). 
Given such an operator $\klopBase$, in \cref{sec:info-geom-framework} we show  that  
for any distribution $p$,
the distribution $\klop p$ contains only that part of the free energy in $p$ which may be turned into work by a constrained driving protocol. %
Formally, we decompose the nonequilibrium free energy of distribution $p$ and energy function $E$ as
\begin{align}
\NFEs=\NFEs[\klop p]+ D(p\Vert\klop p)/\beta,\label{eq:feffdecomp0}
\end{align}
where $D(\cdot \Vert \cdot)$ indicates the Kullback-Leibler divergence.  
Then, for any constrained driving protocol that carries out the transformation $\ptpp$, the extractable work is bounded as 
\begin{align}
\Wp \le 
\NFEs[\klop p]-\NFEf[\klop \pf].\label{eq:ourworkbound0}
\end{align}
We also demonstrate that EP can be lower bounded by %
the contraction of the Kullback-Leibler (KL) divergence between $p$ and $\klop p$ over the course of the protocol,
\begin{equation}
\EPp \ge D(\ps \Vert\klop \ps)-D(\pf\Vert\klop{\pf}).\label{eq:epfreebound0}
\end{equation}

Given these bounds, it can be seen that 
\cref{eq:feffdecomp0} decomposes the nonequilibrium free energy $\NFEs$ into two terms: an \emph{accessible free energy} $\NFEs[\klop p]$, %
 whose decrease over the course of the protocol may be extractable as work, and  an \emph{inaccessible free energy} $D(p\Vert\klop p)/\beta$, %
whose decrease over the course of the protocol    %
cannot be turned into work and must be dissipated as EP.   
The accessible free energy is always less than the overall free energy, $\NFEs[\klop \ps]\le \NFEs$, which follows from \cref{eq:feffdecomp0} and the non-negativity of KL divergence. We also show
that the 
right hand side of \cref{eq:epfreebound0} is non-negative,%
\begin{equation}
D(p\Vert\klop p)-D(p'\Vert\klop{\pf}) \ge 0 ,
\label{eq:d1}
\end{equation}
which implies that our bounds on EP and work,  \cref{eq:epfreebound0,eq:ourworkbound0} respectively, are stronger than the general bounds provided by the second law ($\EP\ge 0$ and \cref{eq:workbound}). 
Note that
\cref{eq:d1} also implies %
an irreversibility condition %
on the dynamics: for any two distributions $p$ and $p'$,
a constrained driving protocol can either carry out the transformation $\ptpp$ %
or the transformation $\pptp$ %
but not both --- unless $D(p\Vert\klop p)=D(p'\Vert\klop{p'})$. 
In \cref{sec:thermo-of-info}, we show that the general framework summarized above  has %
important implications for thermodynamics of information. We consider the type of feedback-control setup discussed above: an observation apparatus first makes a measurement $m$ of the system, then the system undergoes a driving protocol (which can depend on $m$) that carries out the transformation $\dTrans{\pconds}{\pcondf}$. %
Suppose that the driving protocols corresponding to all $m$ obey bounds like \cref{eq:ourworkbound0} for the same operator $\klopBase$.  %
This operator then gives rise to %
the ``mapped'' initial and final conditional distributions $\effCPs$ and $\effCPf$. 
We can then 
bound average extractable work for feedback control under constraints as
\begin{align*}
\left\langle W\right\rangle  \le\NFEs - \NFEf+[\IaccS-\IaccF]/\beta,%
\end{align*}
where the \emph{accessible information} component of the initial mutual information $I(X;M)$ is defined as 
\begin{align}
\IaccS=\IXMs - \IinaccS,
\label{eq:introAccInfo}
\end{align}
and similarly for similarly for $\IaccF$. %
This bound is a refinement of \cref{eq:workbound-2} in the presence of protocol constraints, which shows 
that the amount of extractable work depends on the accessible information $\IaccS$, rather than the actual mutual information $\IXMs$.  
Loosely speaking, the accessible information %
reflects the ``alignment'' between the choice of measured observable and the way the system can be manipulated, given some protocol constraints.
This means that, in the presence of constraints, the thermodynamic value of information 
depends not only on the {{amount}} of measured information, but also the {{content}} of that information~\citep{corning1998thermodynamics,kolchinsky2018semantic}. (See also \cite{kauffman2000investigations} for a popular discussion of some related issues.)

It is important to note that at this general level of analysis, we do not describe how to construct the  operator $\klopBase$, as this construction will typically depend on the structure of the set $\LLL$. However, as described in the following subsection, we do provide explicit expressions for $\klopBase$ for three broad classes of protocol  constraints, which we term  symmetry, modularity, and coarse-grained constraints.

\subsubsection*{Level 2: Symmetry, modularity, and coarse-grained constraints}

At the second level of our analysis, 
we apply the general framework described above to derive bounds on EP and work for three broad classes of protocol constraints:

\begin{itemize}[wide,labelindent=0pt,labelwidth=!]

    \item 
    \cref{sec:Symmetry-constraints} considers \emph{symmetry constraints},  when the available generators possess some symmetry group. %
    Examples of systems with symmetry constraints include the Szilard box in \cref{fig:szilard0}, spin systems on lattices, and 
    gases of indistinguishable particles. %
    The  operator $\klopBase$ corresponding to symmetry constraints, defined in \cref{eq:symmOop}, maps distributions to their ``symmetrized'' versions (which are invariant under the action of the symmetry group).

 \item 
    \cref{sec:modularity} considers \emph{modularity constraints},  when the available generators cause different (though possibly overlapping) subsystems of a multivariate system to evolve independently of each other.  Examples of systems with modularity constraints include digital circuits~\cite{wolpert2020thermodynamic}, ideal gases, and multi-particle Maxwellian demons. The  operator $\klopBase$ corresponding to modularity constraints, defined in \cref{eq:projMbasedef}, maps distributions to their ``uncorrelated'' versions, without statistical dependencies between independent subsystems.

 \item 
\cref{sec:cg} considers \emph{coarse-grained constraints}, when the available generators exhibit closed coarse-grained dynamics which obey some constraints (e.g., coarse-grained symmetry or modularity constraints). %
An example is provided by the Szilard box in \cref{fig:szilard0}:  the particle's vertical position (the coarse-grained macrostate) evolves in a way that does not depend on the horizontal position, and the macrostate equilibrium distribution cannot be controlled by moving the partition. Given a protocol that obeys coarse-grained constraints, we show that the EP can be lower bounded in terms of a ``coarse-grained EP'', \cref{eq:macroNaineq,eq:coarsegrainedEP,eq:cgEP2}, and that this coarse-grained EP can itself be lower bounded by a coarse-grained version of \cref{eq:epfreebound0}.

\end{itemize}

In addition, we also discuss how tighter bounds on work and EP can be derived by combining different kinds of constraints (e.g., when a system obeys  two different symmetry groups, or when it obeys both symmetry \emph{and} modularity constraints).

\subsubsection*{Level 3: Concrete examples}

At the third (and most concrete) level, we illustrate our results for symmetry, modularity, and coarse-grained constraints on  several example systems:

\begin{itemize}[wide,labelindent=0pt,labelwidth=!]

    \item 
    In \cref{subsec:Szilard-box-example}, we use symmetry constraints to derive thermodynamic bounds for the Szilard box in \cref{fig:szilard0}, which possesses vertical reflection symmetry.  

    \item In \cref{sec:ising}, we use symmetry constraints to derive thermodynamic bounds for the  Ising model on a 2D lattice, which possesses translational symmetry.

 \item In \cref{subsec:Szilard-box-example-mod}, we use modularity constraints to derive thermodynamic bounds for the Szilard box in \cref{fig:szilard0}, which are different from the bounds derived in \cref{subsec:Szilard-box-example}. We also demonstrate that stronger results can be derived by combining bounds arising from symmetry and modularity constraints.

 \item In \cref{sec:generalizedszilard,sec:ratchetexample}, we use modularity constraints to derive bounds on work extraction for two multi-particle feedback-control protocols that have been proposed in the literature: a multi-particle Szilard box~\cite{song2021optimal} and a collective flashing ratchet~\cite{cao2004feedback}.

 \item In \cref{subsec:Szilard-box-example-cg}, we use coarse-grained constraints to derive thermodynamic bounds for a version of the Szilard box in \cref{fig:szilard0} in the presence of gravity. We also demonstrate that stronger results can be derived by combining bounds arising from coarse-grained and modularity constraints.

\end{itemize}

\subsubsection*{Literature review and discussion}

After presenting the results summarized above, in \cref{sec:priorwork} we discuss related prior literature. We also compare and contrast our results, such as the decomposition of nonequilibrium free energy in \cref{eq:feffdecomp0}, to some relevant work in quantum thermodynamics~\cite{janzing_quantum_2006,vaccaro_tradeoff_2008}. 
We conclude with a brief discussion in \cref{sec:conclusion}, which also touches upon how our approach generalizes beyond the assumption of a single heat bath. Proofs and derivations are in the appendices.

\section{Preliminaries}

\label{sec:physicalsetup}
\newcommand{\dims}{n}
We consider a physical system with state space $\sX$, which can be either discrete or continuous ($X = \mathbb{R}^{\dims}$). The term ``probability distribution''
 will refer to a probability mass function over $X$ in the discrete case and to a probability density function over $X$ in the continuous case. We  interchangeably use notation like $p(x)$ and $p_x$ (as will be clear from context) to indicate the probability of state $x$. 
We use $\Dset$ to refer to the set of all probability distributions over $X$. 

The system 
 evolves in a stochastic manner during a driving protocol over time $t\in[0,\ft]$. 
We will write $p(t)$ to indicate the distribution at time $t$  corresponding to some initial distribution $p(0)=p$, and $p(\ft)=\pf$ to indicate the distribution at the end of the protocol. 
For a discrete-state system, %
the distribution at time $t$ evolves according to the time-dependent master equation,
\begin{align}
\ppt p_x(t)=\sum_{x'}\left[\Lji(t)p_{x'}(t)-\Lij(t)p_x(t)\right],
\label{eq:ll}
\end{align}
where $\Lij(t)$ is the transition rate from state $x$ to state $x'$.  
We assume that the system is coupled to a heat bath at inverse temperature $\beta$, and so each $L(t)$ obeys local detailed balance (see \cref{sec:conclusion} for a generalization of this assumption).  
Formally, this means that
$\pi^{L(t)}_{x'}\Lji(t) = \pi^{L(t)}_{x} \Lij(t)$ for all $x$,$x'$, and $t$, where $\pi^{L(t)}$ is the stationary distribution of rate matrix $L(t)$, which we assume is unique (though this latter assumption can be relaxed~\footnote{The assumption of unique stationary distributions can be relaxed as long as the operator \unexpanded{$\klopBase$} (as discussed in \cref{sec:info-geom-framework}) satisfies the following weak technical condition: for all \unexpanded{$p\in\Dset$} and each stationary distribution \unexpanded{$\pi$} of each \unexpanded{$\LL\in \LLL$}, \unexpanded{$D(p\Vert \klop{\pi})<\infty$} whenever \unexpanded{$D(p\Vert \pi)<\infty$}. Note that \unexpanded{$\klop{\pi}$} is also a stationary distribution of \unexpanded{$\LL$} by \cref{lem:stlem} in \cref{app:theoretical}, so this condition is automatically satisfied when the generators have unique stationary distributions (since in that case \unexpanded{$\pi=\klop{\pi}$}).  
Note also that if some \unexpanded{$\LL\in\LLL$} have multiple stationary distributions \unexpanded{$\pi$}, the corresponding EP rate in \cref{eq:naBound} can be equivalently defined using any \unexpanded{$\pi$} such that \unexpanded{$D(p\Vert \pi)<\infty$}.}).

The rate of entropy production (\emph{EP rate}) incurred at time $t$ can be written as (Eq.~33 in \cite{esposito2010three})%
\begin{align}
\EPr(p(t),L(t))=-\sum_{x} \ppt p_x(t) \ln \frac{ p_x(t)}{ \pi^{L(t)}_x}
\ge0,%
\label{eq:naBound}
\end{align}
where $\ppt p_x(t)$ is defined in \cref{eq:ll}. 
Note that the right side of \cref{eq:naBound} is sometimes called the ``nonadiabatic EP rate'' in stochastic thermodynamics, and it is equal to the overall EP rate for a system coupled to a single bath and obeying detailed balance~\cite{esposito2010three}. 
The total EP incurred  by a time-extended protocol over $t\in[0,\ft]$ that carries out the transformation $\ptpp$ %
is given by the integral of the EP rate,
\begin{align}
\EPp=\int_0^1 \EPr(p(t),L(t)) \,dt.
\label{eq:totalEP}
\end{align} 
The work extracted during a protocol can be calculated by using \cref{eq:totalEP,eq:EPw}, 
once the initial and final nonequilibrium free energies,  $\NFEs$ and $\NFEf$, are specified. 
To define these free energies, we assume that there is some fixed pair of energy functions,  $E$ and $E'$, %
which specify the Boltzmann equilibrium distributions of $L(0)$ and $L(\ft)$ respectively. %

For a continuous-state system evolving under a continuous master equation~\citep{van1992stochastic,risken1996fokker},  the sums in \cref{eq:ll,eq:naBound} should be replaced by integrals (see Eq.~31 in \cite{van2010three}). 
A prototypical example of a continuous master equation, which we will use below, is a  
Fokker-Planck equation~\citep{ermak1978brownian,van1992stochastic}, %
\begin{align}
\ppt p(x,t) =-\nabla\cdot(\uL(x,t) p(x,t) -\DL(x,t) \nabla p(x,t)),\label{eq:fp-4}
\end{align}
where $\uL $ and $\DL $ are drift and diffusion terms. 

We will often write dynamical equations like  \cref{eq:ll,eq:fp-4} using the notation $\ppt p(t) =\LL(t) p(t)$, where $\LL(t)$ is a bounded linear operator that is called the \emph{(infinitesimal) generator} of the dynamics at time $t$. 
Note that for a continuous-state system in phase space, it may be that the system is isolated from the bath for some $t\in[0,\ft]$, in which case $\ppt p(t) =\LL(t) p(t)$ should be understood in terms of the Liouville equation. (For example, if a system is first isolated and evolves in a Hamiltonian manner, and is then brought in contact with a bath at inverse temperature $\beta$ and allowed to equilibrate).%

\section{General  framework}
\label{sec:info-geom-framework}

We begin by presenting our general mathematical framework. The application of this framework to concrete situations is described in latter sections.

A driving protocol $\{L(t):t\in[0,\ft]\}$ is said to be \emph{constrained} if there is %
some restricted set of generators $\LLL$ such that $\LL(t)\in\LLL$ at all $t$.  For a given  set of allowed generators $\LLL$, %
we consider an associated 
operator  $\klopBase : \Dset \to \Dset$ which  %
satisfies two conditions. The first condition states that%
\begin{align}
D(p\Vert q)=D(p\Vert\klop p)+D(\klop p\Vert q) %
\label{eq:pyth}
\end{align} 
for all $p\in\Dset$ and $q\in\imgKLOP$ with $D(p\Vert q)<\infty$ (where $\imgKLOP=\{\klop p : p\in\Dset\}$ is the image of the operator $\klopBase$). %
\cref{eq:pyth} %
is sometimes called the \emph{Pythagorean identity of KL divergence} in information geometry~\cite{amari2016information}. Any $\klopBase$ that obeys \cref{eq:pyth}  can be written in terms of the following projection~\footnote{This is because 
 \unexpanded{$D(p\Vert q)\ge D(p\Vert \klop p)$} for any \unexpanded{$q\in \imgKLOP$}, which follows from \cref{eq:pyth} and the non-negativity of KL divergence.}
\begin{align}
\klop p = \argmin_{q \in \imgKLOP} D(p\Vert q),
\label{eq:klprojection}
\end{align}
which shows that $D(p\Vert\klop p)$ is the minimal information-theoretic distance from $p$ to the set of distributions $\imgKLOP$.

The second condition is that $\klopBase$ obeys the following \emph{commutativity relation} for all $\LL\in\LLL$:
\begin{align}
e^{\inftime L}\klop p = \klop {e^{\inftime L} p} \quad \forall \inftime \ge0,p\in\Dset. %
\label{eq:comm0}
\end{align}
In other words, given any initial distribution $p$, the same final distribution is reached regardless of whether $p$ first relaxes under $L$ for time $\inftime$ and then undergoes $\klopBase$, or instead
first undergoes $\klopBase$ and then relaxes under $L$ for time $\inftime$. 

Note that the Pythagorean identity in \cref{eq:pyth} concerns only the operator $\klopBase$, while the commutativity relation in \cref{eq:comm0} concerns the relationship between $\klopBase$ and the generators in $\LLL$ (and therefore all of the generators $L(t)$ in the driving protocol, since $L(t)\in\LLL$ at all $t$ by assumption).  
Beyond these two conditions, the operator $\klopBase$ can be arbitrary, and may be linear or nonlinear. 
In the following sections of this paper, will show how to choose $\klopBase$ for various types of constrained protocols.

Importantly, any $\klopBase$ that satisfies the two conditions above maps any distribution $p$ to a corresponding ``accessible'' distribution $\klop p$, which controls the amount of work that can be extracted from $p$ by a constrained driving protocol. 
To prove this, we first show that for any $\LL\in\LLL$ that obeys \cref{eq:comm0}, the equilibrium distribution $\pi^\LL$ satisfies (\cref{lem:stlem} in \cref{app:theoretical})
 \begin{align}
 \pi^\LL \in \imgKLOP.\label{eq:equilinimgKLOP}
 \end{align}
We also derive the following mathematical result, will be central to much of what follows: if $\klopBase$ obeys \cref{eq:pyth} and \cref{eq:comm0} for some generator $\LL$, %
then the EP rate incurred by any distribution $p$ under $\LL$ can be written as the sum of two non-negative terms: the EP rate  incurred by $\klop{p}$ under $\LL$, and the instantaneous contraction of the KL divergence between $p$ and $\klop{p}$. 

\begin{figure}
\begin{centering}
\includegraphics[width=1\columnwidth]{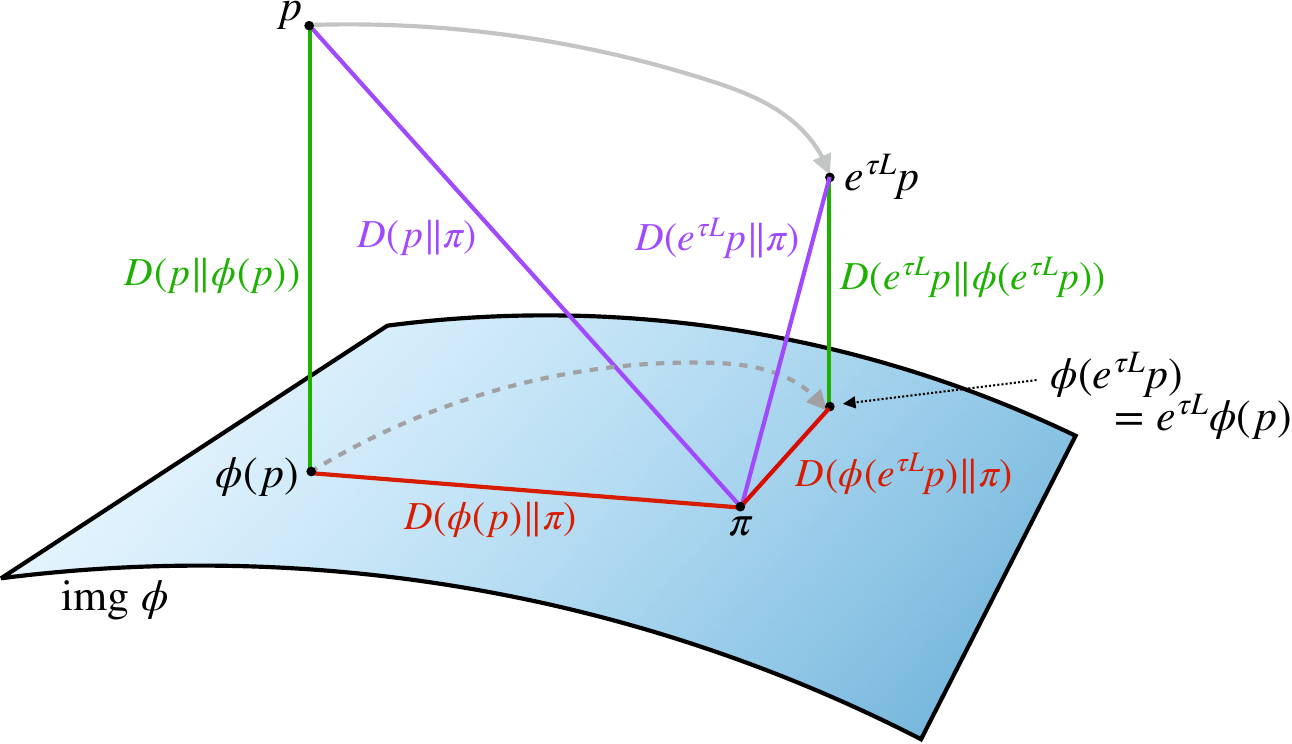}
\par\end{centering}
\caption{Visual explanation of \cref{thm:klinst}: distribution $p$ freely relaxes under $\LL$ for time $\inftime$ (solid gray line). 
The EP incurred during this relaxation (contraction of purple lines) can be decomposed into the  contraction of the KL divergence between $p$ and $\klop p$ (contraction of green lines), plus the EP incurred during the free relaxation of $\klop p$ (contraction of the red lines). The free relaxation of $\klop p$ under $\LL$ is represented by the dotted gray line. %
\label{fig:thm1}}
\end{figure}

\begin{thm}
\label{thm:klinst}
If $\klopBase$ obeys \cref{eq:pyth} and \cref{eq:comm0} for some generator $L$, %
then  for all $p\in \Dset$,
\[
\EPr(p,L) = \EPr(\klop {p},L) -\ddt D(p(t)\Vert\klop{p(t)}),
\]
and $-\ddt D(p(t)\Vert\klop{p(t)})\ge0$, where $\ppt p(t) = \LL p$.
\end{thm}

We sketch the proof of this  theorem %
in terms of a discrete-time relaxation over interval $\inftime$, as shown in \cref{fig:thm1} (see \cref{app:theoretical} for details). Consider some distribution $p$ that relaxes for time $\inftime$ under the generator $\LL$, %
thereby reaching the distribution $e^{\inftime \LL}p$ (solid gray line). The EP incurred by this relaxation is given by the contraction of KL divergence to the equilibrium distribution $\pi$, $\EP(\dTrans{p}{ e^{\inftime \LL} p}) = D(p\Vert \pi) - D(e^{\inftime \LL} p\Vert \pi)$ %
(contraction of purple lines)~\cite{esposito2010three,van2010three}. Given \cref{eq:equilinimgKLOP}, we can apply the Pythagorean identity, \cref{eq:pyth}, to both $D(p\Vert \pi) $ and $D(e^{\inftime \LL} p\Vert \pi)$, which lets us rewrite $\EP(\dTrans{p}{ e^{\inftime \LL} p})$ as the sum of two terms:   $D(p\Vert \klop p)-D(e^{\inftime \LL} p\Vert \klopBase(e^{\inftime \LL} p)$  (green lines) and  $D(\klop p\Vert \pi) - D(\klopBase(e^{\inftime \LL} p)\Vert \pi)$ (red lines).  Applying the commutativity relation, \cref{eq:comm0}, shows that the first term is non-negative by the data-processing inequality and that the second term is equal to $\EP(\dTrans{\klop p}{e^{\inftime \LL} \klop p})$, the EP incurred by letting $\klop p$ relax freely under $\LL$. 
The continuous-time statement found in \cref{thm:klinst} follows by taking the appropriate $\inftime \to 0$ limit, while noting that the EP rate,  \cref{eq:naBound}, can be rewritten in terms of the limit $\lim_{\inftime \to 0} \frac{1}{\inftime} [D(p\Vert \pi) - D(e^{\inftime \LL} p\Vert \pi)]$.

Now suppose that \cref{eq:comm0} holds, so that the assumptions of \cref{thm:klinst} are satisfied during the entire protocol. In that case, as we show in \cref{lem:commProp} in \cref{app:theoretical}, any constrained protocol that carries out the transformation $\ptpp$ must also 
transform the  initial distribution $\klop p$ to the final distribution $\klop \pf$. %
We can then, in essence,
integrate  \cref{thm:klinst} over time and derive  %
the following result about total EP. 
\begin{thm}
\label{thm:resInt}
If $\klopBase$ obeys \cref{eq:pyth} and \cref{eq:comm0} for all $\LL \in \LLL$, then for any constrained  protocol that transforms $\ptpp$, %
\begin{align*} 
&\EPp = \EP(\klopptpp) +\left[ D( \ps \Vert\klop \ps )-D(\pf \Vert\klop{\pf})\right]
\end{align*}
and $D(p\Vert\klop \ps)-D(p'\Vert\klop{\pf})\ge 0$.
\end{thm}

\begin{figure}[b]
\begin{centering}
\includegraphics[width=0.8\columnwidth]{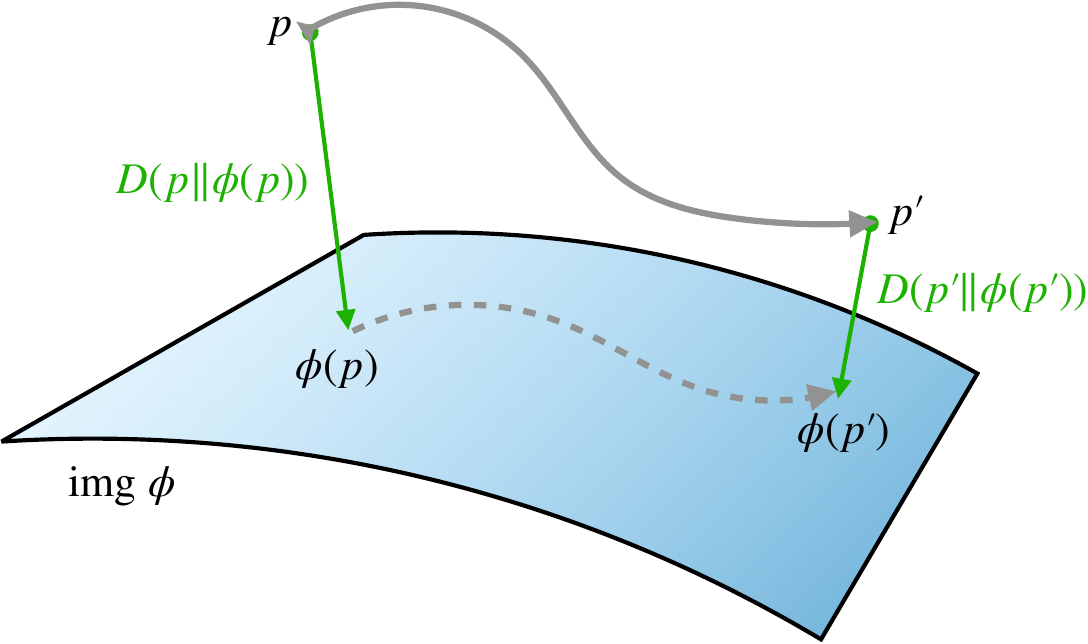}
\par\end{centering}
\caption{
Illustration of \cref{thm:resInt}. Given an appropriate operator $\klopBase$, 
$\EP(\ptpp)$  (the EP incurred during some desired transformation $\ptpp$; solid gray line) 
is equal to 
$\EP(\klopptpp)$ (the EP incurred by that protocol when transforming $\klopptpp$; dashed gray line) plus the contraction of the KL divergence $D(p\Vert \klop p)-D(p'\Vert \klop{p'})$ (contraction of green lines). This contraction of KL divergence is a non-negative lower bound on $\EP(\ptpp)$, as in \cref{eq:epfreebound}. %
 \label{fig:introres}}
\end{figure}
We use \cref{thm:resInt} to derive several useful bounds on EP and work. First, since $\EP(\klopptpp)\ge 0$ by the non-negativity of EP, the contraction of KL divergence between $p$ and $\klop p$ bounds the EP incurred by a constrained driving protocol that carries out the transformation $\ptpp$,
\begin{align}
\EPp \ge D( \ps \Vert\klop \ps )-D(\pf \Vert\klop{\pf}) \ge 0,
\label{eq:epfreebound}
\end{align}
which appeared as \cref{eq:epfreebound0} in the introduction. Furthermore, $D(p\Vert\klop \ps)-D(p'\Vert\klop{\pf})\ge 0$ immediately implies that 
\begin{align}
\EPp \ge \EP(\klopptpp).
\label{eq:bnd3}
\end{align}

We can also derive the decomposition of free energy and the bound on extractable work, which appeared as \cref{eq:feffdecomp0,eq:ourworkbound0} in the introduction. Consider some transformation $\ptpp$, and write the initial nonequilibrium free energy as %
\begin{equation}
\NFEs =\NFEs[\pi]+D(p\Vert\pi)/\beta,\label{eq:Fkl}
\end{equation}
where $\pi \propto\boltz{-\beta E}$ is the Boltzmann distribution for the initial energy function $E$, and  $\NFEs[\pi]$ is  the equilibrium free energy~\citep{esposito2011second}.  
Using \cref{eq:equilinimgKLOP} and  the  %
Pythagorean identity, \cref{eq:pyth}, we %
decompose the nonequilibrium free energy into a sum of the accessible free energy and the inaccessible free energy,
\begin{align}
\NFEs & =\NFEs[\pi]+[D(p\Vert\klop p)+D(\klop p\Vert\pi)]/\beta \nonumber \\
 & =\NFEs[\klop p] + D(p\Vert\klop p)/\beta. \label{eq:feffdecomp}
\end{align}
Using a similar derivation, %
we can write the nonequilibrium free energy at the end of the protocol as
\begin{align}
\NFEf =\NFEf[\klop \pf]+D(\pf\Vert\klop \pf)/\beta. \label{eq:fkl3}
\end{align}
Subtracting  \cref{eq:fkl3} from \cref{eq:feffdecomp} shows that the drop in the nonequilibrium free energy during $\ptpp$ is given by
\begin{multline}
\NFEs-\NFEf =   \NFEs[\klop \ps]-\NFEf[\klop \pf]  + \\ \left[ D( \ps \Vert\klop \ps )-D(\pf \Vert\klop{\pf})\right]/\beta .
 \label{eq:fedrop}
\end{multline}
Combining this result with \cref{thm:resInt} and \cref{eq:EPw}, and then rearranging,  shows that the work involved in carrying out  $\ptpp$ is equal to  the  work involved in carrying out the accessible transformation $\klopptpp$:
\begin{align}
W(\ptpp)=W(\klopptpp).
\label{eq:workeq}
\end{align}
Finally, by combining with \cref{eq:workbound}, we arrive at an upper bound on work that can be extracted by a constrained protocol:
\begin{align}
W(\ptpp)\le \NFEs[\klop \ps]-\NFEf[\klop \pf] ,
\label{eq:ourworkbound}
\end{align}
which is tighter than the bound given by the second law, \cref{eq:workbound}.

The bounds in \cref{eq:epfreebound,eq:ourworkbound}, as well as the decomposition of free energy in \cref{eq:feffdecomp}, are the main theoretical results arising from our general framework. 
\cref{fig:introres} provides a schematic way of understanding these results. 
\cref{thm:resInt} states that, for a constrained protocol that carries out the map $\ptpp$, 
the EP incurred during the system's actual trajectory (solid gray line) is given by the EP that would incurred by a ``projected trajectory'' that carries out the transformation $\klopptpp$ (dashed gray line), plus the drop in the KL divergence from the system's  distribution to the set $\imgKLOP$ over the course of the protocol (contraction of green lines).  
Since the EP of the projected trajectory must be non-negative, the drop in the distance from the system's  distribution to $\imgKLOP$ 
serves as a lower bound on EP, as in \cref{eq:epfreebound}.  In addition, \cref{thm:resInt} states that this decrease in the KL divergence must be positive, meaning that the system's distribution must get closer to $\imgKLOP$ over the course of the protocol.  

Following \cref{fig:introres}, it can be helpful to think of the trajectory $\ptpp$ as composed of three segment: (1) from $p$ down to $\klop p$, (2) from $\klop p$ to $\klop {\pf}$ while staying within $\imgKLOP$, and (3) from $\klop \pf$ up to $\pf$ (note that this decomposition is useful for accounting purposes, but does not generally reflect the actual trajectory the system takes in going from $p$ to $p'$).  
The first and third segments contribute (positively and negatively, respectively) only to EP, while the projected second segment $\klopptpp$ contributes both to EP and to work. 
 Thus, the %
work involved in $\ptpp$ is determined entirely by the work involved in the second segment, as stated in \cref{eq:workeq}. 

Note also the formal similarity between our decomposition of the drop in free energy, \cref{eq:fedrop}, 
and the decompositions of EP in \cref{thm:resInt}. 
Indeed, like \cref{thm:resInt}, the result \cref{eq:fedrop}  can be illustrated with \cref{fig:introres}: 
during the transformation $\ptpp$ (solid gray line), the drop in free energy is given by the drop in free energy incurred by the transformation $\klopptpp$ (dotted gray line), plus the contraction of the KL divergence from the system's distribution to the set $\imgKLOP$ (green lines).

In general, our bounds on EP and work will not always be achievable. Suppose, however, that the final distribution $p'$ is in equilibrium, so $\pf= \klop{\pf}$ by \cref{eq:equilinimgKLOP}. %
\cref{eq:epfreebound} then gives
\begin{align}
\EPp \ge D(p\Vert\klop p).
\label{eq:simplebound1}
\end{align}
This bound is achievable if the generators in $\LLL$ have a continuous curve of equilibrium distributions from $\klop{p}$ to $\pf=\klop{\pf}$.  %
Imagine a protocol in which %
the initial distribution $p$ first relaxes to the equilibrium distribution $\klop p$, and then undergoes quasistatic driving from $\klop p$ to $\klop \pf$ while remaining in equilibrium throughout 
(in terms of \cref{fig:introres}, the system first relaxes along the green arrow connecting $p$ to $\klop p$, then follows the dashed line to $\klop{p'}$ quasistatically).
The relaxation step  incurs $D(p\Vert\klop p)$ of EP, while the quasistatic step incurs a vanishing amount of  EP, so the bound in \cref{eq:simplebound1} will be achieved. %

\subsection{Choice of the $\klopBase$ operator}
\label{sec:choiceofklop}

In general, the operator $\klopBase$ associated with a given set of generators $\LLL$ is not unique. 
For instance, for any driving protocol, the %
identity map $\klop p = p$ always satisfies \cref{eq:pyth} and \cref{eq:comm0}. Choosing $\klopBase$ to be the identity map, however, reduces  the results in \cref{thm:resInt} to  trivial identities and 
the lower bound on EP in \cref{eq:epfreebound} to $0$. 

At a high level, those $\klopBase$ which have smaller $\imgKLOP$ will generally give tighter bounds on EP %
(since, given \cref{eq:klprojection}, a smaller image leads to larger values of $D(p\Vert \klop p)$). 
To illustrate this phenomenon, 
consider the extreme case where all $\LL \in \LLL$ have the same equilibrium distribution $\pi$, so that any constrained driving protocol must be a free relaxation toward $\pi$. 
Then, the operator $\klop p = \pi$ for all $p$ (so $\imgKLOP$ is a singleton)  
satisfies \cref{eq:comm0,eq:pyth} and, when plugged into
\cref{eq:epfreebound}, gives the following bound on EP:
\begin{align}
\EPp \ge D(\ps \Vert \pi) -D(\pf \Vert \pi).
\label{eq:e2}
\end{align}
In fact, the right hand side is an exact expression for the EP incurred by the free relaxation, %
meaning that 
it is the tightest possible bound. 
If, however, the generators $L\in\LLL$ have  different equilibrium distributions, %
then the operator $\klop p = \pi$ (for whatever $\pi$) generally violates %
the commutativity relation in \cref{eq:comm0}, and bounds like \cref{eq:e2} will no longer hold.

In the following sections, we show how to use our results to derive thermodynamic bounds for $\LLL$ that obey some kind of symmetry group, modular decomposition, or coarse-grained structure. %
In more general, possibly unstructured cases, %
it is an open question of whether a non-trivial operator $\klopBase$ exists, and if so how to identify it. 
We explore related issues in a companion paper~\cite{kolchinsky_constraints_paper2}, where we use numerical optimization techniques to derive bounds on EP similar to \cref{eq:epfreebound}.

Importantly, when there are  multiple different operators that  all satisfy the Pythagorean identity and the commutativity relation for the available generators $\LLL$, one can derive  tighter bounds on EP and work by applying our decompositions in an ``iterative'' manner. For instance, imagine that there are two different operators $\klopBase_1$ and $\klopBase_2$ that satisfy \cref{eq:pyth,eq:comm0} (for example, these might represent operators arising from symmetry constraints and modularity constraints, respectively, as described below). Applying \cref{thm:resInt} iteratively leads to ``stacked'' bounds on EP analogous \cref{eq:epfreebound},
\begin{multline}
\EP(\ptpp)  \ge \big[D(\ps \Vert \phi_1(\ps))+D(\phi_1(\ps) \Vert \phi_2(\phi_1(\ps))) \big] -\\
\big[D(\pf \Vert \phi_1(\pf)) + D(\phi_1(\pf) \Vert \phi_2(\phi_1(\pf)))\big]\ge 0.\label{eq:stackedbounds}
\end{multline}
Similarly, applying \cref{eq:workeq} iteratively leads to stacked bounds on extractable work analogous to \cref{eq:ourworkbound},
\begin{align}
W(\ptpp)\le \NFEs[\phi_2(\phi_1 (\ps))]-\NFEf[\phi_2(\phi_1( \pf ))] .
\label{eq:ourworkboundstacked}
\end{align}
Such stacked bounds  are generally tighter than the bounds provided by either $\phi_1$ or $\phi_2$ alone. %
(Note that one can also reverse the order of operations,  and consider the composition  $\phi_1(\phi_2( \ps ))$ rather than $\phi_2(\phi_1( \ps ))$ in \cref{eq:stackedbounds,eq:ourworkboundstacked}, which will in general lead to different bounds.)

\subsection{Fluctuating entropy production}

As we show in detail in \cref{app:fluct}, 
our results also have implications for stochastic fluctuations of trajectory-level EP, as considered in stochastic thermodynamics~\cite{seifert2012stochastic}.  

Consider any constrained driving protocol over $t\in[0,\ft]$ with an associated operator $\klopBase$. 
Let $\traj$ indicate some stochastically sampled trajectory of the system visited during the driving protocol, 
and let 
$\sigma_{p}(\traj)$ indicate the fluctuating  EP incurred by trajectory $\traj$ when initial states are sampled from the initial distribution $p$.  In the appendix, we consider the difference between this fluctuating EP and the fluctuating EP incurred by the same trajectory when initial states are sampled from the accessible initial distribution $\klop p$,
\begin{align}
m_{p}(\traj):=\sigma_p(\traj)-\sigma_{\klop p}(\traj).
\label{eq:fluctmismatch}
\end{align}

By combining \cref{thm:resInt} with recent results in stochastic thermodynamics~\cite{kolchinsky2021state,kwon2019fluctuation}, we show that the expectation of $m_{p}(\traj)$ is equal to the difference of expected EPs, $\langle m_{p}(\traj)\rangle = \EP(\ptpp)-\EP(\klopptpp)$, where $\langle \cdot \rangle$ indicates expectation over trajectories sampled from initial distribution $p$. We also show that $m_{p}(\traj)$ obeys a detailed fluctuation theorem, which  
implies a trajectory-level version of \cref{eq:bnd3}: the probability that the fluctuating EP  under initial distribution $p$ is $\xi$ less than the fluctuation EP under the accessible initial distribution $\klop p$  %
is exponentially small (i.e., it is less than $e^{-\xi}$). We leave further exploration of the connection between our framework and stochastic thermodynamics for future work.

\section{Thermodynamics of information under protocol constraints}
\label{sec:thermo-of-info}

\newcommand{\Lm}{\LL^{(m)}}
The framework introduced in the previous section has implications for the thermodynamics of information under constraints. Consider the type of %
feedback control setup described in the introduction: first an observation apparatus $M$ measures some system observable, then the system undergoes a driving protocol that depends on the measurement outcome $m$. Let $\Lm(t)$ indicate the driving protocol conditioned on  $m$, and $\pconds$ and $\pcondf$ indicate the distributions over system states at the beginning and end of the corresponding driving protocol. 
As standard in the literature~\cite{parrondo2015thermodynamics}, for simplicity we  assume that all protocols start and end with the same energy functions, $E$ and $E'$, and that during the protocols, the measurement apparatus $M$ and the system $X$ are energetically decoupled and that $M$ does not change state.

Given the above assumptions, it is straightforward to show that the EP incurred by the joint ``supersystem'' $X\times M$ obeys 
\begin{align}
\EP_{XM} = \sum_m p(m) \EP_m,
\label{eq:fcepr}
\end{align}
where 
$\EP_m$ is the EP incurred by protocol $\Lm(t)$ in carrying out the transformation $\dTrans{\pconds}{\pcondf}$. 
Similarly, by taking expectations of  \cref{eq:EPw} and rearranging (see derivation of \cref{eq:workbound-2}), the average extracted work under feedback control can be written as 
\begin{align}
\langle W \rangle =\Delta F+[\IXMs\!-\!\IXMf]-\!\sum_m  p(m) \EP_m,
\label{eq:fcwr}
\end{align}
where for notational convenience we've used $\Delta F=\NFEs - \NFEf$ to indicate the drop of marginal free energy.  Thus, any lower bounds on $\EP_m$ (the EP values incurred by the individual protocols $\Lm(t)$) can be translated into bounds on the overall EP and average extractable work for a feedback control setup.

For example, suppose that there is some single set of constraints that applies to all of the driving protocols, in that 
there is some set of  generators $\LLL$ such that $\Lm(t)\in\LLL$ for all $t$ and $m$, as well as an operator $\klopBase$ that obeys the Pythagorean identity, \cref{eq:pyth}, and the commutativity relation, \cref{eq:comm0}, for all $\LL\in\LLL$. In that case, the framework described in \cref{sec:info-geom-framework} leads to bounds on each $\EP_m$ term. In particular, using  \cref{eq:epfreebound,eq:fcepr} %
 gives the bound
\begin{multline}
\EP_{XM} \ge \\\IinaccS - \IinaccF \ge0,\label{eq:EPmeas}
\end{multline}
where we've defined the conditional KL divergence $\IinaccS=\sum_m p(m) D(\pconds \Vert\klop{\pconds})$, and similarly for $D(\PCONDF \Vert\effCPfM)$. 
Plugging into \cref{eq:fcwr} gives the 
following bound on average extractable work: %
\begin{align}
\left\langle W\right\rangle  \le \Delta F+[\IaccS-\IaccF]/\beta,\label{eq:workboundinfo}
\end{align}
where $\IaccS$ is given by
\begin{align}
\IaccS&=\IXMs - \IinaccS,
\label{eq:accinfodef}
\end{align}
and similarly for $\IaccF$. 

We refer to $\IaccS$ as the \emph{accessible information} in measurement $M$, since any decrease in accessible information can contribute to work extraction (\cref{eq:workboundinfo}).
We refer to the 
conditional KL divergence $\IinaccS$ %
as the \emph{inaccessible information}, since any decrease in inaccessible information must be dissipated as EP,  and not extracted as work (\cref{eq:EPmeas}). 
The inaccessible information is non-negative by properties of KL divergence, so $\IaccS\le \IXMs$. 
In addition, whenever $p\in\imgKLOP$ (e.g., when $p$ is an equilibrium distribution, by \cref{eq:equilinimgKLOP}), the accessible information can be rewritten in simpler form as%
\begin{align}
\IaccS = D(\klop{\PCONDS}\Vert p),
\label{eq:accinfodef2}
\end{align}
as follows from \cref{eq:accinfodef} by writing $\IXMs=D(\PCONDS\Vert p)$ and applying the Pythagorean theorem, \cref{eq:pyth}.

In general, measurements of different observables on the same system will give rise to different amounts of accessible and inaccessible information. At a high level, one should choose measurements that maximize the accessible information $\IaccS$, or alternatively the ``efficiency'' quantified as bits of accessible information per bit of measured information, $\IaccS/\IXMs \le 1$.  Optimal measurements satisfy  $\IaccS=\IXMs$, which happens when the conditional distributions over system states $\pconds$ are invariant under the action of $\klopBase$ (i.e., when $\klop{\pconds}=\pconds$ for each $m$).

Note that similar results can also be derived using other kinds of bounds on  $\EP_m$ (e.g., when the individual protocols obey a combination of constraints, so that \cref{eq:stackedbounds} holds). %

\section{Symmetry constraints}

\label{sec:Symmetry-constraints}

\global\long\def\Vp{V_{\mathrm{p}}}%
\global\long\def\Vw{V_{\mathrm{w}}}%
\global\long\def\G{\mathcal{G}}%
\global\long\def\g{g}%
\newcommandx\actg[1][usedefault, addprefix=\global, 1=g]{{#1}}%

\global\long\def\projGbase{\klopBase_{\G}}%
\global\long\def\projG#1{\projGbase(#1)}%
\newcommand{\actgInv}{\actg[g^{-1}]}

\newcommandx\transP[1][usedefault, addprefix=\global, 1=g]{\Phi_{#1}}%

We now use the general framework introduced above to derive bounds on EP under symmetry constraints. 
Consider a compact group $\G$ that has a measurable action over $\sX$, such that each $\actg\in \G$ is a bijection $\sX\to \sX$~\footnote{A compact group $\G$ has a measurable action over $X$ if the action $\G \times X \to X$ is a measurable function, where we assume $\G$ and $X$ are endowed with their respective Borel algebras.}. 
For continuous $X$, we assume that each  $g\in\G$ is a rigid transformation.  For notational convenience,
for each $g\in\G$ we define the composition operator $\transP$, so that for any function $f:\sX \to\mathbb{R}$,
\begin{align}
\label{eq:transPdef}
\transP(f)(x)=f(g(x)).
\end{align}  

We say that a set of generators $\LLL$ obeys \emph{symmetry constraints} (with respect to the action of group $\G$) if the following commutativity relation holds for all $L\in\LLL$: %
\begin{align}
\transP L = L \transP.\qquad\forall \actg\in\G.
\label{eq:commsymmZ}
\end{align}
In other words, $\LLL$ obey symmetry constraints  when, for each $\LL\in\LLL$ and $\actg \in \G$, it does not matter whether one first applies the generator $\LL$ and then the bijection $\actg$ over the state space, or first applies the bijection $\actg$ over the state space and then the generator $\LL$. In more concrete terms, 
for a (continuous or discrete) master equation $\LL$, 
\cref{eq:commsymmZ}  holds if the transition rates are invariant under the action of $\G$:
\begin{align}
\Lji=\sTrans{\actg(x')}{\actg(x)} \qquad \forall x,x'\in \sX, g\in\G.
\label{eq:symmME}
\end{align}

We can also derive simple sufficient conditions %
for potential-driven  Fokker-Planck equations of the type
\begin{align}
Lp=\nabla\cdot(\nabla E_{\LL}) p+\dL\Delta p,\label{eq:fp10}
\end{align}
where $E_{\LL}$ is the energy function of generator $\LL$.
Then,  \cref{eq:commsymmZ} %
holds if all available  
energy functions are invariant under the action of $\G$, 
\begin{align}
E_\LL(x)=E_\LL(\actg(x))\quad\forall x\in X,g\in\G,\LL\in\LLL\,.
\label{eq:symmFP}
\end{align}
(\cref{eq:commsymmZ} is derived from \cref{eq:symmME,eq:symmFP} in  \cref{app:symm}).

We now define a linear operator $\projGbase$ which satisfies the Pythagorean identity and the commutativity relation, \cref{eq:pyth,eq:comm0}, for symmetry constraints. 
Let $\projGbase$ map each $p \in\Dset$
to its average under the action of $\G$,
\begin{align}
\projG p(x) :=\int_{\G} p(g(x)) \,d\mu(\g),
\label{eq:symmOop}
\end{align}
where $\mu$ is the uniform (normalized Haar) measure over $\G$~\footnote{
Technically, the definition of the twirling operator in \cref{eq:symmOop} applies only when $p$ is a finite-valued probability density function (which excludes things such as the Dirac delta ``function''). A more general formulation of our results can be developed in terms of probability measures rather than probability densities (see Ch.~3 in \cite{eaton_group_1989} for a version of \cref{eq:symmOop} defined in terms of probability measures).}. %
For a finite group, the integral in \cref{eq:symmOop} should be replaced by a summation. 
Following the terminology in quantum physics, we sometimes refer to 
$\projGbase$ as a \emph{twirling operator}~\cite{vollbrecht2001entanglement,vaccaro_tradeoff_2008}.
Intuitively, $\projG p$ symmetrizes $p$, removing all information in $p$ concerning the state of the system
along the ``coordinates'' specified by the symmetry constraints.

In \cref{app:symm}, we show that $\projGbase$ %
obeys the Pythagorean identity and, as long as \cref{eq:commsymmZ} holds, the commutativity relation of \cref{eq:comm0}.
Thus, 
any protocol that carries out the transformation $\ptpp$ while obeying symmetry constraints with respect to  $\G$ %
permits the decomposition of EP found in \cref{thm:resInt}, with $\klopBase = \projGbase$, and satisfies all the bounds on work and EP that follow from that result. %

In particular, using  \cref{eq:feffdecomp}, we can decompose the free energy $\NFEs$ of any distribution $p$ into the accessible free energy  $\NFEs[\projG p]$, which is  the free energy in the twirled (and therefore symmetric) version of $p$, and the  inaccessible free energy $D(p\Vert\projG p)/\beta$.
Note that $D(p\Vert\projG p)$ is a non-negative
measure of the asymmetry in distribution $p$ with respect to the symmetry group $\G$,
which vanishes when $p$ is invariant under $\projGbase$.
Thus, for any protocol that obeys symmetry constraints, the first inequality in \cref{eq:epfreebound} states that any ``drop in asymmetry'' must be dissipated as EP, and not turned into work. %
The second inequality in \cref{eq:epfreebound} states that the asymmetry in the system's distribution can only decrease during  the  protocol. %
(Some of the above results for symmetry constraints have been previously uncovered in quantum thermodynamics~\cite{vaccaro_tradeoff_2008,janzing_quantum_2006}; see \cref{sec:priorwork}.)

We finish by discussing thermodynamics of information under symmetry constraints. In general, the results  derived in \cref{sec:thermo-of-info} apply to the twirling operator $\projGbase$ as a special case.  
We can also exploit special properties of $\projGbase$ to further simplify the expression of the inaccessible information term in \cref{eq:accinfodef,eq:EPmeas}. %
 Suppose that 
distribution $\ps$ is invariant under $\projGbase$, so  $\ps=\projG{\ps}$ (e.g., if $p$ is an equilibrium distribution). 
As shown in \cref{app:symmthermoinfo}, we can then rewrite the inaccessible information term as %
\begin{align}
D(\PCONDS\Vert\projG{\PCONDS}) =
\Bigg\langle\! \ln \frac{q(m\vert x)}{\int_{\G} q(m\vert \g(x)) d\mu(\g)}\!\Bigg\rangle,
\label{eq:decomp4b}
\end{align}
where $q(m\vert x)$ is the measurement channel and  $\langle \cdot \rangle$ is indicates expectation under the joint distribution $p(x,m)=p(x)q(x|m)$. 
\cref{eq:decomp4b} conveniently expresses the inaccessible information in terms of the asymmetry of the measurement channel relative to the action of $\G$ (the right side of \cref{eq:decomp4b} vanishes when $q(m\vert x)$ is invariant under that action), which we will exploit  
in some of our examples below.

\subsection{Example: Szilard box with symmetry constraints}

\label{subsec:Szilard-box-example}

\begin{figure}
\begin{centering}
\includegraphics[height=0.4\columnwidth]{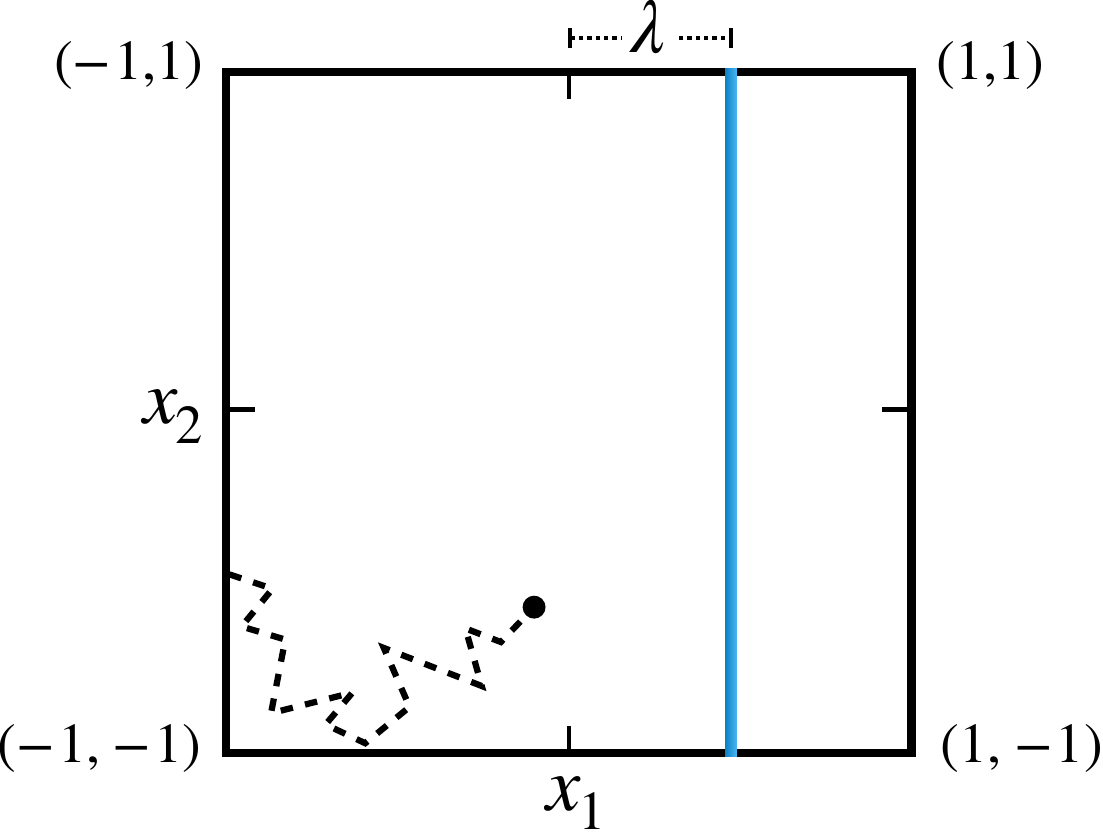}
\par\end{centering}
\caption{%
A Szilard box with energy functions as in \cref{eq:ham0}.
\label{fig:szilard-ham}}
\end{figure}
We demonstrate our results on symmetry constraints %
using the Szilard
box shown in \cref{fig:szilard0}. 
We assume that the box is coupled to a single
heat bath at inverse temperature $\beta=1$, and %
that the particle inside the box has overdamped Fokker-Planck dynamics, so that 
all generators have the form of \cref{eq:fp10}. The system's
state is represented by a horizontal and a vertical coordinate, $x=(x_{1},x_{2})\in\mathbb{R}^{2}$.

Suppose that all energy functions have the form %
\begin{equation}
E_\paramA(x_{1},x_{2})=\Vp(x_{1}-\paramA)+\Vw(|x_{1}|)+\Vw(|x_{2}|),\label{eq:ham0}
\end{equation}
where $\paramA\in\mathbb{R}$ is a controllable parameter that determines the location of the vertical partition,
$\Vp$ is the partition's repulsion potential, and $\Vw$ is the repulsion
potential of the box walls:
\begin{equation}
\Vw(a)=\begin{cases}
0 & \text{if }a\le \boxRight \\
\infty & \text{otherwise}
\end{cases}\label{eq:potential1}
\end{equation}
meaning that the box extends over $(x_{1},x_2)\in\boxExtent^2$~\footnote{Technically, the wall potential as defined in \cref{eq:potential1} is non-differentiable. To be more accurate, one should imagine it in terms of the limit \unexpanded{$\Vw(|x|)=\lim_{\alpha \to \infty} |x|^\alpha$}~\cite{dhar2019run}.}. 
Assume that $\Vp(x-\paramA)=0$ for some value of $\paramA$ (i.e., the partition can be removed by setting $\paramA$ outside
the box). For such $\paramA$, %
let $\Hunif$ indicate the corresponding energy function, and note that it obeys $\Hunif(x_1,x_2)=0$ within the box (and infinity elsewhere),  
corresponding to a uniform equilibrium distribution
$\po(x_1,x_2)=\mathbf{1}_{[-1,1]^2}(x_1,x_2)/4$ (where $\mathbf{1}$ is the indicator function). %
This Szilard box is shown schematically in \cref{fig:szilard-ham}.

The energy functions in \cref{eq:ham0} obey the vertical reflection symmetry $E(x_{1},x_{2})=E(x_{1},-x_{2})$, corresponding to the %
two-element symmetric group $S_2$ whose action is generated by $g(x_{1},x_{2})=(x_{1},-x_{2})$. 
The corresponding twirling of $p$ is the uniform mixture of $p$ and its reflection,
\begin{align}
\projG{p}(x_1,x_2)=(p(x_1,x_2)+p(x_1,-x_2))/2.
\label{eq:symmPhiDef}
\end{align} 

We can use our results to %
derive bounds on the work that can be extracted from this Szilard box. Intuitively, the set of allowed generators $\LL$ --- that is,  Fokker-Planck operators with energy functions as in \cref{eq:ham0}, corresponding to different horizontal locations 
of the vertical partition --- all obey vertical reflection symmetry. Thus, the  dynamics generated by those Fokker-Planck operators commute with $\projGbase$, the twirling operator defined in \cref{eq:symmPhiDef}. %
Using  \cref{eq:ourworkbound}, we can  bound the work extracted during any transformation $\ptpp$ in terms of the decrease of the accessible free energy, $\NFEs[\projG \ps]-\NFEf[\projG \pf]$.

In more detail, consider some driving protocol which starts and ends with the partition
removed. At intermediate times, the driving protocol manipulates the location of the partition so as to bring the system from some initial distribution $p$ to a final equilibrium distribution $p'=\po$ while extracting work. 
The second law gives bounds on EP, $\EPp\ge 0$, and work: %
\begin{align}
W(\dTrans{p}{\po}) & \le\NFE[p][\Hunif]-\NFE[\po][\Hunif]=D(p\Vert \po),\label{eq:w1}
\end{align}
which follows from  \cref{eq:workbound,eq:Fkl}.  
However, this bound can be too optimistic due to the protocol constraints. 
Given \cref{eq:epfreebound}, as well as the fact that the final distribution obeys $\projG \po = \po$, 
we know that $\EPp \ge D(p\Vert\projG{p})$. Similarly, \cref{eq:ourworkbound} gives a  tighter bound on extractable work %
\begin{align}
W(\dTrans{p}{\po}) \le\NFE[\projG \ps][\Hunif]-\NFE[\po][\Hunif]= D(\projG{p}\Vert \po),\label{eq:symmwork1}
\end{align}
where the second equality follows from \cref{eq:Fkl}.

\begin{figure}
\begin{centering}
\includegraphics[width=0.75\columnwidth]{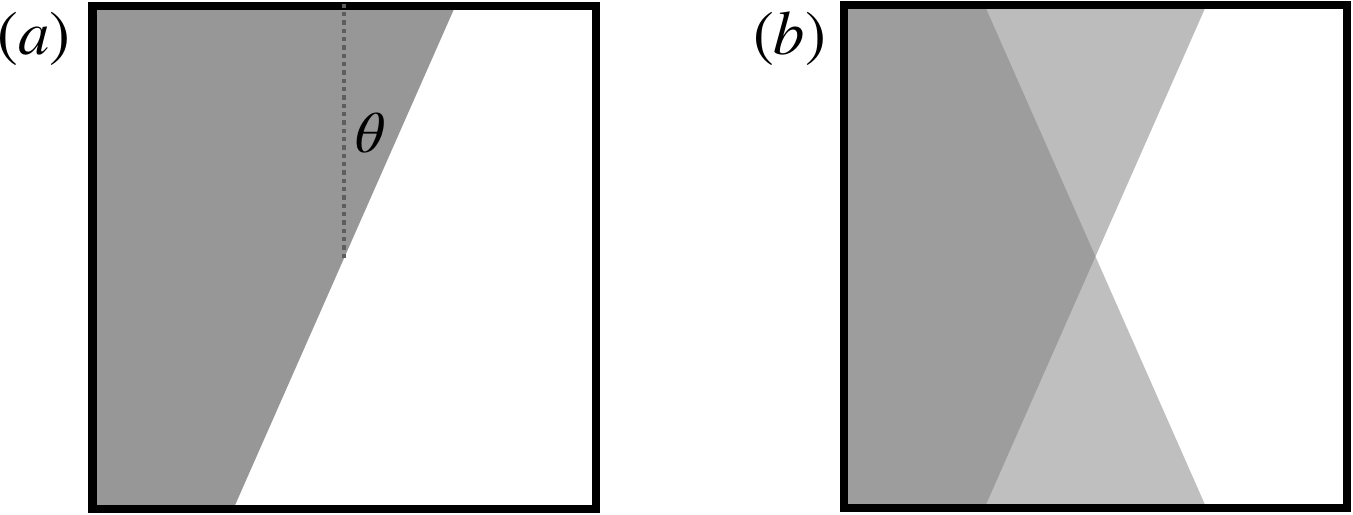}
\par\end{centering}
\caption{%
$(a)$ A non-equilibrium distribution $\ptheta$ that is ``rotated'' by an arbitrary angle $\theta$, \cref{eq:degMeas}.
$(b)$ The distribution in (a) under the action of the vertical reflection twirling operator, $\projG{\ptheta}$. 
\label{fig:szilard-symm-angle}}
\end{figure}

It is easy to use these results to resolve the question raised in the introduction: can one show that work can only be extracted from a measurement of whether the particle is in the left or right half of the box, rather than a measurement of whether the particle is in the top or bottom half of the box?
Suppose that the particle's initial
distribution $p$ is uniform across the {left or right half} of the box. Such a 
 distribution $p$ is invariant under vertical reflection, so $p=\projG{p}$  and  \cref{eq:symmwork1} gives $W(\dTrans{p}{\po}) \le D(p\Vert \po) = \ln 2$, 
the same as the bound set by the second  law, \cref{eq:w1}. %
This bound  can be achieved by quickly moving the partition to the middle
of the box, and then slowly moving it rightward. 
Conversely, suppose that under the initial distribution $p$,
the particle is uniformly distributed across the {top or bottom half} of the box. The twirling of such a distribution is a uniform distribution over the box, $\projG{p}=\po$. In this case, \cref{eq:symmwork1} gives 
$W(\dTrans{p}{\po})\le 0$, 
meaning that no work can be extracted.

\newcommand\angleRange{(\frac{\pi}{4},\frac{3\pi}{4})}

We now demonstrate the power of our approach by analyzing extractable work given a more complex family of initial distributions (while using the same energy functions as above). Suppose that the  initial distribution is concentrated
within half the box, as determined by a separating line that is rotated by an arbitrary angle  %
$\theta \in [-\pi,\pi]$ (see \cref{fig:szilard-symm-angle}(a)). 
This initial distribution can be written formally as %
\begin{align}
\ptheta(x_1,x_2)&=\frac{\mathbf{1}_{[-1,1]^2}(x_1,x_2)}{2}\Theta(x_2 \sin \theta - x_1 \cos\theta),
\label{eq:degMeas}
\end{align}
where $\Theta$ is the Heaviside function.  
For instance, $\ptheta$ for $\theta=0$ corresponds to the particle being in the left half of the box, while $\ptheta$ for $\theta=\pi/2$  corresponds to the particle being in the top half of the box.

Because we are considering the same set of generators as above, we can bound the extractable  work in a 
given $\ptheta$ using the same twirling operator as defined above in \cref{eq:symmPhiDef}. 
(For a sample $\ptheta$, the twirling $\projG{\ptheta}$ is illustrated in \cref{fig:szilard-symm-angle}(b).)
Using \cref{eq:symmwork1}, the extractable work obeys $W(\dTrans{\ptheta}{\po})\le D(\projG{\ptheta} \Vert \po)$. Moreover, as we show in \cref{app:rotatedaccessDeriv}, this KL divergence can be written in closed form as
\begin{align}
D(\projG{\ptheta} \Vert \po)\!=\! {\ln2}\cdot\!\begin{cases}
\frac{1}{2}\vert\tan(\theta-\frac{\pi}{2})\vert &\! \text{\ensuremath{\vert\theta\vert\in\angleRange}}\\
1- \frac{1}{2}\vert\tan\theta\vert & \text{otherwise.}
\end{cases}
\label{eq:accFreeEnergyRotated}
\end{align} 
This result is plotted as a function of $\theta$ %
 in \cref{fig:szilard-symm-data}. %

\begin{figure}
\begin{centering}
\includegraphics[height=0.4\columnwidth]{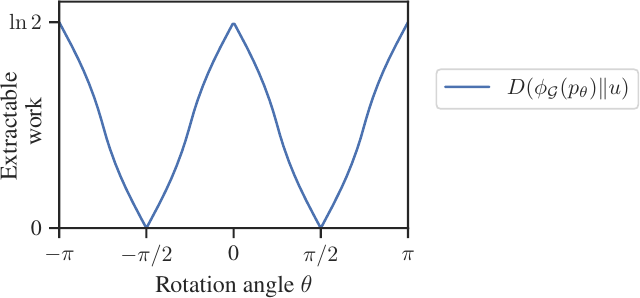}
\par\end{centering}
\caption{%
Szilard box with symmetry constraints: the bound on extractable work %
 as a function of $\theta$,   
\cref{eq:accFreeEnergyRotated}. %
\label{fig:szilard-symm-data}}
\end{figure}

We can also analyze the thermodynamics of information for different measurements of the Szilard box. Imagine that, starting from a uniform equilibrium distribution, one measures which side of the box contains the particle, as determined by a separating line at some arbitrary angle $\theta \in[-\pi,\pi]$. %
For this measurement, the conditional distribution over system states $\pconds$ is equal to $\ptheta$ half the time (as in \cref{fig:szilard-symm-angle}(a)), and equal to $\pthetapi$ the other half the time. 
 Then, for both measurement outcomes, one manipulates the vertical partition so as to drive the particle back to the equilibrium distribution $\pf=\po$ while extracting work. For simplicity, we assume that the initial and final energy functions are the same.

The general bound on average extractable work for feedback control, \cref{eq:workbound-2}, gives
\begin{align}
\langle W \rangle \le \IXMs = \ln 2,
\end{align}
where we've used that $\ps=\pf$ and $\IXMf=0$.  Our results provide a tighter bound,   showing that
the average extractable work is bounded by the accessible information in the measurement,
\begin{align}
\!\langle W \rangle\! \le\! \IaccS[\projGbase]\!=\!\frac{D(\projG{\ptheta}\Vert \po)\!\!+\!\!D(\projG{\pthetapi}\Vert \po)}{2},\label{eq:szwork0}
\end{align}
where we used \cref{eq:workboundinfo,eq:accinfodef2}. 
It can be verified from \cref{eq:accFreeEnergyRotated} that $D(\projG{\ptheta}\Vert \po)=D(\projG{\pthetapi}\Vert \po)$. Thus, the accessible information for a given $\theta$ is simply equal to $D(\projG \ptheta \Vert \po)$, the right side of \cref{eq:accFreeEnergyRotated}, and shown in 
 \cref{fig:szilard-symm-data}.
 As expected, the accessible information achieves a maximum of $\ln 2$ at $\theta=0$ (or $\theta =\pm \pi$), which corresponds to a  measurement of whether the particle is on the left or right side of the box. %
 The accessible information 
 falls nonlinearly (but continuously) to a minimum of 0 at $\theta=\pm\pi/2$, which corresponds to a measurement of whether the particle is on the top or bottom of the box. 

In the example above, the accessible information quantifies in a very literal way the ``alignment'' between the choice of measurement and the way the system can be manipulated.  
More generally, this example illustrates how our  bounds on EP and work depend on the interplay between the operator $\klopBase$, the initial/final distributions $\ps$ and $\pf$, and (for feedback control protocols) the choice of measurement $M$. This interplay can give rise to highly non-trivial thermodynamic bounds, such as in \cref{eq:accFreeEnergyRotated,fig:szilard-symm-data}, even for very simple operators $\klopBase$, such as in \cref{eq:symmPhiDef}.

Finally, we note that our analysis above only assumes that the energy functions are vertically symmetric, which includes many energy functions that do not have the form of the vertical partition defined in \cref{eq:ham0}. Furthermore, while the bounds on work and EP which we derive here are achievable by \emph{some} vertically symmetric energy functions, they are not necessarily achievable by manipulating the  location of a vertical partition. For instance, achieving the extractable work bound for a given $\theta$, \cref{eq:accFreeEnergyRotated}, generally requires that the corresponding twirled distribution $\projG{p}$, such as the one shown in \cref{fig:szilard-symm-angle}(b), is an equilibrium distribution for some available energy function.

We analyze the same system using a different set of constraints  in \cref{subsec:Szilard-box-example-mod,subsec:Szilard-box-example-cg} below. 
(Also see~\cite{still2021partially} for a different recent analysis of the thermodynamics of the Szilard box with rotated measurements, though from the point of view of partial observability rather than protocol constraints.)

\subsection{Example: Feedback control on the Ising model}
\label{sec:ising}

Our bounds on symmetry constraints can be useful for various multi-particle systems with symmetries, such as gases of indistinguishable particles and  spin systems with symmetries.  We demonstrate this by analyzing the thermodynamics of feedback control on an Ising model.
The reader may also be interested in \cref{app:unicyclic}, where we analyze a simpler and more pedagogical example of a discrete-state system with symmetry constraints.

Consider a 2D Ising model on a square lattice on a torus, containing a total of $N^2=N\times N$ spins. The state of the lattice is indicated as $x\equiv (x_{1},\dots,x_{N^2})$, where $x_{i}\in\{-1,1\}$ is the state of the spin at location $i$. We assume that the energy functions have the following form,
\begin{align}
E(x) = -J \sum_{\mathclap{(i,j)\in\mathcal{N}}} x_{i} x_{j} - H \sum_{i} x_{i}.
\label{eq:ising}
\end{align}
where $\mathcal{N}$ is the set of all nearest neighbors on the lattice, $J$ is the coupling strength, and $H$ is the external magnetic field. 
%

%
Energy functions like these are invariant under the symmetry group $\G$ corresponding to horizontal and vertical translations of the lattice  (for simplicity, we ignore other symmetries of the lattice, such as reflections and rotations). The action of this group is given by a set of $N^2$ bijections $g_{a,b} : X\to X$ for $a,b \in \{0,\dots,N-1\}$, where $g_{a,b}(x)$ translates the lattice state $x$ to the right by $a$ spins and upward by $b$ spins (with periodic boundary conditions). We assume that the system evolves according to Glauber dynamics~\cite{krapivsky_kinetic_2010}, or some other dynamics that respects the translational symmetry of the 2D lattice, such that \cref{eq:symmME} is satisfied.

Given these assumption, we can derive thermodynamic bounds for the 2D Ising model in terms of the following twirling operator, %
\begin{align}
\projG{p}(x)=N^{-2} \sum_{a=0}^{N-1}\sum_{b=0}^{N-1} g_{a,b}(x).
\label{eq:symmPhiDefIsing}
\end{align} 
We use this twirling operator to analyze the thermodynamics of the following feedback-control setup on the Ising model, %
also shown in \cref{fig:ising}. 
The lattice is initially in equilibrium $p$ at some temperature $\beta$ and $J=1,H=0$ (no external field). The state of the spin at location $1$ is then measured under the measurement channel $q(m\vert x)= \deltaFunc[m](x_1)$, where $\delta$ is the Kronecker delta. Since there is no initial external field, the two outcomes $m\in\{-1,1\}$ have equal probability and $\IXMs=\ln 2$. The measured outcome is then used to select a driving protocol, which extracts work from the system by manipulating the control parameters $J$ and $H$.  At the end of the protocol corresponding to each outcome, the system is brought back to the original equilibrium (so $\pcondf=p$ for all $m$). %
For simplicity, we assume that the initial and final energy functions are the same.

\begin{figure}
\includegraphics[valign=c,width=0.18\columnwidth]{\figdir/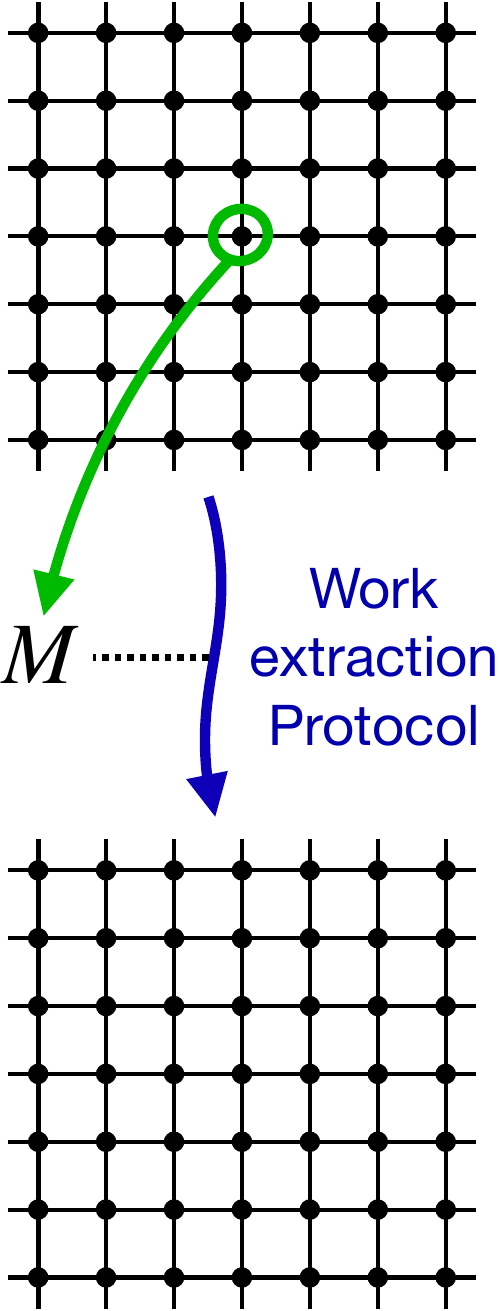}
\hspace{10pt}
\includegraphics[valign=c,width=0.68\columnwidth]{\figdir/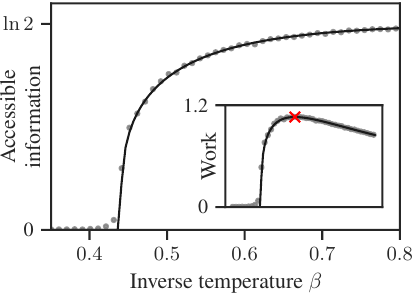}
\caption{Thermodynamics of information on a 2D Ising model. Left: a measurement $M$ is made of the state of a single spin (green), and then used to %
drive the system while extracting work (blue). Right: the accessible information $\IaccS[\projGbase]$ increases with inverse temperature after the critical value $\beta_c \approx 0.44$ (grey circles from Monte Carlo simulations, black line from  closed-form expression, \cref{eq:analyticIsing}). Inset shows the bound on extractable work, $\IaccS[\projGbase]/\beta$, which peaks at $\beta\approx 0.547$ (red cross). %
\label{fig:ising}}
\end{figure}

Under this setup, one can verify that $\IaccF[\projGbase]=0$ and $\NFEs[p]=\NFEf[p']$, so \cref{eq:workboundinfo} bounds 
 average extractable work as %
$\langle W \rangle \le \IaccS[\projGbase]/\beta$,
where $\IaccS[\projGbase]$ is the accessible information from \cref{eq:accinfodef}. 
Using \cref{eq:accinfodef,eq:decomp4b},  we can  write this accessible information as
\begin{align}
&\IaccS[\projGbase] = %
\ln 2 - \Bigg\langle  \ln \frac{q(m\vert x)}{N^{-2}\sum_{a,b} q(m\vert \g_{a,b}(x))} \Bigg\rangle,\label{eq:hg7b}
\end{align}
where $\langle \cdot \rangle$ indicates expectation over the joint distribution $p(x)q(m\vert x)$, where $p(x)$ is the initial equilibrium distribution at inverse temperature $\beta$ and $J=1,H=0$.  We emphasize that the accessible information depends on  $\beta$ (though we leave this dependence implicit in the notation).

In general, one can estimate the accessible information in \cref{eq:hg7b}  using various numerical techniques (e.g., by sampling from the initial equilibrium distribution using Monte Carlo methods). It is also possible to use Onsager's well-known solution of the 2D Ising model to calculate the accessible information in closed form. In particular, in \cref{app:ising} we show that in the thermodynamic limit  $N\to\infty$,
\begin{align}
\IaccS[\projGbase] =  %
\begin{cases}0 &\text{for $\beta \le \beta_c$}\\
 \ln 2 - h_2\Big(\frac{1+\sqrt[8]{1- (\sinh2\beta)^{-4}}}{2}\Big) &\text{for $\beta > \beta_c$.}
\end{cases}
\label{eq:analyticIsing}
\end{align}
where $h_2(x)=-x\ln x -(1-x)\ln(1-x)$ is the binary entropy function %
and $\beta_c=\ln(1+\sqrt{2})/2\approx 0.44$ is the critical inverse temperature of the 2D Ising model. 
This result is verified in \cref{fig:ising}, where we compare \cref{eq:analyticIsing} with a Monte Carlo estimate of \cref{eq:hg7b} on a 100x100 lattice.  It can be seen that 
in the high temperature (low $\beta$) regime, %
the accessible information vanishes. In the low temperature (high $\beta$) regime, the amount of accessible information increases, approaching $\ln 2$ as $\beta \to \infty$.

We also plot the bound on average extractable work, $\langle W\rangle \le \IaccS[\projGbase]/\beta$, in the inset in \cref{fig:ising}. This bound is the ratio of two terms: the accessible information $\IaccS[\projGbase]$ and the inverse temperature $\beta$, both of which are increasing in $\beta$. In fact, it can be seen from \cref{fig:ising} that the bound on extractable work peaks at a finite value of $\beta$, the  optimal inverse temperature for work extraction. Using \cref{eq:analyticIsing} and numerical techniques, we find this optimal value to be $\beta \approx 0.547$, leading to the bound $\langle W\rangle \le 1.06$ joules.

This shows that the amount of accessible information provided by a given measurement can depend on the structure of correlations in the system, and therefore vary dramatically as the system undergoes a phase transition. %
At a high level, 
any driving protocol that is restricted to energy functions like  \cref{eq:ising} can only extract work from ``global'' (i.e., translationally invariant) information. %
If the measurement acquires such information (e.g., if it directly measures the spatially-averaged magnetization), then in principle all of the acquired information may be extractable as work. Measurement of the state of a single spin, however, in general  provides only local information. The temperature dependence observed in  \cref{eq:analyticIsing} and \cref{fig:ising} arises from the presence of long-range order in the magnetic regime ($\beta>\beta_c$). In this regime, the state of each spin is highly correlated with the magnetization of the entire lattice, so local and global information are equivalent.  In the high temperature regime ($\beta<\beta_c$), the state of a single spin is not correlated with any kind of global information, and so most of the measured information is inaccessible. %

For a different kind of analysis of the thermodynamics of a 1D Ising model under constraints, see \cite{lekscha_quantum_2018}.

\section{Modularity constraints}
\label{sec:modularity}
Many systems of interest exhibit modular organization, meaning
that their degrees of freedom can be grouped into decoupled 
subsystems. Examples of modular systems include computational devices such as
digital circuits \citep{gershenfeld1996signal,Boyd:2018aa,wolpert2020thermodynamic}, regulatory networks in biology
\citep{schlosser2004modularity}, and brain
networks \citep{sporns2016modular}.

We use our framework to derive bounds on work and EP for modular systems. %
We begin by introducing some terminology and notation. Consider a system whose degrees
of freedom are indexed by the set $V$, such that the overall state
space can be written as $X=\bigtimes_{v\in V}\sX_{v}$, where $X_{v}$
is the state space of degree of freedom $v$. We use the term \emph{subsystem}
to refer to any subset of the degrees of freedom, $\modSubsys\subseteq V$. 
We use $X_{\modSubsys}$ 
to indicate the random variable representing the state of subsystem $\modSubsys$ and $x_{\modSubsys}$ to indicate an actual state  of $\modSubsys$.  
Given some distribution $p$ over the entire system, we use 
$p_{\modSubsys}$ to indicate a marginal
distribution over subsystem $\modSubsys$, and %
$[Lp]_A$ to indicate the 
derivative of the marginal distribution of subsystem $A$ under the generator $L$.  

We use the term \emph{modular decomposition}
to refer to a set of subsystems $\modDecomp$, %
such that each $v\in V$ belongs
to at least one subsystem $\modSubsys\in\modDecomp$. 
Note that some of the degrees of freedom $v\in V$ can belong to more than one subsystem
in $\modDecomp$. We use
\begin{align}
\label{eq:olapdef}
\olaps = \bigcup_{ A,B\in\modDecomp: A\ne B} (A\cap B)
\end{align}
to indicate those degrees of freedom that belong to more than one subsystem in $\modDecomp$, which we refer to as the \emph{overlap}. We will often write $\olapsShort$ instead of $\olaps$ for notational simplicity.

\newcommandx\Lmod[1][usedefault, addprefix=\global, 1=\modSubsys]{\LL^{(#1)}}
\newcommandx\Rmod[1][usedefault, addprefix=\global, 1=\modSubsys]{R^{(#1)}}

We say that the available driving protocols obey %
\emph{modularity constraints} (with respect to the modular decomposition $\modDecomp$) if each generator $\LL\in\LLL$ can
be written as a sum of generators of the different subsystems in $\modDecomp$, %
\begin{align}
\LL=\sum_{\modSubsys\in\modDecomp}\Lmod,
\label{eq:m0}
\end{align}
and each $\Lmod$ obeys two properties:  
the dynamics over the marginal distribution $p_{\modSubsys}$ are closed under $\Lmod$ (depend only on the marginal distribution over $\modSubsys$),
\begin{equation}
p_{\modSubsys}=q_{\modSubsys}\implies[\Lmod p]_\modSubsys=[\Lmod q]_\modSubsys\qquad\forall p,q\in\Dset,\label{eq:m2}
\end{equation}
and the distribution over other subsystems besides
$\modSubsys$ does not change under $\Lmod$,
\begin{equation}
[\Lmod p]_{\otherS}=0\qquad\quad\forall p\in\Dset,\otherS \in \modDecomp \setminus\{ \modSubsys\}.\label{eq:m3}
\end{equation}
In other words, we require that each subsystem evolves independently, and does not affect the other subsystems.

The role of the degrees of freedom in the overlap is somewhat subtle. It can be verified that \cref{eq:m3} %
implies that the degrees of freedom in the overlap cannot change state when evolving under $\LL$. Importantly, however, the overlap 
may influence the dynamics of those degrees of freedom that \emph{can} change state. For example, consider an inclusive model of a feedback control setup: there are two nested subsystems, $\modDecomp=\{\modSubsys,\otherS\}$ with $\otherS \subseteq \modSubsys$, and the degrees of freedom in $\olapsShort = \otherS$ (the controller) cannot change state but can influence the evolution of $\modSubsys\setminus\otherS$. More elaborate feedback control setups, in which the same controller can control multiple subsystems, can be modeled using decompositions with multiple non-nested subsystems. Other examples of modular decompositions with overlap include circuits~\cite{wolpert2020thermodynamic}, spin systems where some spins are pinned by local magnetic fields, and many-particle systems where some particles have no mobility. %

We can also provide more concrete conditions when \cref{eq:m2,eq:m3} hold for discrete-state master equations and Fokker-Planck equations. 
For discrete-state master equations, it can be verified by inspection that %
\cref{eq:m2,eq:m3} %
 hold when all $L \in\LLL$
can be written in the form
\begin{align}
\Lij=\sum_{\modSubsys\in\modDecomp} \Rmod_{x_{\modSubsys}',x_{\modSubsys}} \deltaFunc[x_{V\setminus \modSubsys}](x_{V\setminus \modSubsys}'),
\label{eq:mcond}
\end{align}
where $\delta$ is the Kronecker delta and $\Rmod$ is some rate matrix over subsystem $\modSubsys$ that does not allow the degrees of freedom in the overlap to change state ($ \Rmod_{x_{\modSubsys}',x_{\modSubsys}} = 0$ if $x_{\modSubsys \cap \olapsShort} \ne x_{\modSubsys\cap \olapsShort}'$).
For Fokker-Planck equations, for simplicity consider overdamped dynamics of the form 
\begin{equation}
Lp =\sum_{v\in V} \gamma_{v}^L\partial_{x_{v}}\Big[(\partial_{x_{v}}E_L)p + \dL\partial_{x_{v}}p\Big],\label{eq:fp9}
\end{equation}
where $\gamma_{v}^L$ is the mobility coefficient along dimension $v$ and $E_L$ is the potential energy function associated with generator $\LL$. 
Such equations can represent potential-driven Brownian
particles coupled to a heat bath, where the different mobility coefficients represent different particle masses or sizes~\footnote{One
can also apply the results in this section to Fokker-Planck equations
that can be put in the form of \cref{eq:fp9} via an appropriate
change of variables, see \citep[Sec. 4.9]{risken1996fokker}.}.  
Now imagine that for all $\LL\in\LLL$, the
energy functions are additive over the subsystems, and that the degrees of freedom in the overlap have no mobility:
\begin{align}
E_L(x)=\sum_{\modSubsys\in\modDecomp}E^{(\modSubsys)}_L(x_{\modSubsys}),\quad\;\;\;\gamma_{v}^L=0 \;\;\; \forall v \in \olapsShort.\label{eq:modenergy2}
\end{align}
In that case, \cref{eq:fp9} can be rewritten in the form of \cref{eq:m0},  with $\Lmod p=\sum_{v\in \modSubsys\setminus \olapsShort} \gamma_{v}^L\partial_{x_{v}}[(\partial_{x_{v}}E^{(\modSubsys)}_L)p_\modSubsys + \dL\partial_{x_{v}}p_\modSubsys]$, 
and satisfies 
\cref{eq:m2,eq:m3}.

We now define  the following nonlinear operator $\projMbase$: %
\begin{align}
\projM{p} = p_{\olapsShort} \prod_{{\modSubsys \in \modDecomp}} p_{\modSubsys \setminus \olapsShort \vert \modSubsys \cap \olapsShort}.
\label{eq:projMbasedef}
\end{align}
This operator preserves the statistical correlations {within} each subsystem  $\modSubsys \in \modDecomp$, as well as within the overlap $\olapsShort$, while destroying %
all other statistical correlations. %
As a simple example, if all the subsystems in $\modDecomp$ are non-overlapping, then $\projM{p}$ has the product form $\projM{p}=\prod_{\modSubsys \in \modDecomp} p_{\modSubsys}$.  
In  \cref{app:mod}, we show that $\projMbase$ obeys the Pythagorean identity, \cref{eq:pyth}. 
We also show that if some generator $\LL(t)$ 
obeys \cref{eq:m2,eq:m3}, then $e^{\inftime L(t)}$ commutes with
$\projMbase$, so \cref{eq:comm0} holds. 
This means that for any protocol that carries out the transformation $\ptpp$ while obeying modularity constraints, %
the decompositions and bounds for EP and work derived in \cref{sec:info-geom-framework} are satisfied for $\klopBase = \projMbase$. %
In particular, using  \cref{eq:feffdecomp}, we can decompose the free energy $\NFEs$ of any distribution $p$ into the accessible free energy  $\NFEs[\projM p]$ %
and the  inaccessible free energy $D(p\Vert\projM p)/\beta$.
Note that $D(p\Vert \projM p)$ is a non-negative measure of the amount of statistical correlations between the subsystems of $\modDecomp$ under distribution $p$, which vanishes when each subsystem is conditionally independent given the overlap $\olapsShort$. Thus, for a protocol that obeys modularity constraints, \cref{eq:epfreebound} states that the drop in those statistical correlations is a lower bound on EP, and that the amount of statistical correlation between the subsystems of $\modDecomp$  cannot increase over the course of the protocol. 
(There is a fair amount of closely related prior work; see \cref{sec:priorwork}.)

A particularly simple application of our bounds %
 occurs when $\modDecomp$ contains two (possibly overlapping) subsystems, $\modDecomp=\{A,B\}$. 
In that case, the bounds in \cref{eq:epfreebound} can be rewritten in terms of the drop of a conditional mutual information between the two subsystems,
\begin{align}
\EPp \ge I(X_A;X_B\vert X_{A\cap B})-I(X_A';X_B'\vert X_{A\cap B}') \ge 0.\label{eq:ep-mod-condmi}
\end{align}
If the subsystems do not overlap, this can be further rewritten as  the drop of the regular mutual information,
\begin{align}
\EPp \ge I(X_A;X_B)-I(X_A';X_B') \ge 0.\label{eq:ep-mod-mi}
\end{align}
More generally, if  $\modDecomp$  contains an arbitrary number of non-overlapping subsystems, the EP can be bound as
\begin{align}
\EPp \ge \mathcal{I}(p)-\mathcal{I}(p') \ge 0,
\label{eq:ep-mod-tc}
\end{align} 
where $\mathcal{I}(p)=\big(\sum_{A\in\modDecomp}S(p_A)\big)-S(p)$ is the {multi-information}
in distribution $p$ with respect to partition $\modDecomp$~\footnote{The multi-information
is a well-known generalization of mutual information, which is also
sometimes called ``total correlation'' \citep{watanabe1960information}.}.

We finish by discussing thermodynamics of information under modularity constraints. In general, the results  derived in \cref{sec:thermo-of-info} apply to modularity constraints as a special case.  However, we can also exploit special properties of the operator $\projMbase$  to further simplify the expression of accessible information. %
Suppose that the distribution $\ps$ is invariant under $\projMbase$, so $\ps=\projM{\ps}$ (e.g., if $p$ is an equilibrium distribution, see \cref{eq:equilinimgKLOP}). Using \cref{eq:projMbasedef}, we can then rewrite \cref{eq:accinfodef2} as
\begin{align}
\IaccS[\projMbase] &=I(X_\olapsShort;M)+\sum_{\mathclap{\modSubsys\in\modDecomp}} I(X_{\modSubsys} ; M \vert X_{\modSubsys\cap \olapsShort}).
\label{eq:modInfo}
\end{align}
Thus, the accessible information in measurement $M$ is the information that $M$ provides about the overlap, plus the conditional mutual information between each subsystem and $M$ given the relevant part of the overlap. This means that only information about individual subsystems --- not about inter-subsystem correlations --- can be turned into work.  
If there is no overlap, \cref{eq:modInfo} can be further simplified as  %
\begin{align}
\IaccS[\projMbase]&=\sum_{\mathclap{\modSubsys\in\modDecomp}} I(X_\modSubsys ; M).
\label{eq:modInfo2}
\end{align}
We will use these expressions in some of our examples below.

\subsection{Example: Szilard box with modularity constraints}
\label{subsec:Szilard-box-example-mod}

We illustrate our results for modularity constraints on a Szilard box. In doing so,  we will demonstrate two important concepts: first, how the same set of generators $\LLL$ can be analyzed under different constraints, resulting in different bounds on work and EP (compare this section to \cref{subsec:Szilard-box-example}); second, how bounds arising from multiple constraints can be stacked on top of each in an iterative manner, as in \cref{eq:stackedbounds} (we will combine bounds from modularity and symmetry constraints).

We consider the same setup as in \cref{subsec:Szilard-box-example}: there is a single overdamped particle in a box  coupled to a bath at inverse temperature $\beta=1$, which evolves under potential energy functions as in  \cref{eq:ham0}. 
This system is driven from some initial distribution $p$ to a final uniform equilibrium distribution, $\pf=\po$ while extracting work.

\begin{figure} 
\begin{centering}
\includegraphics[width=0.75\columnwidth]{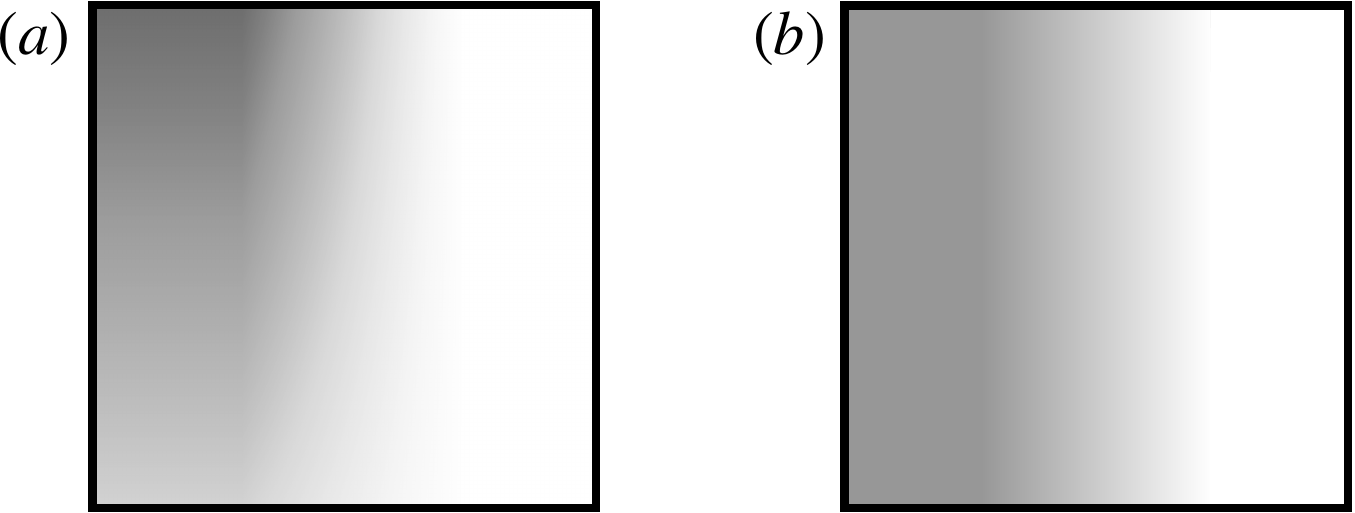}
\par\end{centering}
\caption{%
$(a)$ Given a   ``rotated'' distribution $\ptheta$, as shown above in \cref{fig:szilard-symm-angle}(a), this shows the decorrelated distribution $\projM{\ptheta}$, as in \cref{eq:twirlMsz}. 
$(b)$ The decorrelated and twirled distribution, $\projG{\projM{\ptheta}}$. 
\label{fig:szilard-mod-angle}}
\end{figure}

\begin{figure}[b] 
\begin{centering}
\includegraphics[height=0.4\columnwidth]{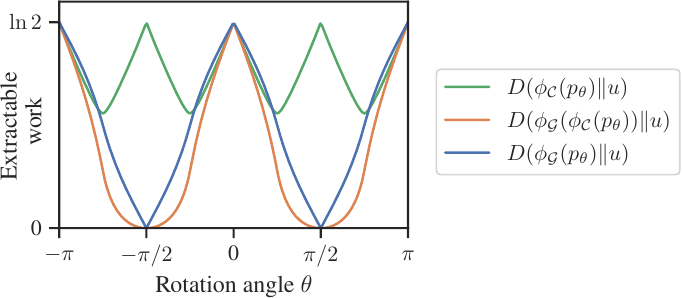}
\par\end{centering}
\caption{%
Bounds on extractable work %
 as a function of $\theta$, as derived from only modularity constraints (in green, 
\cref{eq:cf01}), a combination of modularity+symmetry constraints (in orange, \cref{eq:cf02}), and only symmetry constraints (in blue, \cref{eq:accFreeEnergyRotated}). %
\label{fig:szilard-mod-data}}
\end{figure}

Note that the energy functions in \cref{eq:ham0} have no interaction terms between $x_1$ (the horizontal position of the particle) and $x_2$ (the vertical position of the particle). That means that the allowed driving protocols obey modularity constraints for a decomposition of the system into two subsystems, $\modDecomp = \{\{X_1\},\{X_2\}\}$ (since \cref{eq:modenergy2} is satisfied for the decomposition). 
 This allows us to analyze EP and work using an operator $\projMbase$ which maps each joint distribution over $X_1\times X_2$ into a product distribution,
\begin{align}
\projM{p}(x_1, x_2) = p(x_1)p(x_2).
\label{eq:twirlMsz}
\end{align}
In particular, using the same derivation as in \cref{eq:symmwork1}, we can bound the extractable work in terms of the accessible free energy in $p$,
\begin{align}
W(\dTrans{p}{\po})\le D(\projM{p}\Vert \po).
\label{eq:projMbound}
\end{align}

As discussed in \cref{subsec:Szilard-box-example}, this system also obeys symmetry constraints, corresponding to the vertical reflection twirling operator $\projGbase$ defined in \cref{eq:symmPhiDef}. %
We can use \cref{eq:ourworkboundstacked} to  bound the extractable work using a combination of $\projMbase$ and $\projGbase$, 
 \begin{align}
W(\dTrans{p}{\po})&\le D(\projM{\projG{p}}\Vert \po)
\label{eq:projMbound2}\\
W(\dTrans{p}{\po})&\le D(\projG{\projM{p}}\Vert \po).
\label{eq:projMbound3}
\end{align}

For concreteness, imagine that the initial distribution $p$ is concentrated within half the box, as determined by a separating line  rotated by some arbitrary angle 
$\theta \in [-\pi,\pi]$, so  $\ps=\ptheta$ from \cref{eq:degMeas} (see \cref{fig:szilard-symm-angle}(a) for an illustration). 

We consider the extractable work bound in \cref{eq:projMbound} for the initial distribution $\ptheta$.
For a given $\ptheta$, the corresponding decorrelated initial distribution $\projM{\ptheta}$  is illustrated in \cref{fig:szilard-mod-angle}(a). 
Then, the accessible free energy in \cref{eq:projMbound} can be expressed in closed form as (see \cref{app:szModDeriv}),
\begin{multline}
D(\projM{\ptheta}\Vert \po) = \ln 4 - \frac{1}{2}\Big[
\min\{|\tan \theta|,|\tan(\pi/2-\theta)|\}\\
+f(\max\{|\tan \theta|,|\tan(\pi/2-\theta)|\})\Big],
\label{eq:cf01}
\end{multline}
where for notational convenience we've defined
\begin{align}
f(x)=1-\frac{1+x^2}{2x} \ln \frac{x+1}{x-1} - \ln \frac{x^2-1}{4x^2}.
\label{eq:fffunc}
\end{align}
\cref{eq:cf01} is plotted in \cref{fig:szilard-mod-data} in green.  Note that this function peaks both at $\theta\in \{-\pi,0,\pi\}$ (i.e., when the particle is in the left or right half of the box) as well as $\theta \in \{-\pi/2,\pi/2\}$ (i.e., when the particle is in the top or bottom half of the box) --- precisely those $\theta$ for which $p_\theta$ has no correlations between the horizontal and vertical position of the particle.

Next, we consider the extractable work bound in \cref{eq:projMbound2} for the initial distribution $\ptheta$. 
It can be verified that $\projG{\projM{\ptheta}}(x_1,x_2)=\ptheta(x_1)\po(x_2)$, which is illustrated  in \cref{fig:szilard-mod-angle}(a).  
The right hand side of \cref{eq:projMbound2} can again be expressed in closed form as (see \cref{app:szModDeriv})
\begin{align}
D(\projG{\projM{\ptheta}}\Vert \po) = \ln 2 -\frac{1}{2}\begin{cases}
f(|\tan \theta|) & \text{if $|\theta| \in \angleRange$}\\
|\tan \theta|
& \text{otherwise}
\end{cases}
\label{eq:cf02}
\end{align}
with $f$ defined as in \cref{eq:fffunc}. This result is shown in \cref{fig:szilard-mod-data} in orange. 
Note also that $\projG{\projM{\ptheta}}=\projM{\projG{\ptheta}}$ for all $p_\theta$, so the bounds in \cref{eq:projMbound2,eq:projMbound3} are equivalent.

For comparison  we also plot the extractable work bound derived using symmetry constraints, \cref{eq:accFreeEnergyRotated} (\cref{fig:szilard-mod-data} in blue). 
It is clear that the bound derived by exploiting a combination of modularity and symmetry constraints  (in orange) is strictly tighter  than the bounds derived by using either only modularity (green) or only symmetry constraints (blue) individually.

One can also use the bounds derived in this section to analyze the accessible information in a measurement of the Szilard box. Imagine that, starting from a uniform equilibrium distribution, one measures which side of the box contains the particle, as determined by a separating line at some arbitrary angle $\theta \in[-\pi,\pi]$. %
For this measurement, the conditional distribution over system states $\pconds$ is equal to $\ptheta$ half the time  and equal to $\pthetapi$ the other half the time. One can then derive bounds on accessible information such as \cref{eq:szwork0}, while using the bounds derived in this section (\cref{eq:projMbound,eq:projMbound2,eq:projMbound3}).

\subsection{Example: Generalized Szilard box}

\label{sec:generalizedszilard}

\begin{figure}[t!]
\includegraphics[height=0.3\columnwidth]{\figdir/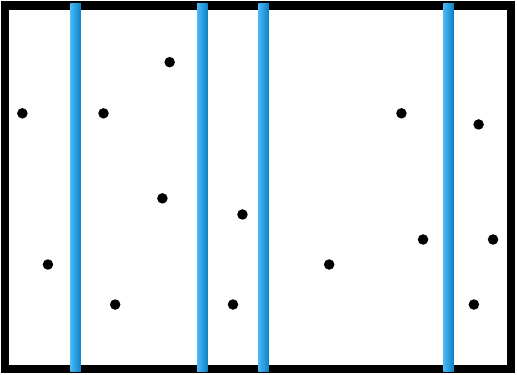}
\caption{
A generalized Szilard box with multiple particles~\cite{song2021optimal}.
\label{fig:genszilard}
}
\end{figure}

\newcommand{\numparticles}{N}

Our results on modularity constraints can be useful for analyzing the thermodynamics of multi-particle systems. %
As an example, consider the ``generalized Szilard box'' feedback-control scenario analyzed in \cite{song2021optimal}. Here, a box containing an ideal gas of $\numparticles$ particles, which are indexed by $v\in V$, begins in uniform equilibrium with a heat bath at inverse temperature $\beta$. Several partitions are inserted into the box, separating the box into separate volumes, and a measurement $M$ is made of the number of particles in each volume (see the illustration in \cref{fig:genszilard}). The box is then separated from the bath and, depending on the outcome of the measurement, the partitions are moved so as to equalize the pressure within each volume while extracting work. To make the process repeatable, suppose that at the end of the protocol, the partitions are removed and the box is again equilibrated with the bath (note that this last step does not contribute to extracted work). 

The ideal gas assumption means that the  particles do not interact, so by \cref{eq:m2,eq:m3} the protocol obeys modularity constraints with respect to a decomposition in which each particle is a separate subsystem.
The corresponding operator $\projMbase$ is given by %
\begin{align}
\projM{p}(x) = \prod_{v=1}^N p(x_v).
\label{eq:indepProjM}
\end{align}
Given \cref{eq:workboundinfo}, the average extractable work for the above feedback-control scenario is bounded by $\langle W\rangle \le\IaccS[\projMbase]/\beta$, which can also be written in terms of the %
information provided by the measurement $M$ about each individual particle,
\begin{align}
\langle W\rangle  \le   \sum_{v=1}^\numparticles I(X_v;M)/\beta,
\label{eq:songresult}
\end{align}
as follows from \cref{eq:modInfo2}. In fact, by symmetry of the initial distribution, the measurement provides the same information about each particle, $I(X_v;M)= I(X_1;M)$ for all $v$, so we can further rewrite \cref{eq:songresult} as $\langle W \rangle \le \numparticles\cdot I(X_1;M)/\beta$. 

This shows that %
\cref{eq:songresult}, which is reported as one of the main results of \cite{song2021optimal} (Eq.~5),   
follows  immediately from our framework. Moreover, our derivation holds under a broader set of conditions than those considered in \cite{song2021optimal}, since it does not rely on any of the details of setup (such as the type of partitions, the particular  work extraction protocol, or even the assumption that the particles are identical).

\subsection{Example: Collective flashing ratchet}
\label{sec:ratchetexample}

As a final example of modularity constraints, we consider the ``collective flashing ratchet'', a classic model in the literature on the thermodynamics of information~\cite{cao2004feedback,craig_feedback_2008}. This system involves $\numparticles$ overdamped particles evolving under an additive potential
\begin{align}
\label{eq:ratchetpotential}
E(x) = \lambda \sum_{v=1}^\numparticles  V(x_v).
\end{align} 
where $V$ is a single-particle potential and $\lambda \in \{0,1\}$ is a control parameter that can be used to turn the potential on/off. The single-particle  potential $V$ is chosen as an asymmetrical sawtooth ``ratchet'' pattern, shown in \cref{fig:ratchetpotential}, where $\alpha\in[0,1/2]$ parameterizes the degree of asymmetry. 

By manipulating  $\lambda$ over time, possibly in a way that depends on measurements of the system, the particles can be driven so as to have a net directional flux, or to do work against the externally applied force~\cite{feito_information_2007}. For instance, in a feedback control setup, $\lambda$ is determined by the outcome of some measurement $M$. The most common strategy involves turning the ratchet potential on when the net force on the particles is positive, and turning it off otherwise, according to the following measurement channel~\cite{cao2004feedback}:
\begin{align}
q(m\vert x) = \deltaFunc[m]\big[{\textstyle \Theta\big(\sum_v V'(x_v)\big)}\big],
\label{eq:ratchetMeasurement}
\end{align}
where  $\Theta$ is the Heaviside function. 
Note that this system has been experimentally realized~\cite{lopez_realization_2008}.

Suppose that starting from some initial distribution $p$, the measurement in \cref{eq:ratchetMeasurement} is performed. As common in the literature~\cite{cao2004feedback}, we assume 
that under $p$ the particles are identically and independently distributed, and that each particle is in the increasing part of the potential ($V'(x_v) \ge 0$) with probability $\alpha$ (see \cref{fig:ratchetpotential}). 
The measurement outcome is then used to drive the system back to distribution $p$ while extracting work by manipulating the system's energy function, all while coupled to a heat bath at inverse temperature $\beta$.
We assume that the driving protocols start and end on the same energy function, and that   only additive potentials (without interaction terms) are applied to the system during the driving (this  assumption allows for potentials such as \cref{eq:ratchetpotential}, as well as many others).  %

\begin{figure}[t!]
\includegraphics[height=0.3\columnwidth]{\figdir/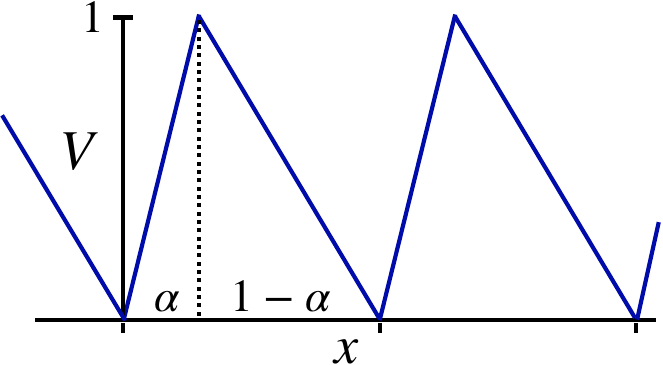}s
\caption{
The sawtooth potential of the flashing ratchet, from \cite{cao2004feedback}. 
\label{fig:ratchetpotential}}
\end{figure}

The driving protocols obey \cref{eq:modenergy2}  for a decomposition where each particle is its own subsystem, corresponding to the same type of $\projMbase$ as in \cref{eq:indepProjM}, $\projM{p}(x) = \prod_{v\in V} p(x_v)$. 
As in \cref{subsec:Szilard-box-example-mod}, we can use \cref{eq:workboundinfo} to bound average extractable work as $\langle W \rangle \le \IaccS[\projMbase]/\beta$. Using \cref{eq:modInfo2}, %
\begin{align}
\IaccS[\projMbase] = \sum_{v=1}^\numparticles I(X_v;M)=\numparticles  \cdot I(X_1;M),
\end{align}
where we've used  
the measurement provides the same information about each particle, $I(X_v;M)= I(X_1;M)$ for all $v$ (as  follows from a symmetry argument).

In  \cref{app:ratchet}, we show that $\IaccS[\projMbase]$ can be computed in closed form. %
Values of $\IaccS[\projMbase]$ for different values of $\numparticles$ (the number of particles) and $\alpha$ (the asymmetry parameter) are plotted in \cref{fig:ratchet}(left). 
Note that the accessible information shows a non-monotonic behavior in the number of particles for $\alpha \ne 0.5$. This occurs because for a highly asymmetric potential, the total amount of acquired information  grows with $N$: $I(X;M)$ grows from a minimum value of $h_2(\alpha)$ for $N=1$ to a maximum value of $\ln 2$ as $N\to\infty$. 
Given this observation, we also calculate the ``efficiency'' of the measurements in terms of the ratio $\IaccS[\projMbase]/\IXMs$. This is shown in \cref{fig:ratchet}(right) for various values of $\numparticles$ and $\alpha$. Interestingly, lower values of $\alpha$ (higher values of asymmetry) have higher efficiency values.

In the $\numparticles\to\infty$ limit, accessible information and efficiency converge to a single value, irrespective of $\alpha$. In \cref{app:ratchet}, we show that the accessible information $\IaccS[\projMbase]$ converges to  $1/\pi\approx 0.32$ nats, while the efficiency $\IaccS[\projMbase]/\IXMs$ converges to $1/(\pi\ln2)\approx 0.46$ (dotted lines in \cref{fig:ratchet}).

For a different (and complementary) theoretical analysis of extracted work in a feedback controlled flashing ratchet, see \cite{feito_information_2007}.

\begin{figure}[t!]
\includegraphics[valign=c,width=1\columnwidth]{\figdir/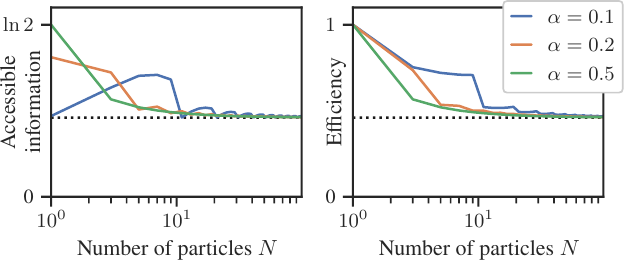}

\caption{
Left: accessible information $\IaccS[\projMbase]$ for the collective flashing ratchet, as a function of $\numparticles$ (number of particles) and $\alpha$ (asymmetry). Right: the efficiency of the measurements, $\IaccS[\projMbase]/\IXMs$. 
\label{fig:ratchet}}
\end{figure}

\section{Coarse-grained constraints}
\label{sec:cg}

In our final results section, we consider bounds on EP and work that arise from coarse-grained constraints.  

We begin by introducing some notation and preliminaries. Let $\cgf:\sX\to\CG$ be some coarse-graining
of the microscopic state space $\sX$, where $\CG$ is a set
of macrostates.  For any distribution $p$ over $\sX$, we use $\cgP(z)=\int \deltaFunc[\cgf(x)](z) p(x) \,dx$  
to indicate the corresponding distribution over the macrostates $\CG$,  
 $p_{X\vert Z}(x|z)=p(x)/\cgP(z)$ to indicate
the conditional probability distribution of microstates within macrostates, and $\cgDset := \{\cgP : p\in \Dset\}$ to indicate the set of all coarse-grained  distributions. Finally,
for any generator $\LL$ and distribution $p$, we use $[\LL p]_Z$  to indicate the resulting instantaneous dynamics of the coarse-grained distribution $\cgP$.

To derive our bounds, we suppose that the dynamics over the coarse-grained distributions are closed, i.e., %
for all $\LL\in\LLL$, 
\begin{align}
\cgP  = \cgP [q] \implies [L p]_Z = [L q]_Z 
\qquad\forall p,q\in\Dset.
\label{eq:closed0}
\end{align}
Given this assumption, the evolution of the coarse-grained distribution $\cgP $ can be represented by a coarse-grained  generator, which we write as $\ppt \cgP = \cgLL \cgP$ (discussed in detail below). 

We can specify more concrete conditions that guarantee that  \cref{eq:closed0} holds 
for a given generator $\LL$ (see \cref{appsec:cg} for details).  
For a discrete-state rate matrix $\LL$, it is satisfied when %
\begin{equation}
\sum_{{x:\cgf(x)=z}} \Lji={\cgLL}_{z,\cgf(x')}\quad\forall x',z\ne \cgf(x'),
\label{eq:closed0ME}
\end{equation}
where ${\cgLL}_{z,z'}$ is some coarse-grained transition rate from macrostate  $z'$ to macrostate $z$. 
\cref{eq:closed0ME} states that %
for each microstate $x'$, 
the total rate of transitions from $x'$ 
to microstates located in another macrostate $z \ne \cgf(x')$ 
 depends only on the macrostate $\cgf(x')$,  not on $x'$ directly. 
This condition has been sometimes called ``lumpability'' in the literature~\cite{nicolis2011transformation}. 

For a continuous-state master equation, \cref{eq:closed0} is satisfied when a continuous-state version of \cref{eq:closed0ME} (with sums replaced by integrals) holds. 
Moreover, for certain Fokker-Planck equation and linear coarse-graining functions, %
\cref{eq:closed0ME} can be replaced by a simple coarse-graining condition on the energy functions. %
Suppose  
each $\LL\in\LLL$ is a Fokker-Planck operator like
\begin{equation}
\LL p =\nabla\cdot(\nabla E_\LL) p +\dL\Delta p,\label{eq:cgFP}
\end{equation}
and that $\cgf$ is a linear function, $\cgf(x) = W x$ 
(where $W$ is some full-rank $m\times n$ matrix, $m\le n$). Without loss of generality, we assume that $W$ is scaled so that $WW^T=I$~\footnote{If \unexpanded{$\cgf(x)=Wx$} and \unexpanded{$WW^T\ne I$}, one can define an equivalent, rescaled coarse-graining function \unexpanded{$\cgf'(x)=W'x$}, where \unexpanded{$W':=(WW^T)^{-1/2}W$}, which obeys \unexpanded{$W'W'^T=I$}.}.
In addition, suppose that each energy function satisfies %
\begin{align}
W \nabla E_L(x)=-\cgEnergy(\cgf (x))\quad \forall x
\label{eq:closed0FP}
\end{align}
for some arbitrary macrostate drift function $\cgEnergy:Z\to \mathbb{R}$.
Then, the coarse-grained generator $\cgLL$  itself will have a Fokker-Planck form (see \cite{duong2018quantification} and \cref{appsec:cg}),
\begin{align}
\cgLL \cgP  = 
- \nabla\cdot\cgEnergy \cgP +\invBeta \Delta\cgP.\label{eq:cgfp2}
\end{align}
The right side of \cref{eq:cgfp2} depends only on $\cgP$ and not the full  microstate distribution $p$, so \cref{eq:closed0} will be satisfied.

Importantly, if \cref{eq:closed0} holds, the EP rate at time $t$ can be bounded as (see \cref{appsec:cg}):
\begin{align}
\EPr(p(t),L(t))\ge -\sum_{z} \ppt \cgP(z,t) \ln \frac{ \cgP(z,t)}{ \cgP[\pi]^{\LL(t)}(z)}
\ge0,%
\label{eq:macroNaineq}
\end{align}
where $\ppt \cgP(t) = \cgLL \cgP(t)$ and 
$\cgP[\pi]^{\LL(t)}$ is the coarse-grained version of $\pi^{\LL(t)}$, the stationary distribution of $\LL(t)$. %
The right hand side of \cref{eq:macroNaineq} is the coarse-grained version of \cref{eq:naBound}, %
which arises from the macrostate distribution $\cgP$ being out of equilibrium. We then define the total ``coarse-grained EP'' over the course of the protocol as the time integral of the middle term in \cref{eq:macroNaineq},
\begin{align}
\cgEPna(\dTrans{\cgP}{\cgPf})=\int_0^\ft -\sum_{z} \ppt \cgP(z,t) \ln \frac{ \cgP(z,t)}{ \cgP[\pi]^{\LL(t)}(z)}\; dt.
\label{eq:coarsegrainedEP}
\end{align}
Given \cref{eq:macroNaineq}, the coarse-grained EP serves as a non-negative lower bound on the total EP,
\begin{align}
\EP(\ptpp) \ge \cgEPna(\dTrans{\cgP}{\cgPf})\ge0. 
\label{eq:cgEP2}
\end{align}
Note that \cite{esposito2012stochastic} previously derived a coarse-grained EP rate for discrete-state master equations, which differs from the one that appears on the right hand side of \cref{eq:macroNaineq}; however, \cref{eq:macroNaineq} can be seen as the ``nonadiabatic component'' of the coarse-grained EP rate from~\cite{esposito2012stochastic}, and is thus a lower-bound on it~\cite{esposito2010three}.

We say that the available driving protocols obey \emph{coarse-grained constraints} %
if the generators $\LL \in\LLL$ exhibit closed  dynamics over $Z$, \cref{eq:closed0},  and there is some operator $\cgPhi : \cgDset \to \cgDset$ that obeys  the Pythagorean identity, \cref{eq:pyth},  and the commutativity relation, \cref{eq:comm0}, with respect to all $\cgLL$. For example, this coarse-grained operator $\cgPhi$ might reflect the presence of symmetry or modularity constraints on the coarse-grained dynamics. 

We can then use \cref{eq:cgEP2} and the framework developed in \cref{sec:info-geom-framework} to derive bounds on work and EP.  
In particular,  \cref{eq:epfreebound}  implies the following bound on coarse-grained EP,
$\cgEPna(\dTrans{\cgP}{\cgPf})\ge D(\cgP\Vert \cgPhi(\cgP)) - D(\cgPf\Vert \cgPhi(\cgPf))\ge0$.
Combined with \cref{eq:cgEP2}, this lets us bound overall EP as
\begin{align}
\EPp \ge D(\cgP\Vert \cgPhi(\cgP)) - D(\cgPf\Vert \cgPhi(\cgPf)) \ge 0.
\label{eq:epfreeboundcg}
\end{align}
Via \cref{eq:EPw}, this also gives a bound on extractable work like
\begin{multline}
\Wp \le \NFEs[\ps]-\NFEf[\pf] - \\ [D(\cgP\Vert \cgPhi(\cgP)) - D(\cgPf\Vert \cgPhi(\cgPf))]/\beta.\label{eq:workfreeboundcg}
\end{multline}
\cref{eq:epfreeboundcg,eq:workfreeboundcg} can also be used to derive bounds on average work extraction in feedback control protocols, using the strategy described in \cref{sec:thermo-of-info}.

If $\cgPhi$ represents coarse-grained symmetry or modularity constraints,  then \cref{eq:epfreeboundcg} implies that any 
asymmetry or inter-subsystem correlation in the macrostate distribution can only be dissipated away, not turned into work.
Another simple application %
occurs when all $\LL\in\LLL$ have the same coarse-grained equilibrium distribution, 
i.e., there is some $\cgP[\pi]$
such that $\cgLL\cgP[\pi]=0$ for all $\LL$. 
In this case, %
$\cgPhi(p)=\cgP[\pi]$ satisfies \cref{eq:pyth,eq:comm0} at the coarse-grained level (compare to the derivation of \cref{eq:e2} above). Applying \cref{eq:epfreeboundcg} then gives  %
\begin{align}
\EPp \ge D(\cgP \Vert \cgP[\pi]) - D(\cgPf\Vert \cgP[\pi]) \ge 0,
\label{eq:cgEPf}
\end{align}
as well as a corresponding extractable work bound, as in \cref{eq:workfreeboundcg}. 
This shows that if the coarse-grained equilibrium distribution $\cgP[\pi]$ cannot change, then any
deviation between the actual coarse-grained distribution $\cgP$ and $\cgP[\pi]$ must be dissipated
as EP, not turned into work.

\subsection{Example: Szilard box}
\label{subsec:Szilard-box-example-cg}

\begin{figure}
\begin{centering}
\includegraphics[width=0.5\columnwidth]{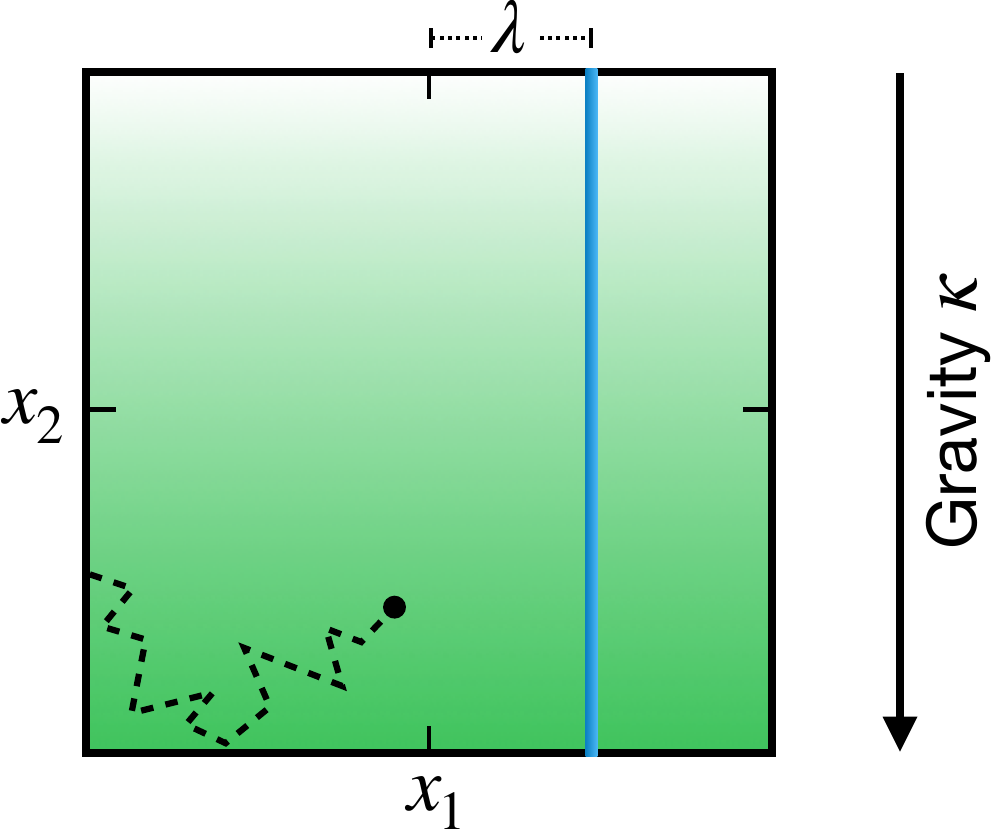}
\par\end{centering}
\caption{A two-dimensional Szilard box with a Brownian particle, in the presence
of gravity.\label{fig:szilard-grav}}
\end{figure}

We demonstrate our results on coarse-grained constraints using the Szilard box. 
We consider a similar setup as in \cref{subsec:Szilard-box-example,subsec:Szilard-box-example-mod}, where there is a  single overdamped particle in a box  coupled to a bath at inverse temperature $\beta=1$. 
However, we now assume  that there is a vertical gravitational force, as illustrated in \cref{fig:szilard-grav}. Formally, this means that the available potential  energy functions
have the form
\begin{equation}
E_{\paramA}(x_{1},x_{2})=\Vp(x_{1}-\paramA)+\Vw(|x_{1}|)+\Vw(|x_{2}|)+\grav x_{2},\label{eq:modh}
\end{equation}
where $\grav$ is a fixed constant that determines the strength of gravity. %
Unlike \cref{eq:ham0}, this energy function in \cref{eq:modh} no longer obeys
the reflection symmetry $(x_{1},x_{2})\mapsto(x_{1},-x_{2})$.

The microstate of the particle is represented by the horizontal and vertical position, $x=(x_1,x_2)$. We  consider a coarse-graining  in which %
the macrostate is the vertical coordinate of the particle $Z=X_2$, corresponding to the coarse-graining function $\cgf(x_1,x_2)=Wx=x_2$ with $W=[0\;1]$. 
It is easy to check that the potential energy functions in \cref{eq:modh} satisfy %
\begin{align}
W \nabla E_{\paramA}(x)=%
\partial_{x_{2}}[\Vw(|x_{2}|)+\grav x_{2}],
\label{eq:cgZ0}
\end{align}
which obeys \cref{eq:closed0FP} and therefore guarantees that the coarse-grained dynamics are closed.  
In fact, the coarse-grained generators have the Fokker-Planck form of \cref{eq:cgfp2} 
with the coarse-grained drift function $\cgEnergy(x_2) = -\partial_{x_{2}}[\Vw(|x_{2}|)+\grav x_{2}]$, which leads to the following Boltzmann stationary distribution:
\begin{align}
\pi_{X_2}(x_{2})&\propto\boltz{-\beta[\Vw(|x_{2}|)+\grav x_{2}]} \nonumber \\
&=\mathbf{1}_{\boxExtent}(x_{2})\boltz{-\beta\grav x_{2}},
\label{eq:boltz}
\end{align}
where in the second line we used the form of  $\Vw(\cdot)$ from \cref{eq:potential1}. 
Since the coarse-grained equilibrium distribution is the same for all energy functions having the form \cref{eq:modh}, we can use the EP bound in   
\cref{eq:cgEPf}.

\newcommand{\poCG}{\pi^\varnothing}
\newcommand{\eCG}{E^\varnothing}

Suppose that the system starts from some initial distribution $\ps$ and is then driven to a final equilibrium distribution $\pf$ while extracting work. We assume that the partition is removed at the beginning and end of the protocol, corresponding to the energy function $\eCG(x_1,x_2)=\Vw(|x_{1}|)+\Vw(|x_{2}|)+\grav x_{2}$, with the Boltzmann distribution
\begin{align}
\poCG(x_1,x_2) \propto \mathbf{1}_{\boxExtent^2}(x_1,x_{2})\boltz{-\beta\grav x_{2}}.
\end{align}
We will also assume that the final distribution is in equilibrium, so $\pf=\poCG$. Then, the extractable work involved in this transformation can be expressed as
\begin{align}
W(\dTrans{\ps}{\poCG}) &= \NFE[p][\eCG]-\NFE[\poCG][\eCG] -\EP(\dTrans{\ps}{\poCG})\nonumber\\
& =  D(p\Vert \poCG)-\EP(\dTrans{\ps}{\poCG}),\label{eq:ffd2}
\end{align}
where we used \cref{eq:EPw,eq:feffdecomp0}. We can then upper bound extractable work by combining \cref{eq:ffd2} with various lower bounds on $\EP(\dTrans{\ps}{\poCG})$.  

For instance, the second law states that $\EP(\dTrans{\ps}{\poCG})\ge0$, so
\begin{align}
W(\dTrans{\ps}{\poCG})  \le  D(p\Vert \poCG).\label{eq:ffd3}
\end{align}
We can also derive a stronger bound by exploiting coarse-grained constraints. For the coarse-graining described above, \cref{eq:cgEPf} implies that $\EP(\dTrans{\ps}{\poCG}) \ge D(p_{X_2} \Vert \pi_{X_2})$, which gives the bound
\begin{align}
W(\dTrans{\ps}{\poCG}) &\le D(p  \Vert \poCG)-D(p_{X_2} \Vert \pi_{X_2})\nonumber \\
& =D(p_{X_1\vert X_2} \Vert \poCG_{X_1\vert X_2}) .\label{eq:ffd4}
\end{align}
We can also bound EP and work using other kinds of constraints. For instance, the energy functions in \cref{eq:modh} have no interaction terms between $x_1$ and $x_2$, and therefore obey modularity constraints for the decomposition  $\modDecomp = \{\{X_1\},\{X_2\}\}$ (see the analysis in \cref{subsec:Szilard-box-example-mod}).  This allows us to bound EP and work using the operator $\projMbase$, as defined above in \cref{eq:twirlMsz}. 
In particular, using \cref{thm:resInt}, we have that
\begin{align}
\EP(\dTrans{p}{\poCG}) &= D(p\Vert \projM{p}) + \EP(\dTrans{\projM{p}}{\poCG})\label{eq:ffd5}\\
&\ge D(p\Vert \projM{p}).\nonumber
\end{align}
which implies the extractable work bound
\begin{align}
W(\dTrans{\ps}{\poCG}) &\le D(p\Vert \poCG) - D(p \Vert \projM{p})=D(\projM{p} \Vert \poCG) .
\label{eq:ffd6}
\end{align}

Finally, we can also combine modularity and coarse-grained constraints. The coarse-grained constraints implies that $\EP(\dTrans{\projM{p}}{\poCG})\ge D(\projM{p}_{X_2}\Vert \pi_{X_2})$ by \cref{eq:cgEPf}. Plugged into \cref{eq:ffd5}, this gives
\begin{align}
\EP(\dTrans{p}{\poCG}) \ge D(p\Vert \projM{p}) + D(\projM{p}_{X_2}\Vert \pi_{X_2}),
\end{align}
resulting in the extractable work bound
\begin{align}
W(\dTrans{\projM{p}}{\poCG}) \le D(\projM{p}_{X_1\vert X_2}\Vert \poCG_{X_1\vert X_2}),
\label{eq:ffd7}
\end{align}
where we've again used the chain rule of KL divergence.

\begin{figure}
\begin{centering}
\includegraphics[height=0.4\columnwidth]{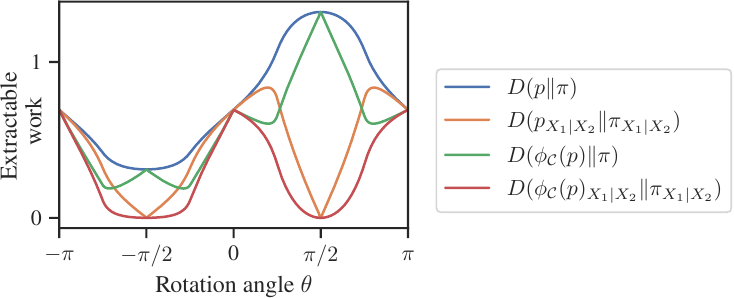}
\par\end{centering}
\caption{%
Szilard box with gravity: bounds on extractable work %
 as a function of $\theta$, as derived from the second law (in blue, \cref{eq:ffd3}), coarse-grained constraints (in orange, 
\cref{eq:ffd4}), modularity constraints (in green, \cref{eq:ffd6}), and a combination of modularity+coarse-grained constraints (in red, \cref{eq:ffd7}). %
\label{fig:szilard-cg-data}}
\end{figure}

We now illustrate these bounds using a concrete set of initial distributions.
Imagine that the initial distribution $p$ is the equilibrium distribution $\poCG$ restricted to half the box, as determined by a rotated separating line  at some angle 
$\theta \in [-\pi,\pi]$,
\begin{align}
\ptheta(x_1,x_2)= \frac{1}{2}\poCG(x_1,x_2)\Theta(x_2 \sin \theta - x_1 \cos\theta).
\end{align}
(Compare to \cref{eq:degMeas}, for the Szilard box without gravity). 
For these initial distributions and gravity parameter $\grav=1$, we plot the four extractable work bounds derived above, \cref{eq:ffd3,eq:ffd4,eq:ffd6,eq:ffd7}, as a function of $\theta$ in \cref{fig:szilard-cg-data} (values are calculated numerically). Note that, unlike the results presented in \cref{fig:szilard-symm-data,fig:szilard-mod-data}, the plots are no longer symmetric under the transformation $\theta \mapsto -\theta$. This arises because gravity breaks the vertical reflection symmetry, so the nonequilibrium free energy of a distribution concentrated on the top half of the box ($\theta=\pi/2$) is greater than the nonequilibrium free energy of a distribution concentrated on the bottom half of the box ($\theta=-\pi/2$). It can also be seen that work bounds derived from coarse-grained constraints, \cref{eq:ffd4} (orange), can be either weaker or stronger than the work bounds derived from modularity constraints, \cref{eq:ffd6} (green), depending on the value of $\theta$. For all $\theta$, however, the work bound derived by combining both constraints, \cref{eq:ffd7} (red), is stronger than the work bound derived from either constraint individually.

\section{Relevant literature}
\label{sec:priorwork}

In previous work on the general topic of {thermodynamic bounds under constraints},  Wilming et al.~\citep{wilming_second_2016} considered how extractable work depends on constraints
on the Hamiltonian, given a  quantum system coupled to
a finite-sized heat bath. That
paper derived an upper bound on the work that could be extracted by
carrying out a physical process which consists of sequences
of (1) unitary transformations of the system and bath, and (2) total
relaxations of the system to some equilibrium Gibbs state (see also a similar setup for closed systems in \cite{perarnau-llobet_work_2016}). 
Building on~\citep{wilming_second_2016}, \cite{lekscha_quantum_2018} analyzed the efficiency of a heat engine coupled to two baths and subject to ``local control'' constraints (i.e., a many particle system where local Hamiltonians can be changed but the interaction Hamiltonians cannot). %
In contrast to these works, 
we consider a classical system coupled to idealized reservoir(s). %
We then derive bounds on EP and work for a much broader set of
protocols.

At a high level,  our approach complements previous research on the relationship between EP, extractable work and different aspects of the driving protocol, such as temporal duration~\citep{esposito2010finite,sivak2012thermodynamic,shiraishi_speed_2018,gomez2008optimal,then2008computing,zulkowski2014optimal,schmiedl2007optimal},
stochasticity of control parameters \citep{machta2015dissipation},
non-idealized work reservoirs \citep{verley_work_2014}, cyclic protocols~\cite{schmiedl2007optimal,allahverdyan2004maximal}, the presence of additional conservation laws~\cite{uzdin2021passivity}, and the design of ``optimal protocols''~\citep{solon2018phase,gingrich2016near,aurell2011optimal}.

There is also previous work related to our analysis of {thermodynamics of information under constraints} in \cref{sec:thermo-of-info}.
\citep{still2020thermodynamic} recently analyzed the thermodynamics of feedback control under a somewhat different formulation of constraints~\footnote{That paper proposed to divide the system  into two subsystems $Y$ and $\bar{Y}$, such that the accessible information is given by $I(M;Y)$, under three assumption:  (1) the system's marginal distributions remains constant during all steps of feedback control, (2) the conditional distribution of $\bar{Y}$ given the system and the measurement does not change during the driving, and (3) all conditional information about $\bar{Y}$ is lost by the time that driving begins. After private communication with the author of \cite{still2020thermodynamic}, we think that condition (3) may need to be formalized as \unexpanded{$p(\bar{y}(t_2)\vert y(t_2), z(t_0))=p(\bar{y}(t_2)\vert y(t_2))$}, although this equation does not appear in that paper.}. 
In this work, we analyze the thermodynamics of information for a broader set of constraints. It is not immediately clear how the framework in \cite{still2020thermodynamic} compares to ours, or whether it could be applied to the examples considered in this paper, although such a comparison is an interesting direction for future work.

Some of our results concerning {work extraction under modularity constraints} in \cref{sec:modularity}
have appeared in prior literature. \cref{eq:ep-mod-mi} was derived in~\citep{Boyd:2018aa} 
for the special case of an isothermal processes with two non-overlapping subsystems, where one of the subsystems is held fixed. 
For the more general case of an arbitrary discrete-state system coupled to one or more reservoirs which have rate matrices as in \cref{eq:mcond}, 
\cref{eq:ep-mod-mi}  was also previously derived  in \cite{wolpert_thermo_comp_review_2019,wolpert2020thermodynamic}, while 
\cref{eq:ep-mod-tc} %
was previously derived in \cite{wolpert_thermo_comp_review_2019,wolpert.thermo.bayes.nets.2020,wolpert2020thermodynamic}.
Decompositions with overlap were previously considered in \cite{wolpert2020fluctuation,wolpert2020minimal}. %
In addition, Example\,1 in \cite{wolpert2020strengthened} can be used to derive the first inequality \cref{eq:ep-mod-condmi} for discrete-state systems~\footnote{The reader should be aware that those papers used different terminology
from this paper. In \cite{wolpert2020fluctuation,wolpert2020minimal}, each degree of freedom 
\unexpanded{$v \in V$} is called a ``subsystem'', the modular decomposition \unexpanded{$\modDecomp$} is called a ``unit structure'', while each \unexpanded{$\modSubsys\in\modDecomp$} is called a ``unit''.}. 

Those papers also derived some results that were more general than the ones derived here, in that they apply
even if the overlap changes state. 
Our paper goes beyond this previous work though to include continuous-state systems, and to derive inequalities such as $D(p\Vert \projM {p})-D(p'\Vert \projM{\pf})\ge 0$, albeit for the more
restricted scenario where the overlap does not change state. 

Some of our results concerning {work extraction under symmetry constraints}, presented in \cref{sec:Symmetry-constraints}, 
appeared in previous work on quantum thermodynamics. %
For a finite-state quantum system coupled to a work reservoir and heat bath, Vaccaro et al.~\cite{vaccaro_tradeoff_2008} investigated %
how much work can be extracted by bringing some initial quantum state $\rho$ to a maximally mixed state, with a uniform initial and final Hamiltonian, using discrete-time operations that commute with the action of some symmetry group $\G$. %
It was shown that the work that can be extracted from $\rho$ under such transformations is equal to the work that can be extracted from the (quantum) twirling $\projG \rho$, analogous to \cref{eq:workeq} for symmetry constraints. This research also derived an operational measure of asymmetry that is the quantum equivalent of $D(p\Vert\projG p)$,
and showed that asymmetry can only decrease under operations that commute with $\G$. 
Janzing~\cite{janzing_quantum_2006} extended \cite{vaccaro_tradeoff_2008} to consider arbitrary Hamiltonians, in the process deriving analogues of our decomposition of free energy (\cref{eq:feffdecomp}) for the special case of the twirling operator $\projGbase$.  A similar decomposition of free energy into coherent and incoherent components has recently appeared in \cite{lostaglio_description_2015,santos_role_2019} (this is a special case of the result in \cite{janzing_quantum_2006}, since a decohering map is a twirling operator~\cite{elphick2019spectral}). 
Finally, the idea of probability distributions that are invariant under symmetry groups, as well as a version of the twirling operator $\projGbase$, is a topic of research in probability and statistics; for details, see Ch.~3 in \cite{eaton_group_1989}.

While our approach is restricted to classical systems, in some respects our results for symmetry constraints are more general 
than this earlier work, since they hold for arbitrary (discrete and/or uncountably infinite) state spaces and for systems coupled to more than one reservoir (see \cref{sec:conclusion}). Moreover, for Fokker-Planck dynamics, we derive simple conditions for symmetry constraints stated in terms of the energy functions, which makes these results applicable to a large set of problems in stochastic thermodynamics and biophysics.

More fundamentally, one of the ways in which we go beyond previous literature on symmetry and modularity constraints is that by providing a unified mathematical framework that applies to a broad set of constraints, including symmetry, modularity, and coarse-grained constraints (as well as their combinations) as special cases. A key idea in our framework is that the information-geometric Pythagorean identity, \cref{eq:pyth}, is the essential property that allows an operator $\klopBase$ to uncover the thermodynamically accessible part of any distribution $p$ (assuming also that $\klopBase$ commutes with the dynamics). The Pythagorean identity is satisfied by many $\klopBase$, including both linear operators such as twirling operators $\projGbase$ and nonlinear operators such as modular decomposition operators $\projMbase$.  We believe this idea can be extended to the quantum domain, though we leave this for future work.

Finally, our approach is also related to ``resource theories'', which are an active area of research
in various areas of quantum physics~\cite{chitambar2019quantum}, including quantum thermodynamics~
\cite{wilming_second_2016,gallego_thermodynamic_2016,brandao_resource_2013,lostaglio_stochastic_2015,faist_fundamental_2018,yunger_halpern_beyond_2016}. 
A resource theory quantifies %
a physical resource %
in an operational way, in terms of what transformations are possible when the resource is available. 
Most resource theories are based on a common set of formal elements, such as a %
\emph{resource quantifier} (a real-valued function that measures the amount of a resource), a set of \emph{free states} (statistical states that lack the resource), and \emph{free operations} (transformations between statistical states that do not increase the amount of resource). 
In fact, some previous work on symmetry constraints in quantum thermodynamics~\cite{vaccaro_tradeoff_2008,janzing_quantum_2006} can be seen as part of a broader literature on the resource theory of asymmetry~\cite{marvian_extending_2014,marvian_asymmetry_2014,marvian_modes_2014}. 

Our approach has similar operational motivations as %
resource theories; for example, we define  ``accessible free energy'' in an operational way, as a quantity that governs extractable work under protocol constraints. %
Moreover, many elements of our framework %
are analogous to elements of the resource theory framework: the set of allowed generators (which we call $\LLL$) plays the role of the free operations, the image of the  operator $\klopBase$ plays the role of the set of free states, and the KL divergence $D(p\Vert \klop p)$ serves as the resource quantifier.  In addition,
the commutativity relation \cref{eq:comm0} (see \cref{sec:info-geom-framework}) has recently appeared in work on so-called %
resource destroying maps~\cite{liu2017resource}.  
However, unlike most resource theories, our focus is on the thermodynamics of classical systems modeled as driven continuous-time open systems. 
Further exploration of the connection between our approach and resource theories is left for future work.

\section{Discussion}
\label{sec:conclusion}

In this paper, we analyzed the EP and work incurred by a driving protocol that carries out some transformation $\ptpp$, while subject to constraints on the set of available generators. %
We constructed a general framework that allowed us  derive several decompositions and bounds on EP and extractable work, and demonstrated that this framework has implications for the thermodynamics of feedback control under constraints. Finally, we used our framework to analyze three broad classes of protocol constraints, reflecting symmetry, modularity, and coarse-graining.

Note that our bounds on EP and extractable work, %
such as \cref{eq:epfreebound,eq:ourworkbound}, are expressed in terms of state functions, i.e., they depend only on the initial and final distributions $\ps$ and $\pf$ and not on the path that the system takes in going from $\ps$ to $\pf$.  In general, it may  be possible to derive other %
bounds on work and EP that are not written in this form, which may be tighter. 
Nonetheless, bounds written in terms of state functions have some important advantages. 
In particular, they  allow one to quantify the  inherent ``thermodynamic value''  (in terms of EP and work) of a  distribution $p$ relative to a set of available generators, %
irrespective of what protocol brought the system there or what future protocols that system may undergo (as long as those protocols obey the relevant constraints).

For simplicity, our results were derived for isothermal protocols, where the system is coupled to a single heat bath at a constant inverse temperature $\beta$ and obeys local detailed balance (LDB).  Nonetheless, many of our results continue to hold for more general protocols, in which the system is coupled to any number of thermodynamic reservoirs and/or violates LDB. 
For a general protocol, our EP rate in \cref{eq:naBound} refers to the  so-called \emph{nonadiabatic EP rate}~\cite{van2010three,esposito2010three,lee_fluctuation_2013}, which is a non-negative quantity that reflects the %
contribution to EP that is due to the system being out of the stationary distribution. %
In the general case, our decompositions in \cref{thm:klinst,thm:resInt}, as well as EP lower bounds in \cref{eq:epfreebound,eq:EPmeas}, apply to nonadiabatic EP, rather than overall EP.  
Importantly, the nonadiabatic EP rate is a lower bound on the overall EP rate whenever the stationary distribution of $\LL$ is symmetric under conjugation of odd-parity variables~\cite{lee_fluctuation_2013}, which holds in most cases of interest such as discrete-state master equations (which typically have no odd variables), overdamped dynamics (which have no odd variables), and many types of underdamped dynamics. In such cases, \cref{eq:epfreebound,eq:EPmeas} provide lower bounds not only on  the nonadiabatic EP, but also on the overall EP, regardless of the number of coupled reservoirs or LDB.
However, the relationship between work and EP in \cref{eq:EPw}, as well as our bounds on work which make use of this relationship such as \cref{eq:workeq,eq:ourworkbound}, hold only for isothermal protocols.
Note that our EP bound for closed coarse-grained dynamics, \cref{eq:macroNaineq}, concerns the overall EP rate, \emph{not} the nonadiabatic EP rate, even for non-isothermal protocols (see \cref{subsec:cgNabound} for details).

There are several possible directions for future research.

First, it remains an open question of whether our framework can also be used to analyze other classes of constraints, beyond the three classes (symmetry, modularity, and coarse-graining) considered in this paper.

Second, our results point to a novel connection between entropy production, which plays a central role in nonequilibrium thermodynamics, and the Pythagorean identity in \cref{eq:pyth}, which plays a central role in  information geometry. This contributes to the growing number of existing results that demonstrate formal relationships between information geometry and nonequilibrium thermodynamics~\cite{ito2018stochastic,takahashi2017shortcuts,ito2020unified,nicholson2018nonequilibrium,ito2020stochastic,nakamura2019reconsideration}.  One direction for future work would be to extend the framework developed in this work for classical to quantum systems. In this extension, one would derive bounds on quantum work and EP by considering a quantum operator $\klopBase$ over density matrices which obeys quantum analogues of the Pythagorean identity in \cref{eq:pyth}~\cite[p.~44]{petzQuantumInformationTheory2008} and the commutativity relation in \cref{eq:comm0}.

Finally, our results may also lead to some new treatments of foundational questions in thermodynamics. In stochastic thermodynamics, probability distributions over system states are usually interpreted in  a ``subjective'' sense, in that the distribution $p$ assigned to a system typically reflects \emph{what one knows} about the system (for this reason, this distribution changes once  a measurement is made of the system's state~\cite{parrondo2015thermodynamics}). At the same time, our results show that for constrained driving protocols, one can often assign a different distribution to the system,  $\klop p$, which reflects \emph{what one can control} about the system. This also leads to the difference between the overall nonequilibrium free energy,  defined in terms of the distribution $p$, and the accessible free energy, defined in terms of the distribution $\klop p$. 
Note that thermodynamic entropy is often understood in an operational way, e.g., in terms of constrained  macroscopic control, as has been previously discussed by Jaynes~\cite{jaynes1992gibbs} and others.  An interesting direction for future work would explore whether the distinction between the distributions $p$ and $\klop p$ maps onto the distinction between (microscopic) statistical mechanical entropy and (macroscopic) thermodynamic entropy. In particular, one might ask whether this mapping can resolve some classic paradoxes concerning the relationship between statistical mechanical and thermodynamic entropy, 
such as the Gibbs paradox~\cite{jaynes1992gibbs} (mixing of indistinguishable particles increases statistical mechanical entropy but not thermodynamic entropy) and Loschmidt's paradox (for an isolated Hamiltonian system, statistical mechanical entropy remains constant while the thermodynamic entropy can increase). This direction could also be related to a recent axiomatic treatment of thermodynamic entropy which has been developed within the framework of quantum resource theory~\cite{weilenmann2016axiomatic}.

\section*{Acknowledgments}

We thank Massimiliano Esposito and Henrik Wilming for helpful discussions. This research
was supported by grant number FQXi-RFP-IPW-1912 from the Foundational
Questions Institute and Fetzer Franklin Fund, a donor advised fund
of Silicon Valley Community Foundation. The authors thank the Santa
Fe Institute for helping to support this research.

\noindent \bibliographystyle{IEEEtran}
\bibliography{refs}

\appendix
\crefalias{section}{appsec}
\crefalias{subsection}{appsubsec}

\input{appendix.tex}

\clearpage
\end{document}

%% file: commands.tex
\usepackage[T1]{fontenc}
\usepackage[latin9]{inputenc}
\usepackage{bm}
\usepackage{amsmath}
\usepackage{amsthm}
\usepackage{amssymb}
\usepackage{graphicx}
\usepackage{xargs}[2008/03/08]
\usepackage{stmaryrd}
\usepackage[hidelinks]{hyperref}
\usepackage{mathtools}
\usepackage{float}

\usepackage{enumitem}
\usepackage[export]{adjustbox}

\makeatletter
\theoremstyle{plain}
\newtheorem{thm}{\protect\theoremname}

\newtheorem{lem}{Lemma}
\newtheorem{prop}{Proposition}
\DeclareMathOperator*{\argmin}{\arg\,\min}

\usepackage[capitalize]{cleveref}

\crefname{appsec}{Appendix}{Appendices}
\crefname{appsubsec}{Appendix}{Appendices}

\global\long\def\IXMs{I(X;M)}%

\global\long\def\IXMf{I(X';M)}%

\newcommandx\deltaFunc[1][usedefault,addprefix=\global,1=x]{\delta_{#1}}

\newcommandx\IaccS[1][usedefault, addprefix=\global, 1=\klopBase]{I_\mathrm{acc}^{#1}(X;M)}%
\newcommandx\IaccF[1][usedefault, addprefix=\global, 1=\klopBase]{I_\mathrm{acc}^{#1}(X';M)}%

\newcommand{\pconds}{p_{X|m}}
\newcommand{\pcondf}{p_{X'|m}'}
\newcommand{\PCONDS}{p_{X|M}}
\newcommand{\PCONDF}{p_{X'|M}'}

\newcommandx\effCPs{\klop{\pconds}}%
\newcommandx\effCPf{\klop{\pcondf}}%
\newcommandx\effCPsM{\klop{\PCONDS}}%
\newcommandx\effCPfM{\klop{\PCONDF}}%
\newcommandx\effMIs{I_{\klopBase}(X;M)}%
\newcommandx\effMIf{I_{\klopBase}(X';M)}%

\newcommandx\IinaccS{D(\PCONDS \Vert\effCPsM)}
\newcommandx\IinaccF{D(\PCONDF \Vert\effCPfM)}
\newcommand{\Atheta}{A_\theta}

\newcommand{\ptheta}{p_\theta}
\newcommand{\pthetapi}{p_{\theta+\pi}}

\newcommand{\figdir}{.} 
\global\long\def\HorE{E}%
\global\long\def\ps{p}%
\global\long\def\Hs{\HorE}%
\global\long\def\pf{p'}%
\global\long\def\Hf{\HorE'}%

\global\long\def\traj{\bm{x}}
\global\long\def\rtraj{\tilde{\traj}}
\global\long\def\dTrans#1#2{#1\!\shortrightarrow\!#2}%

\global\long\def\pptp{\dTrans{\pf}{\ps}}
\global\long\def\ptpp{\dTrans{\ps}{\pf}}
\global\long\def\klopptpp{\dTrans{\klop \ps}{\klop \pf}}%

\global\long\def\ft{1}%
\global\long\def\EP{\Sigma}%
\global\long\def\EPr{\dot{\Sigma}}%
\global\long\def\EPp{\EP(\ptpp)}%
\global\long\def\Wp{W(\ptpp)}%

\global\long\def\LL{L}%

\global\long\def\klopBase{\phi}%
\global\long\def\klop#1{\klopBase(#1)}%

\global\long\def\ddt{{\textstyle \frac{d}{dt}}}%
\global\long\def\sX{X}%
\global\long\def\LLL{\Lambda}%
\global\long\def\LL{L}%
\newcommandx\sTrans[3][usedefault, addprefix=\global, 1=L]{#1_{#3#2}}%
\global\long\def\Lij{\sTrans{x}{x'}}%
\global\long\def\Lji{\sTrans{x'}{x}}%
\global\long\def\DLbase{\mathsf{D}}%

\global\long\def\uL{\mathsf{F}^{\LL}}%
\global\long\def\DL{\DLbase^{\LL}}%
\global\long\def\dL{\dLbase^{\LL}}%

\global\long\def\uL{\mathsf{A}}%
\global\long\def\DL{\DLbase}%
\global\long\def\dL{\beta^{-1}}%

\global\long\def\Dset{\mathcal{P}}%

\global\long\def\boltz#1{e^{#1}}%
\global\long\def\ppt{{\textstyle \partial_{t}}}%

\newcommandx\piL[1][usedefault, addprefix=\global, 1=L]{\pi^{#1}}%
\newcommand{\img}{\mathrm{img}\;}
\newcommand{\supp}{\mathrm{supp}\;}
\newcommand{\imgKLOP}{\img \klopBase}

\global\long\def\modDecomp{\mathcal{C}}%
\global\long\def\otherS{B}%
\newcommandx\LLA[1][usedefault, addprefix=\global, 1=A]{\LL^{#1}}%
\global\long\def\projMbase{\klopBase_{{\modDecomp}}}%
\global\long\def\projM#1{\projMbase(#1)}%

\global\long\def\modSubsys{A}%
\global\long\def\olaps{O(\modDecomp)}%
\global\long\def\olapsShort{O}%

\newcommand{\invBeta}{\beta^{-1}}

\newcommand{\lntwooverbeta}{(\ln 2)/\beta}

\newcommand{\inftime}{\tau}
\newcommandx\eLb[1][usedefault, addprefix=\global, 1=p, 1=\inftime]{e^{#1 L}}%
\newcommandx\eL[2][usedefault, addprefix=\global, 1=p, 2=\inftime]{\eLb[#2]#1}%

\newcommandx\NFE[2][usedefault, addprefix=\global, 1=\ps,2=\Hs]{F_{#2}(#1)} %
\newcommandx\NFEs[1][usedefault, addprefix=\global, 1=\ps]{\NFE[#1][\Hs]} %
\newcommandx\NFEf[1][usedefault, addprefix=\global, 1=\pf]{\NFE[#1][\Hf]}%

\def\cgMarker{\hat}
\global\long\def\cgLL{\cgMarker{\LL}}%
\global\long\def\cgPhi{\cgMarker{\klopBase}}%

\global\long\def\cgf{\xi}%
\global\long\def\cgEPna{\cgMarker{\EP}}%
\global\long\def\CG{Z}%
\newcommandx\cgP[1][usedefault, addprefix=\global, 1=p]{#1_\CG}%
\newcommandx\cgPf[1][usedefault, addprefix=\global, 1=p]{{#1_\CG'}}%
\global\long\def\z{z}%
\global\long\def\zz{z'}%
\newcommandx\cgUL[1][usedefault, addprefix=\global, 1=\LL]{\cgMarker{u}^{#1}}%
\global\long\def\cgDset{{{\Dset}_\CG}}%
\global\long\def\cgEnergy{\hat{F}}%

\global\long\def\Hunif{E^{\varnothing}}%
 
\global\long\def\po{u}%
\global\long\def\paramA{\lambda}
\newcommand{\boxLeft}{-1}
\newcommand{\boxRight}{1}
\newcommand{\boxExtent}{[\boxLeft,\boxRight]}

\global\long\def\grav{\kappa}%
\AtBeginDocument{

}

\usepackage{newtxtext}

\usepackage{color}%

\crefname{thm}{Theorem}{Theorems}
\crefname{cor}{Corollary}{Corollary}
\crefname{lem}{Lemma}{Lemmas}
\crefname{prop}{Proposition}{Propositions}
\crefname{defn}{Definition}{Definitions}
\crefname{rem}{Remark}{Remarks}

\usepackage{color}  %
\RequirePackage[dvipsnames,usenames]{xcolor}

\DeclareFontFamily{U}{mathx}{\hyphenchar\font45}
\DeclareFontShape{U}{mathx}{m}{n}{
      <5> <6> <7> <8> <9> <10>
      <10.95> <12> <14.4> <17.28> <20.74> <24.88>
      mathx10
      }{}
\DeclareSymbolFont{mathx}{U}{mathx}{m}{n}
\DeclareMathSymbol{\bigtimes}{1}{mathx}{"91}

\makeatother

\usepackage{babel}
\providecommand{\theoremname}{Theorem}

\usepackage{IEEEtrantools}

\newcommand{\ak}[1] {}

 	\renewcommand{\ak}[1] {{\color{Orange}{AK: #1}}}

%% file: appendix.tex

%
%

\section{Derivations for \cref{sec:info-geom-framework,sec:thermo-of-info}}
\label{app:theoretical}

\subsection{Proofs of \cref{thm:klinst,thm:resInt}}

We first prove a few helpful lemmas.

\begin{lem}
\label{lem:stlem}
If $\LL$ obeys $\eL[\klop p][ ] = \klop{ \eL[p][ ] }$ for all $p\in\Dset$, then $\LL$ has a stationary distribution $\pi\in\imgKLOP$.
\end{lem}
\begin{proof} %
Let $q$ be some stationary distribution of $\LL$.  Then,
\begin{align}
\eL[\klop q][ ] = \klop {\eL[q][ ]} = \klop q .
\end{align} 
Thus, $\klop q \in \imgKLOP$ is stationary under $\LL$. %
\end{proof}

\newcommand{\ppdist}{r}
\newcommand{\pppdist}{s}

\begin{lem}
\label{lem:contlem}
If $\eL[\klop p] = \klop {\eL[p] }$ for all $p\in\Dset$ and $\inftime \ge 0$, then for any $\ppdist,\pppdist \in\Dset$, 
\begin{align*}
-\ddt D(\ppdist(t) \Vert \klop{\pppdist(t)}) \ge 0,
\end{align*}
where $\ppt \ppdist =L\ppdist$ and $\ppt \pppdist= L\pppdist$.
\end{lem}
\begin{proof}
Expand the
derivative as %
\begin{align*}
 & -\ddt D(\ppdist(t)\Vert\klop{\pppdist(t)})\nonumber \\
 & \quad=\lim_{\inftime\to0}\frac{1}{\inftime}\left[D(\ppdist \Vert\klop {\pppdist})-D(\eL[\ppdist]  \Vert\klop{\eL[\pppdist]})\right]\nonumber \\
 & \quad=\lim_{\inftime\to0}\frac{1}{\inftime}\left[D(\ppdist\Vert\klop {\pppdist})-D(\eL[\ppdist] \Vert \eL[\klop \pppdist])\right]\ge0. %
\end{align*}
where in the last line we used the commutativity relation  %
and the data processing inequality for KL divergence~\cite{csiszar_information_2011}.
\end{proof}

\begin{lem} %
\label{lem:commProp}
Consider a protocol $\{\LL(t):t\in[0,\ft]\}$ and an operator $\klopBase$ that obeys \cref{eq:pyth,eq:comm0}. Then
\[\klop{p(t)}=\klop{p}(t),\]
where $p(t)$ is the distribution at time $t$ given initial distribution $p$, and $\klop{p}(t)$ is the distribution at time $t$ given initial distribution $\klop{p}$.
\end{lem}
\begin{proof}
Using \cref{lem:contlem} with $\ppdist=\klop{p}(t)$ and $\pppdist=p(t)$, %
\begin{align}
\ddt D(\klop{p}(t) \Vert \klop{p(t)}) \le 0.
\label{eq:appMM1}
\end{align}
Note that  
\[
D([\klop{p}](0) \Vert \klop{p(0)})=D(\klop{p}\Vert \klop{p})=0,
\]
and that $D(\klop{p}(t) \Vert \klop{p(t)})\ge 0$ for all $t$ by non-negativity of KL divergence. 
Combined with \cref{eq:appMM1}, this implies $D(\klop{p}(t) \Vert \klop{p(t)})=0$ for all $t$, 
and therefore $\klop{p}(t) = \klop{p(t)}$~\cite[Thm. 8.6.1]{cover_elements_2006}.  
\end{proof}

We are now ready to prove \cref{thm:klinst,thm:resInt}. Note that in the proof of \cref{thm:klinst}, we make the 
assumption that there is some stationary distribution $\pi^L$  of $L$ such that $D(p\Vert \pi^L)<\infty$, and similarly in \cref{thm:resInt} we make the assumption that $D(p(t)\Vert \pi^{\LL(t)})<\infty$ at all $ t\in[0,\ft]$. These are weak and physically realistic assumptions, which essentially mean that we restrict our attention to distributions with finite nonequilibrium free energy  (see \cref{eq:Fkl}).  

In addition, in these proofs we will use that the EP rate incurred by distribution $p$ under the generator $L$ with stationary distribution $\pi$ can be written as
\begin{align}
\EPr(p,L) & = %
\lim_{\inftime \to 0} \frac{1}{\inftime} \left[ D(p\Vert\pi) - D(\eL[p]\Vert\pi)\right]. \label{eq:app0}
\end{align}
This can be derived from \cref{eq:naBound}, by noting that  the KL divergence can be written as
\begin{align}
D(p\Vert\pi)=-S(p)-\mathbb{E}_p\big[\ln\pi \big],
 \label{eq:app0b}
\end{align}
where $\mathbb{E}_p$ indicates expectation under the distribution $p$, and then using that %
\begin{align}
&-\sum_{x} \ppt p_x(t) \ln p_x=\lim_{\inftime \to 0} \frac{1}{\inftime} \left[  S(\eL[p])-S(p)\right]
 \label{eq:app0c} \\
&\sum_{x} \ppt p_x(t) \ln \pi_x=\lim_{\inftime \to 0} \frac{1}{\inftime} \left[\mathbb{E}_{\eL[p]}[\ln \pi]- \mathbb{E}_p[\ln \pi]\right], \label{eq:app0d}
\end{align}
where $\ppt p_x(t)$ is defined as in \cref{eq:ll}. 
(As usual, summations should be replaced by integrals for continuous-state systems.)

\begin{proof}[Proof of \cref{thm:klinst}]
Consider a generator $\LL$ with a stationary distribution $\pi$, and some distribution $p\in\Dset$ such that $D(p\Vert \pi)<\infty$. 
By \cref{lem:stlem}, $\klop \pi\in\imgKLOP$ is also a stationary distribution of $\LL$. If $\LL$ has a unique stationary distribution, then $\pi=\klop \pi$ and so $\pi\in\imgKLOP$; otherwise, as long as $D(p\Vert \klop \pi)<\infty$ (see~\cite{Note3}), we can assume that $\klop \pi=\pi$ in \cref{eq:app0}. %
Then, assuming that $\pi\in\imgKLOP$, we rewrite the term in the brackets in \cref{eq:app0} as
\begin{align*}
& D(p\Vert \klop{p})+D(\klop{p}\Vert\pi)\nonumber \\
& \qquad\qquad- D(\eL[p]\Vert\klop{\eL[p]})-D(\klop{\eL[p]}\Vert \pi)
\\
& =D(p\Vert \klop{p})-D(\eL[p]\Vert\klop{\eL[p]}) \nonumber \\
& \qquad\qquad+D(\klop{p}\Vert\pi) -D(\klop{\eL[p]}\Vert \pi)
 \\
& =D(p\Vert \klop{p})-D(e^{\inftime L} p\Vert\klop{\eL[p]}) \nonumber \\
& \qquad\qquad+D(\klop{p}\Vert\pi) -D(\eL[\klop{p}] \Vert \pi),
\end{align*}
where we used the Pythagorean identity of \cref{eq:pyth}, rearranged, and then used the commutativity relation of \cref{eq:comm0}. %
Plugging into \cref{eq:app0} gives
\begin{align*}
\EPr(p,L)&=\lim_{\inftime \to 0} \frac{1}{\inftime} \left[ D(p\Vert \klop{p})-D(\eL[p] \Vert \klop{\eL[p]}) \right] \\
& \quad+ \lim_{\inftime \to 0} \frac{1}{\inftime} \left[ D(\klop{p}\Vert\pi) -D(\eL[\klop{p}] \Vert \pi) \right] \\
&=-\ddt D(p(t)\Vert \klop{p(t)}) + \EPr(\klop p,L).
\end{align*}
The non-negativity of $-\ddt D(p(t)\Vert \klop{p(t)})$ follows by taking $\ppdist=\pppdist=p$ in \cref{lem:contlem}.
\end{proof}

\begin{proof}[Proof of of \cref{thm:resInt}]
\newcommand\piLt{\pi^{L(t)}}
Using \cref{eq:totalEP,thm:klinst}, write %
\begin{align*}
&\EPp=\int_{0}^{\ft}\EPr(p(t),\LL(t))\,dt \\
&\;\;=-\int_{0}^\ft \ddt D(p(t)\Vert \klop{p(t)}) \,dt+\int_{0}^\ft \EPr(\klop {p(t)},L(t))\,dt.
\end{align*}
Both integrals have a simple expression. 
First, by the fundamental theorem of calculus, %
\begin{align*}
-\int_{0}^{\ft} \ddt D(p(t)\Vert\klop{p(t)})\,dt=D(p\Vert\klop p)-D(p'\Vert\klop{\pf}). %
\end{align*}
This expression is non-negative, since $-\ddt D(p(t)\Vert\klop{p(t)}) \ge 0$ by \cref{lem:contlem}.  Second, using  \cref{lem:commProp},
\begin{align*}
 \int_0^\ft \EPr(\klop{p(t)},L(t))\, dt&= \int_0^\ft \EPr(\klop{p}(t),L(t))\, dt\\
&\qquad= \EP(\dTrans{\klop{\ps}}{\klop{\pf}}) .
\end{align*}
\end{proof}

\subsection{Trajectory-level version of \cref{eq:bnd3}}
\label{app:fluct}

Stochastic thermodynamics has shown that thermodynamic properties of physical processes (such as heat, work, and EP) can be defined as stochastically fluctuating quantities at the level of individual trajectories. We first briefly review the basic concepts of stochastic thermodynamics (for more details, the reader should consult~\cite{van2015ensemble,seifert2012stochastic,Seifert2005,esposito_three_2010}).

Let $\traj = (x, \dots, x')$ indicate a continuous-time trajectory of system states $\traj$ over time interval $t\in[0,\ft]$, where $x$ and $x'$ indicate the initial and final system states respectively, and let $P(\traj\vert x)$ indicate the conditional probability of observing trajectory $\traj$ given initial state $x$. For a given initial distribution $p(x)$, the probability of observing trajectory $\traj$ is then given by $p(\traj)=p(x)P(\traj\vert x)$, and the corresponding final distribution is given by  $p'(x')=\int P(x'\vert x)p(x) dx$. In addition, let $\tilde{P}(\tilde{\traj}\vert x')$ indicate the conditional probability of observing the time-reversed and  trajectory $\tilde{\traj}=({x'},\dots,{x})$ given the final state ${x'}$ under a ``time-reversed'' driving protocol~\cite{seifert2012stochastic}. 

Trajectory-level EP  is then defined in terms of the asymmetry between forward and reversed trajectory probabilities,
\begin{align}
\label{eq:dft}
\sigma_{p}(\bm{x}) = \ln p(x) - \ln p'(x') + \ln \frac{P(\traj\vert x)}{\tilde{P}(\rtraj\vert x')},
\end{align}
which is sometimes referred to as a detailed fluctuation theorem. 
(The above expression should be slightly modified the presence of odd-parity variables such as momentum, though in a way which does not change our derivations; see~\cite{spinney2012entropy}.) 
The expectation of trajectory-level EP across all trajectories is equal to the standard expression for integrated EP as used in the main text,
\begin{align}
\langle \sigma_{p}(\traj) \rangle = \EP(\ptpp),
\end{align}
where $\langle \cdot\rangle$ refers to expectations under the trajectory distribution $p(\traj)$.  Furthermore, by a simple manipulation, the detailed fluctuation theorem in \cref{eq:dft} leads to the following integral fluctuation theorem for EP,
\begin{align}
\langle e^{-\sigma_{p}} \rangle &= \int_{p(x)>0} p(x) P(\traj\vert x)\frac{p'(x')\tilde{P}(\rtraj\vert x')}{p(x)P(\traj\vert x)}D\traj \nonumber \\
&=\int_{p(x)>0} p'(x')\tilde{P}(\rtraj\vert x') D\traj = \gamma,
\label{eq:app53}
\end{align}
where $\int \;\cdot \;D\traj$ is the path integral. In this result, $\gamma \in (0,1]$ reflects the ``absolute irreversibility'' of the process under initial distribution $p$~\cite{murashita2014nonequilibrium}. When $p$ has full support, $\gamma=1$, giving the standard integral fluctuation theorem, $\langle e^{-\sigma_{p}} \rangle=1$.

Now consider the extra trajectory-level EP incurred by some trajectory $\traj$ on initial distribution $p$, additional to the trajectory-level EP incurred by the same trajectory on initial distribution $\klop p$,
\begin{align}
m(\traj) &:= \sigma_{p}(\traj) -\sigma_{\klop p}(\traj) \label{eq:apptrajM0} \\
& = \ln \frac{p(x)}{ \klop {\ps}(x)} - \ln \frac{\pf(x')}{ \klop {\ps}'(x')}\label{eq:apptrajM1}\\
& = \ln \frac{p(x)}{ \klop {\ps}(x)} - \ln \frac{\pf(x')}{ \klop {\pf}(x')}\label{eq:apptrajM}
\end{align}
where in the second line we used that the last term in \cref{eq:dft} cancels (as it does not depend on the initial or final distributions) and in the third line we used that $\klop{\ps}'=\klop{\pf}$ by \cref{lem:commProp}.  \cref{eq:apptrajM0} appears in the main text as \cref{eq:fluctmismatch}.  It is  easy to verify that $m(\traj)$ agrees in expectation with the contraction of KL divergence between $p$ and $\klop p$,
\begin{align}
\langle m \rangle &=  D( \ps \Vert\klop \ps )-D(\pf \Vert\klop{\pf}),
\end{align}
where, as before, $\langle \cdot\rangle$ refers to expectations under the trajectory distribution $p(\traj)$. 
Then, given \cref{thm:resInt}, this implies that the expectation $m(\traj)$ is also equal to the extra total EP incurred by initial distribution $p$ rather than the accessible distribution $\klop p$,
\begin{align}
\langle m \rangle & =\EP(\ptpp) - \EP(\klopptpp).
\end{align}

In \cite{kolchinsky2021state}, it is shown that $m(\traj)$ obeys a fluctuation theorem (see also \cite{kwon2019fluctuation}). We re-derive the relevant results here. First, a simple rearrangement of \cref{eq:apptrajM1} gives the following detailed fluctuation theorem,
\begin{align}
m(\traj) &:= \ln \frac{p(x)}{\pf(x')} + \ln \frac{P(\traj\vert x)}{Q(\rtraj\vert x')},
\end{align}
where  the conditional distribution $Q(\rtraj \vert x')$ is given by 
\begin{align*}
Q(\rtraj \vert x') := \frac{ P(\traj\vert x)\klop \ps(x) }{\klop{\ps}'(x')}.
\end{align*}
In words, $Q(\rtraj \vert x')$ is the Bayesian posterior probability of trajectory $\bm$ given final state $x'$, when the process begins on initial distribution $\klop p$. 
A similar derivation as in \cref{eq:app53} shows that $m$ obeys an integral fluctuation theorem,
\begin{align}
\langle e^{-m} \rangle=\int_{p(x)>0} p'(x'){Q}(\rtraj\vert x') D\traj = \chi.
\label{eq:app54}
\end{align}
Here $\chi\in(0,1]$ indicates the absolute irreversibility of the process on initial distribution $p$ relative to initial distribution $\klop p$. $\chi$ is equal to 1 when $p$ and $\klop p$ have the same support, which then leads to a standard integral fluctuation theorem $\langle e^{-m} \rangle=1$. 

Importantly,  \cref{eq:app54} implies that the probability that the  trajectory-level EP on initial distribution $p$ is $\xi$ less than the trajectory-level EP on initial distribution $\klop p$ is exponentially suppressed,
\begin{align}
\mathrm{P}[  \sigma_{p}  < \sigma_{\klop p}-\xi  ] & \stackrel{(a)}{=}\mathrm{P}[  m  < -\xi  ] \stackrel{(b)}{\le}  \chi e^{-\xi} \stackrel{(c)}{\le } e^{-\xi}.
\end{align}
Here, $(a)$ uses the definition of $m(\traj)$, $(b)$ uses a standard derivation in stochastic thermodynamics (see \cite{jarzynski_equalities_2011}, or the appendix in \cite{kolchinsky2021state}), while $(c)$ uses that $\chi\in(0,1]$.

\section{Symmetry constraints}
\label{app:symm}

\subsection{$\projGbase$ obeys the Pythagorean identity, \cref{eq:pyth}}
\label{app:symmpyth}

In the following derivations, all integrals should be understood in
the Lebesgue sense. For discrete state systems, integrals over $X$
can be replaced by summations. 

The state space $X$ is assumed to be Borel measurable. Similarly,
we assume that the action of the group $\G$ (i.e., the function $\G\times X\to X:(g,x)\mapsto g(x)$)
is Borel measurable. Note that these assumptions imply that for any
probability distribution $p\in\Dset$, the function $(g,x)\mapsto p(g(x))$
is measurable, since it is the composition of two Borel measurable
functions: $(g,x)\mapsto g(x)$ and $x\mapsto p(x)$.

We begin with a few intermediate results.
\begin{lem}
\label{lem:inv1} For any $p\in\Dset$, $g\in\G$, and $x\in X$,
\[
\projG p(x)=\projG p(g(x)).
\]
\end{lem}

\begin{proof}
Using the definition of $\projGbase$ in \cref{eq:symmOop}, write
\begin{align*}
\projG p(g(x)) & =\int_{\G}p(g'(g(x)))\,d\mu(g')\\
 & ={\textstyle \int_{\G}}\,p(g'(x))\,d\mu(g')=\projG p(x),
\end{align*}
where we performed a change of variables $x\mapsto g^{-1}(x)$ and
used the invariance properties $\G$ and the Haar measure $\mu$.
\end{proof}
\begin{lem}
\label{lem:inv4}For any $p\in\Dset$, measurable set $\Omega\subseteq X$,
and function $f:X\to\mathbb{R}$,
\begin{equation}
\int_{\Omega}p(x)f(x)=\int_{\Omega}\projG p(x)f(x)dx\label{eq:gdh3}
\end{equation}
if the following three conditions hold: (1) $g(\Omega)=\Omega$ for
all $g\in\G$, (2) $f(x)=f(g(x))$ for all $x\in X$ and $g\in\G$,
(3) either $\vert\int_{\Omega}p(x)f(x)\,dx\vert<\infty$, or $f$
is measurable and non-negative.
\end{lem}

\begin{proof}
To begin, write the left hand side of \cref{eq:gdh3} as
\begin{align}
\int_{\Omega}p(x)f(x)\,dx & =\int_{\G}\left[\int_{\Omega}p(x)f(x)\,dx\right]d\mu(g)\nonumber \\
 & =\int_{\G}\left[\int_{g^{-1}(\Omega)}p(g(x))f(g(x))\,dx\right]d\mu(g)\nonumber \\
 & =\int_{\G}\left[\int_{\Omega}p(g(x))f(x)\,dx\right]d\mu(g).\label{eq:int01-1}
\end{align}
In the second line, we substituted $x\mapsto g(x)$ within each inner
integral, while using that each $g$ is a rigid transformation (so
the absolute value of its Jacobian is 1). In the last line, we used
conditions \emph{(1)} and \emph{(2)}.

We now show that we can exchange the order of integrals in \cref{eq:int01-1}
using condition \emph{(3)} and Tonelli's theorem. First, if $f$ is
measurable and non-negative, then the function $x\mapsto p(g(x))f(x)$
is non-negative and measurable (since it is a product of two non-negative
measurable functions), so the integrals can be exchanged by \citep[Thm 3.7.7, ][]{benedetto_integration_2009}.
Alternatively, assume that $\vert\int_{\Omega}p(x)f(x)\,dx\vert<\infty$,
which means that the function $x\mapsto p(x)f(x)$ is integrable.
This implies that
\begin{align}
\infty & >\int_{\Omega}p(x)\vert f(x)\vert\,dx\nonumber \\
 & =\int_{\G}\left[\int_{\Omega}p(x)\vert f(x)\vert\,dx\right]d\mu(g)\\
 & =\int_{\G}\left[\int_{g^{-1}(\Omega)}p(g(x))\vert f(g(x))\vert\,dx\right]d\mu(g)\nonumber \\
 & =\int_{\G}\left[\int_{\Omega}p(g(x))\vert f(x)\vert\,dx\right]d\mu(g)\label{eq:int01-2}
\end{align}
where the first line follows from definition of Lebesgue integrability,
while the rest follows from the same steps as \cref{eq:int01-1}. Given
\cref{eq:int01-2}, the function $(g,x)\mapsto p(g(x))f(x)$ must be
integrable, which again allows us to exchange the order of the integrals
in \cref{eq:int01-1} \citep[Thm 3.7.8,  ][]{benedetto_integration_2009}.

We then derive our result by rewriting \cref{eq:int01-1} as 
\begin{align*}
\int_{\Omega}p(x)f(x)\,dx & =\int_{\Omega}\left[\int_{\G}p(g(x))f(x)\,d\mu(g)\right]dx\\
 & =\int_{\Omega}\projG p(x)f(x)\,dx,
\end{align*}
where we used the definition of $\projGbase$.
\end{proof}
Finally, we prove that $\projGbase$ obeys the Pythagorean identity.
\begin{prop}
\label{prop:projGpyth} For any $p,q\in\Dset$ such that $D(p\Vert\projG q)<\infty$,
\begin{align}
D(p\Vert\projG q)=D(p\Vert\projG p)+D(\projG p\Vert\projG q).\label{eq:pythsymm}
\end{align}
\end{prop}

\begin{proof}
For any $p\in\Dset$, we indicate the support set as $\supp p=\{x\in X:p(x)>0\}$.
We first prove that
\begin{equation}
\supp p\subseteq\supp{\projG p}\subseteq\supp{\projG q}.\label{eq:suppinc}
\end{equation}
By the definition of $\projGbase$ in \cref{eq:symmOop}, if $\projG p(x)>0$
for some $x\in X$, then $p(g(x))>0$ for that $x$ and some $g\in\G$.
In addition, the assumption that $D(p\Vert\projG q)<\infty$ implies
that $\supp p\subseteq\supp{\projG q}$ \citep{cover_elements_2006}
(except for a set of measure 0, which we can safely ignore). Combining
these facts implies that if $\projG p(x)>0$ for some $x$, then $\projG q(g(x))>0$
for that $x$ --- and therefore also $\projG q(x)>0$ since $\projG q$ is invariant 
under $\G$, \cref{lem:inv1}. This proves that $\supp{\projG p}\subseteq\supp{\projG q}$.
Finally, by \cref{lem:inv1} and \cref{lem:inv4}, 
\[
\int_{\supp{\projG p}}p(x)\,dx=\int_{\supp{\projG p}}\projG p(x)\,dx=1,
\]
which implies that $\supp p\subseteq\supp{\projG p}$ (up to a set
of measure 0).

Next, write the KL divergence on the left hand side of \cref{eq:pythsymm}
as \citep[Eq.  8.58, ][]{cover_elements_2006}
\begin{align}
 & D(p\Vert\projG q)=\int_{\supp p}p(x)\ln\frac{p(x)}{\projG q(x)}dx\nonumber \\
 & \quad=D(p\Vert\projG p)+\int_{\supp p}p(x)\ln\frac{\projG p(x)}{\projG q(x)}dx\nonumber \\
 & \quad=D(p\Vert\projG p)+\int_{\supp{\projG p}}p(x)\ln\frac{\projG p(x)}{\projG q(x)}dx,\label{eq:gf2}
\end{align}
where the last line uses \cref{eq:suppinc} (in particular, that $\supp p\subseteq\supp{\projG p}$
and $p(x)\ln\frac{\projG p(x)}{\projG q(x)}=0$ for $x\in\supp{\projG p}\setminus\supp p$). 

The integral in \cref{eq:gf2} is bounded from above by $D(p\Vert\projG q)<\infty$,
since $D(p\Vert\projG p)\ge0$. We also show that this integral is
bounded from below. Note that $\projG p(x)$ and $\projG q(x)$ are
both non-negative measurable functions, which follows from the fact
that $x\mapsto p(g(x))$ and $x\mapsto p(g(x))$ are non-negative
measurable functions, the definition of $\projGbase$, and Tonelli's
theorem \citep[Thm 3.7.7, ][]{benedetto_integration_2009}. Thus,
the function $x\mapsto\frac{\projG q(x)}{\projG p(x)}$ is also non-negative
and measurable, letting us bound the integral in the following way:
\begin{align*}
 & \int_{\supp{\projG p}}p(x)\ln\frac{\projG p(x)}{\projG q(x)}dx\\
 & \qquad\ge-\ln\left[\int_{\supp{\projG p}}p(x)\frac{\projG q(x)}{\projG p(x)}dx\right]\\
 & \qquad=-\ln\left[\int_{\supp{\projG p}}\projG p(x)\frac{\projG q(x)}{\projG p(x)}dx\right]\\
 & \qquad=-\ln\left[\int_{\supp{\projG p}}\projG q(x)\,dx\right]\ge-\ln1=0.
\end{align*}
where in the second line we used Jensen's inequality, while in the
third line we applied \cref{lem:inv4}. Finally, we use \cref{lem:inv4}
to rewrite the integral in \cref{eq:gf2} as
\begin{multline*}
\int_{\supp{\projG p}}p(x)\ln\frac{\projG p(x)}{\projG q(x)}dx=\\
\int_{\supp{\projG p}}\projG p(x)\ln\frac{\projG p(x)}{\projG q(x)}dx=D(\projG p\Vert\projG q).
\end{multline*}
\end{proof}

\subsection{$\projGbase$ obeys the commutativity relation, \cref{eq:comm0}}

It is easy to verify that $\transP$  is a linear operator. It then follows that if $\transP$ commutes with the linear operator $\LL$, as in \cref{eq:commsymmZ}, then it also commutes with the exponential $e^{\inftime \LL}=\sum_k \frac{1}{k!}\inftime^k \LL^k$. We then have
\begin{align*}
\eL[\projG p]& = \eL[\int \transP p \,d\mu(g)]\\
&=\int \eL [\transP p] \,d\mu(g)\\
&=\int \transP \eL [p]  \,d\mu(g)\\
&= \projG {\eL[p]}
\end{align*}
where in the second line we exchanged the bounded operator $e^{\inftime\LL}$ and the (Bochner) integral, and in the third line we used that $\transP$ and $e^{\inftime \LL}$ commute.

\subsection{Derivation of \cref{eq:commsymmZ} from \cref{eq:symmME} and \cref{eq:symmFP}}

Consider some $f :\sX\to\mathbb{R}$ and a continuous-state master equation $\LL$ such that
\begin{align}
[\LL f](x)=\int \left[\Lji f(x')-\Lij f(x)\right]\,dx'.
\label{eq:cme}
\end{align}
(The derivation for discrete-state master equations, as in \cref{eq:ll}, is the same, but with integrals replaced with summations). Then,
\begin{align}
&[\transP\LL f](x)	=[\LL f](\actg(x)) \nonumber \\
	&\quad=\int [\LL_{\actg(x)x'}f(x')-\LL_{x'\actg(x)}f(\actg(x))]dx'\label{eq:appY0}\\
	&\quad=\int [\LL_{\actg(x)\actg(x')}f(\actg(x'))-\LL_{\actg(x')\actg(x)}f(\actg(x))]dx'\label{eq:appY1}\\
	&\quad=\int [\LL_{x x'}f(\actg(x'))-\LL_{x' x}f(\actg(x))]dx'\label{eq:appY2}\\
	&\quad= \int [L_{xx'}[\transP f](y)-L_{x'x}[\transP f](x)]dx'\label{eq:appY3} \\
	&\quad=[\LL\transP f](x),
\end{align}
which implies $\transP \LL = \LL\transP$, \cref{eq:commsymmZ}. Here we used the definition of $\transP$ in the first line and \cref{eq:cme} in \cref{eq:appY0}. In \cref{eq:appY1}, we used the variable substitution $x' \mapsto \actg(x')$, along with the fact that $\actg$ is volume preserving.  In \cref{eq:appY2}, we used \cref{eq:symmME}. %

\global\long\def\pgg{\transP f}
\global\long\def\pggP{(\transP f)}

Next, we show that \cref{eq:symmFP} is sufficient for \cref{eq:commsymmZ} 
to hold, assuming that all $\actg\in\G$ are rigid transformation and the $\LL\in\LLL$ refer to Fokker-Planck equations
of the form \cref{eq:fp10}. First, given some (sufficiently smooth) function $f : \sX\to\mathbb{R}$, 
write \cref{eq:fp10} as
\begin{equation}
\ppt f = \LL f = \nabla\cdot ((\nabla E) f)+\dL\Delta f.\label{eq:appfp2}
\end{equation}
For any $g\in\G$, write the diffusion term in \cref{eq:appfp2}
as
\begin{align}
\Delta f =\Delta(\pgg\circ \actgInv)= \Delta\pggP \circ \actgInv,\label{eq:zzm0}
\end{align}
where we used the identity $f = \transP[\actgInv] \pgg = \pgg\circ \actgInv$ and that the Laplace
operator commutes with rigid transformations. Now consider the
drift term in \cref{eq:appfp2}. Using the product rule, %
\begin{equation}
\nabla\cdot((\nabla E) f)=(\nabla f)^{T}(\nabla E)+f \Delta E.\label{eq:zmm}
\end{equation}
We can rewrite the second term above as
\begin{align}
f \Delta E  & =(\pgg \circ \actgInv) \Delta E \nonumber \\
& =(\pgg \circ \actgInv) \Delta(E\circ \actgInv) \nonumber \\
&=(\pgg \circ \actgInv)((\Delta E)\circ \actgInv) \nonumber \\
&= (\pggP(\Delta E)) \circ \actgInv ,\label{eq:zmm2}
\end{align}
where we used $f = \pgg \circ \actgInv$, 
the invariance of $E$ under $\G$ (\cref{eq:symmFP}), and in the third line that the Laplace operator commutes with rigid transformations. 
Now consider the first term on the right hand side of \cref{eq:zmm}:
\begin{align}
(\nabla f)^{T}(\nabla E)& =(\nabla(\pgg\circ \actgInv)^{T}\nabla(E\circ \actgInv)\nonumber \\
 & =(J^{T}(\nabla\pggP\circ \actgInv))^{T}(J^{T}((\nabla E)\circ \actgInv)) \nonumber \\
 & =(\nabla\pggP \circ \actgInv)^{T}J J^{T}((\nabla E) \circ \actgInv)\nonumber \\
 & =(\nabla\pggP \circ \actgInv)^{T}((\nabla E)\circ \actgInv)\nonumber \\
 & =(\nabla\pggP^{T}(\nabla E)) \circ \actgInv , \label{eq:zmm3}
\end{align}
where $J$ indicates the Jacobian of $\actgInv$.  In the first line, we again used the identity $f = \pgg \circ \actgInv$ and the invariance of $E$ under $\G$, in the second line we used the chain rule, and in the fourth line we used that $J J^T = I$ for rigid transformations. 
Plugging \cref{eq:zmm2,eq:zmm3} back into \cref{eq:zmm} and rearranging gives
\begin{align}
\nabla\cdot((\nabla E) f) = \nabla\cdot((\nabla E) \pggP) \circ \actgInv.
\end{align}
Combined with \cref{eq:zzm0,eq:appfp2}, this in turns implies that $\LL f=(\LL \pgg) \circ \actgInv$, or in other words that
\[
\transP \LL f =\LL \transP f.
\]

\subsection{Derivation of \cref{eq:decomp4b}}
\label{app:symmthermoinfo}

First, write the inaccessible information term in \cref{eq:accinfodef} as
\begin{align}
&D(\PCONDS\Vert\projG{\PCONDS}) = \sum_m p(m) D(\pconds\Vert\projG{\pconds}) \nonumber \\
&   =\sum_m p(m,x) \ln\frac{p(x\vert m)}{\int p(g(x)\vert m)\mu g} \nonumber \\
&	=\sum_m p(m,x) \ln\frac{p(x)q(m\vert x)/p(m)}{\int p(g(x))q(m\vert g(x))/p_{g}(m)\,\mu(g)},\label{eq:appG877}
\end{align}
where we've defined $p(m)=\sum_x p(x)q(m\vert x)$ and  $p_g(m)=\sum_x p(g(x))q(m\vert x)$, and used the definition of $\projGbase$ in \cref{eq:symmOop}. (Here we assume for simplicity that both  $X$ and $M$ are discrete valued;  otherwise the summations in \cref{eq:appG877} should be replaced with integrals.) 

Recall that we assumed that $p$ is invariant under $\G$, so $\projG p=p$. By \cref{lem:inv1},  $p(x)=p(g(x))$ for all $x$ and $g\in\G$, which in turn implies that $p(m)=p_g(m)$.  Plugging into \cref{eq:appG877} then gives
\begin{align*}
D(\PCONDS\Vert\projG{\PCONDS})=\sum_m p(m,x) \ln\frac{q(m\vert x)}{\int  q(m\vert g(x))\,\mu(g)},
\end{align*}
which appears in the main text as  \cref{eq:decomp4b}.

\subsection{Example: Szilard box, derivation of \cref{eq:accFreeEnergyRotated}}
\label{app:rotatedaccessDeriv}

\newcommand{\darkgraymeas}{P_{\theta}}

We derive \cref{eq:accFreeEnergyRotated} using a simple geometric argument.

Consider the twirling of $\ptheta$, as shown in \cref{fig:szilard-symm-angle}(b).  From the definition of $\projGbase$ and \cref{eq:degMeas}, it is easy to see that
 \begin{enumerate}
 	\item The dark gray areas in \cref{fig:szilard-symm-angle}(b) (where both $\ptheta(x_1,x_2)=1/2$ and $\ptheta(x_1,-x_2)=1/2$) have probability density $\projG{\ptheta}(x_1,x_2)=1/2$.
 	\item The light gray areas in \cref{fig:szilard-symm-angle}(b) (where either $\ptheta(x_1,x_2)=1/2$ or $\ptheta(x_1,-x_2)=1/2$, but not both) have probability density $\projG{\ptheta}(x_1,x_2)=1/4=u(x_1,x_4)$.
 	\item The white areas in \cref{fig:szilard-symm-angle}(b) (where $\ptheta(x_1,x_2)=0$ and $\ptheta(x_1,-x_2)=0$) have probability density $\projG{\ptheta}(x_1,x_2)=0$. 
 \end{enumerate}
 Given this, 
\begin{align}
D(\projG{\ptheta} \Vert \po)= \ln 2 \cdot \darkgraymeas,
\label{eq:appnf4}
\end{align}
where $\darkgraymeas$ is the probability assigned by $p$ to the dark gray areas (i.e., those $(x_1,x_2)$ where $\ptheta(x_1,x_2)=1/2=\ptheta(x_1,-x_2)=1/2$).

\begin{figure}[b]
\begin{centering}
\includegraphics[width=0.75\columnwidth]{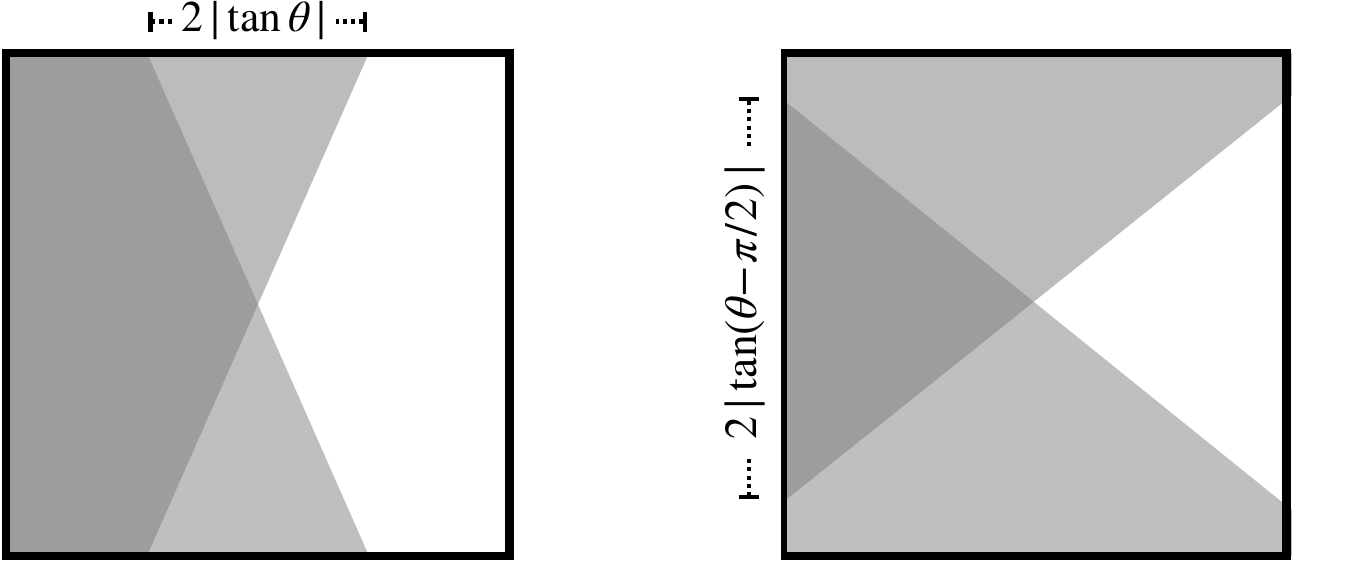}
\par\end{centering}
\caption{%
The twirling  $\projM{\ptheta}$ for two cases. Left: $\vert\theta\vert\in\angleRange$.
Right:  $\projM{\ptheta}$ for $\vert\theta\vert\in[-\pi,\pi]\setminus\angleRange$. 
\label{fig:szilard-two-cases}}
\end{figure}

To calculate the value of $\darkgraymeas$, is suffices to consider two separate cases:
\begin{enumerate}
	\item $\vert\theta\vert\in[-\pi,\pi]\setminus\angleRange$
	\item $\vert\theta\vert\in\angleRange$
\end{enumerate}
which are shown visually in \cref{fig:szilard-two-cases}. 
Using this figure, and a bit of trigonometry, it can be shown that $\darkgraymeas=1-\frac{1}{2}\vert \tan \theta\vert$ in the first case, and  $\darkgraymeas=\frac{1}{2}\vert \tan(\theta - \pi/2)\vert$ in the second case.  Combining these results  with \cref{eq:appnf4} gives \cref{eq:accFreeEnergyRotated}.

\subsection{Example: Symmetry constraints on a discrete-state master equation}
\label{app:unicyclic}

Here we  demonstrate  our results on symmetry constraints %
using a simple
finite-state system.  The system contains $n$ states,  $x=\{0,\dots,n-1\}$. 
We consider a group generated by circular shifts, representing  $m$-fold circular symmetry:
\begin{align}
\actg(x)=x+n/m\quad\mathrm{mod}\quad n.
\end{align}
Assume that  the driving protocol obeys the following symmetry group at all $t\in[0,\ft]$:
\begin{equation}
\Lij(t)=L_{\actg(x')\actg(x)}(t),\label{eq:symmRM}
\end{equation}
An example of such
a master equation would be a unicyclic network, where the $n$ states
are arranged in a ring, and transitions between nearest-neighbor
states obey \cref{eq:symmRM}. Such unicyclic networks are often used to model biochemical
oscillators and similar biological systems \citep{barato2017coherence}.
This kind of system is illustrated in \cref{fig:symm}, with $n=12$ and $m=4$. 

Imagine that this system starts from the initial distribution $p(x)\propto x$,
so the probability  grows linearly from 0 (for $x=0$) to maximal
(for $x=n$). For the 12 state system with 4-fold symmetry, this initial distribution
is given by 
\[
p(x) = \frac{x}{\sum_{x'=0}^{11} x'} = \frac{x}{66},
\]
and is shown on the left hand side of \cref{fig:symm}. How much work can
be extracted by bringing this initial distribution to some other
distribution $p'$, while using rate matrices of the form \cref{eq:symmRM}? %
This is bounded by the drop of the accessible free energy, via \cref{eq:ourworkbound}:
\begin{align}
W(\ptpp) \le \NFEs[\projG{\ps}]-\NFEf[\projG{\pf}].
\end{align}
Using the example system with 12 states
and 4-fold symmetry, the twirled distribution $\projG{\ps}$ is given by
\begin{multline*}
\projG{p}(x) = \\\frac{x + (x+3\text{ mod } 12)+(x+6\text{ mod } 12)+(x+9\text{ mod } 12)}{4\times 66}.
\end{multline*}
For example, for the distribution $p(x)=x/66$, 
\begin{align*}
\projG{p}(0)&=(0+3+6+9)/(4\times 66)&\approx 0.068\\
\projG{p}(1)&=(1+4+7+10)/(4\times 66)&\approx 0.083\\
\projG{p}(2)&=(2+5+8+11)/(4\times 66)&\approx 0.098\\
\projG{p}(3)&=(3+6+9+0)/(4\times 66)&\approx 0.068\\
\dots & \dots
\end{align*}
This twirled distribution is shown 
on the right panel of \cref{fig:symm}.

\begin{figure}
\begin{centering}
\includegraphics[width=0.75\columnwidth]{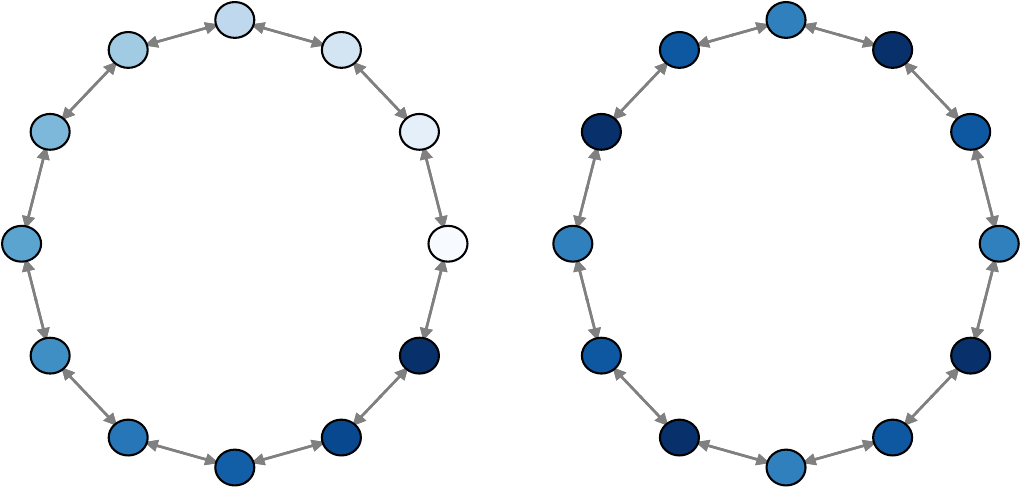}
\par\end{centering}
\caption{A unicyclic master equation over 12 states with 4-fold symmetry, as
in \cref{eq:symmRM}. Left: an initial distribution $p({x})\propto x$
which does not respect the 4-fold symmetry. Right: the twirling $\projG p$, which is 
invariant to the symmetry. (Colors indicate relative probability
assigned to each of the 12 states.) The extractable work 
depends on the accessible free energy in $p$, which is given by $\NFEs[\projG p]$. \label{fig:symm}}
\end{figure}

\subsection{Example: 2D Ising model, derivation of \cref{eq:analyticIsing}}
\label{app:ising}

We begin by recalling the expression for accessible information in our feedback-control protocol over the 2D Ising model, which appears as \cref{eq:hg7b} in the main text:
\begin{align}
&\IaccS[\projGbase] \!= %
\ln 2 - \Big\langle \!\ln \frac{q(m\vert x)}{N^{-2}\sum_{a,b} q(m\vert g_{a,b}(x))} \Big\rangle.\label{eq:hg7bApp}
\end{align}
Using $q(m\vert x)= \deltaFunc[m](x_1)$, the expectation term in \cref{eq:hg7bApp} can be  rewritten as
\begin{align}
-\sum_x p(x)\sum_{\mathclap{m\in\{-1,1\}}}\deltaFunc[m](x_1) \ln \big[N^{-2}{\sum_{a,b} \deltaFunc[m](\g_{a,b}(x)_1)}\big].
\label{eq:hg9}
\end{align}
Let $z(x) = (1+\sum_{i} x_{i}/N^2)/2$ indicate the magnetization of lattice state $x$, normalized to lie between 0 and 1.  Note that for any lattice state $x$, the frequency that spin 1 is in state 1 averaged across all translations is equal to the magnetization of $x$,
\[
N^{-2}\sum_{a,b} \deltaFunc[1](\g_{a,b}(x)_1)=z(x).
\] 
In addition, by symmetry, the probability that spin 1 is in state 1 averaged across all states that have magnetization $z$ is equal to $z$,
\[
\sum_x p(x \vert z) \deltaFunc[1](x_{1}) = z.
\] 
Using these results and $\deltaFunc[-1](x)=1-\deltaFunc[1](x)$, we can rewrite the expression in  \cref{eq:hg9} as
\begin{align}
&-\sum_x p(x) [\deltaFunc[1](x_1) \ln z(x) + (1-\deltaFunc[1](x_1)) \ln (1-z(x))]\nonumber\\
&=\sum_{z} p(z) [-z\ln z -(1-z)\ln(1-z)] \equiv \langle h_2(z) \rangle,
\label{eq:m23}
\end{align}
where $p(z')=\sum_x p(x)\deltaFunc[z'](z(x))$ is the probability that the system has magnetization $z'$ and $h_2$ is the binary entropy function. 

We now consider the $N\to\infty$ limit, and use Onsager's expression for the spontaneous magnetization for the 2D Ising model~\cite{yang1952spontaneous}. When $\beta$ is below the critical inverse temperature, $\beta_c=\ln(1+\sqrt{2})/2\approx 0.44$, 
 the magnetization distribution $p(z)$ concentrates at $z=1/2$, so \cref{eq:m23} approaches $h_2(1/2)=\ln 2$. 
 When $\beta> \beta_c$, the magnetization distribution concentrates on a uniform 
 mixture of two delta functions at $z=f(\beta)$ and $z=1-f(\beta)$, where 
$f(\beta)=(1+\sqrt[8]{1- (\sinh2\beta)^{-4}})/2$.
In this case, \cref{eq:m23} approaches $(h_2(f(\beta))+h_2(1-f(\beta)))/2=h_2(f(\beta))$. 
Combining these results with \cref{eq:hg7bApp} implies that %
 $\IaccS[\projGbase] =0$ for $\beta \le \beta_c$ and  $\IaccS[\projGbase] =\ln 2 - h_2(f(\beta))$ for  $\beta > \beta_c$, which appears as \cref{eq:analyticIsing} in the main text.

\section{Modularity constraints}
\label{app:mod}

\global\long\def\projMshortBase{\klopBase_\modDecomp}%
\global\long\def\projMshort#1{\projMshortBase(#1)}%
\global\long\def\olapsShort{O}%

\subsection{$\projMbase$ obeys the Pythagorean identity, \cref{eq:pyth}}

\global\long\def\projMapp#1{\projM{#1}}%
We show that $\projMshortBase$ obeys the Pythagorean identity:
\begin{align}
D(p\Vert \projMshort{q})=D(p\Vert\projMshort p)+D(\projMshort p\Vert \projMshort{q}).
\label{eq:pythmod}
\end{align}
for all $p,q\in\Dset$ such that $D(p\Vert \projG{q})<\infty$. 
For any $p,\ppdist \in \Dset$, %
\begin{align}
& \mathbb{E}_p[\ln \projMapp{\ppdist}] =\mathbb{E}_p[\ln \ppdist_{\olapsShort}] + \sum_{{\modSubsys\in\modDecomp}} \mathbb{E}_p [\ln \ppdist_{\modSubsys \setminus \olapsShort \vert \modSubsys \cap \olapsShort}]\nonumber\\
&\quad=\mathbb{E}_{\projMapp{p}}[\ln \ppdist_{\olapsShort}]+ \sum_{{\modSubsys\in\modDecomp}} \mathbb{E}_{\projMapp p} [\ln \ppdist_{\modSubsys \setminus \olapsShort \vert \modSubsys \cap \olapsShort}]\label{eq:appMM9a}\\
&\quad= \mathbb{E}_{\projMapp{p}}[\ln \projMapp{\ppdist}], \label{eq:appMM9}
\end{align}
where $a_{\olapsShort}$ and $a_{\modSubsys \setminus \olapsShort \vert \modSubsys \cap \olapsShort}$ indicate marginal and conditional distributions, respectively. 
In \cref{eq:appMM9a}, we used that $p$ and $\projMapp{p}$ have the same marginals over all subsystems all $\modSubsys\in\modDecomp$ as well as the overlap $\olapsShort$ (this can be verified from the definition of $\projMbase$, \cref{eq:projMbasedef}). %
Then,
\begin{align*}
D(p\Vert \projMapp{q}) &= D(p\Vert\projMapp{p} ) +  \mathbb{E}_p[\ln \projMapp{p} - \ln \projMapp{q}]\\
& = D(p\Vert\projMapp{p} ) +  \mathbb{E}_{\projMapp{p}}[\ln \projMapp{p} - \ln \projMapp{q}]\\
& =  D(p\Vert\projMapp{p} ) +  D(\projMapp{p} \Vert \projMapp{q}),
\end{align*}
where the second line follows by applying \cref{eq:appMM9} twice, first taking $\ppdist=p$ and then taking $\ppdist=q$.

\subsection{$\projMbase$ commutes with $e^{\inftime \LL}$}

\label{app:modcomm}
\global\long\def\subOnly{{\bm{\modSubsys}}}%
\global\long\def\oSsi{{\subOnly^c}}%

\newcommandx\eTA[1][usedefault, addprefix=\global, 1=\modSubsys]{e^{\inftime \Lmod[#1]}}%

\newcommand{\dX}{\deltaFunc[x]}
\def\dXS#1{\deltaFunc[x_#1]}

\newcommand{\dXx}{\dX(x')}
\def\dXSx#1{\dXS{#1}(x'_#1)}

\newcommand{\condP}{p}
\newcommand{\gAXbase}[1]{T^{(#1)}_\inftime}
\newcommand{\gAX}{\gAXbase{\modSubsys}(x'|x)}
\def\gAXmargS#1{\gAXbase{\modSubsys}(x'_#1|x)}

\newcommand{\gAXA}{\gAXbase{\modSubsys}({x_\subOnly'\vert x_\modSubsys})}

\newcommand\unionAllOthers{\bigcup_{\otherS \in \modDecomp\setminus \{\modSubsys \}} \otherS }

We show that if for some generator $\LL$,
\cref{eq:m2,eq:m3} hold for all $\modSubsys\in\modDecomp$, then
 $\projMshortBase$
and $e^{\inftime\LL}$ obey the commutativity relation of \cref{eq:comm0}.  We assume that all $\Lmod$ in \cref{eq:m3} are bounded linear operators. 

Before  deriving our result, 
we introduce some helpful notation:
\begin{enumerate}
\item  $\dXx$ indicates the delta function distribution over $X$ centered at $x$ (this is the Dirac delta for continuous $X$, and the Kronecker delta for discrete $X$).  
For any subsystem $S\subseteq V$, $\dXSx S$ indicates the delta function distribution over $X_S$ centered at $x_S$.  
\item  $\gAX = [\eTA \dX](x')$ indicates the conditional distribution over $X$, given that the system starts on state $x$ and then evolves under $\Lmod$ for time $\inftime$.  
\item For any $\modSubsys\in\modDecomp$, 
\[
{\textstyle \subOnly:=\modSubsys\setminus \big(\unionAllOthers\big)=A\setminus \olaps}
\]
indicates the set of degrees of freedom that belong exclusively to $\modSubsys\in\modDecomp$ (and no other subsystems), and 
\[
{\textstyle \oSsi := V\setminus \subOnly=\unionAllOthers}.
\] 
 indicates the complement of $\subOnly$, which is the set of degrees of freedom that fall into at least one of the other subsystem besides $\modSubsys$.
\end{enumerate}

To derive the commutativity relation, we proceed in three steps, which are described in detail in the subsections below. In the first step, we show that, for all $\inftime \ge 0$ and $\modSubsys\in\modDecomp$, the conditional distribution $\gAX$ can be written in the following product form:
\begin{align}
\gAX  = \gAXA \dXSx \oSsi.
\label{eq:appHa}
\end{align}
In the second step, we show that \cref{eq:appHa} implies the following commutativity relation for any $p\in\Dset$ and each $\modSubsys \in \modDecomp$:
\begin{align}
\eTA \projM{p} =  \projM{\eTA p}.
\label{eq:appHb}
\end{align}
In the third step, we show that the generators corresponding to all subsystems commute:
\begin{align}
\Lmod \Lmod[\otherS] = \Lmod[\otherS]\Lmod \qquad\forall \modSubsys,\otherS\in\modDecomp.
\label{eq:appHc}
\end{align}

We then combine these three results to show that  $\projMbase$ and $\eLb$ commute.  Write
\begin{align*}
\eLb \projM p =e^{\sum_{\modSubsys\in\modDecomp}\inftime\Lmod} \projM p=\prod_{\modSubsys\in\modDecomp}e^{\inftime\Lmod}\projM p.
\end{align*}
where we used  \cref{eq:m0,eq:appHc} to expand the operator exponential. Then, using \cref{eq:appHb}, write %
\[ 
\prod_{\modSubsys\in\modDecomp}e^{\inftime\Lmod}\projM p =\projMbase\Bigg( \prod_{\modSubsys\in\modDecomp}\eTA p\Bigg)=\projM {\eLb p}.
\]
Combining these two results implies that $\eLb \projM p=\projM {\eLb p}$ for all $p\in\Dset$ and  $\inftime \ge 0$, as in \cref{eq:comm0}.

\subsubsection{Derivation of \cref{eq:appHa}}

To derive \cref{eq:appHa}, consider the conditional distribution over $\subOnly$ given initial state $x$, as induced by $\Lmod$: %
\begin{align}
 \gAXmargS{\subOnly}&=[\eTA \dX]_\subOnly(x'_\subOnly) \nonumber \\
 &=[\dX]_\subOnly(x'_\subOnly)+\sum_{k\ge1}\frac{\inftime^{k}}{k!}[{\Lmod}^{k} \dX]_\subOnly(x'_\subOnly) \nonumber \\
 &=\dXSx \subOnly +\sum_{k\ge1}\frac{\inftime^{k}}{k!}[{\Lmod}^{k} \dX]_\subOnly(x'_\subOnly).\label{eq:appH9}
\end{align}
where in the second line we expanded  the operator exponential as $\eTA = \sum_k \inftime^k {\Lmod}^k/k!$.  
Note that $\subOnly\subseteq \modSubsys$, so %
$[\Lmod \dX]_{\subOnly}$ is a function of $[\Lmod \dX]_{\modSubsys}$, 
which in turn is a function of $x_{\modSubsys}$ 
by \cref{eq:m2}. 
Similarly, $\dXSx \subOnly$ depends only on $x_\modSubsys$, not $x$. 
This means the right hand side of \cref{eq:appH9} depends only on $x_\modSubsys$, which we indicate by
\begin{align}
  \gAXmargS{\subOnly} = \gAXA.
\label{eq:appH0}
\end{align}
Now consider the conditional distribution over any other subsystem $\otherS \ne \modSubsys$ given initial state $x$, as induced by $\Lmod$:
\begin{align}
 \gAXmargS{\otherS}&=\dXSx \otherS +\sum_{k\ge1}\frac{\inftime^{k}}{k!}[{\Lmod}^{k}\dX]_{\otherS}(x'_\otherS)\nonumber \\
 &=\dXSx \otherS ,\label{eq:n1}
\end{align}
where we used that $[\Lmod \dX]_{\otherS}=0$ by \cref{eq:m3}. 

Now, it is straightforward to show that if some distribution $p$ over $X_V$ has delta function marginals $p_\otherS = \dXS \otherS$ for all $\otherS \ne \modSubsys$, then $p$ must have following product form:
\begin{align}
\label{eq:appProdDelta0}
p(x') = p_{\subOnly}(x_{\subOnly}')\,\dXSx \oSsi,
\end{align}
where we use hat  $\oSsi=\unionAllOthers$. 
\cref{eq:appHa} follows by taking $p(x') = \gAX$ in \cref{eq:appProdDelta0}, while using \cref{eq:appH0}.

\subsubsection{Derivation of \cref{eq:appHb}}

Consider any $\inftime \ge0 $ and  $\modSubsys\in\modDecomp$. Using \cref{eq:m2} and the identity $\eTA = \sum_k \inftime^k {\Lmod}^k /k!$, one can show that whenever two distributions $p,q\in\Dset$ obey $p_A=q_A$, it must be that $[e^{\inftime\Lmod}p]_\modSubsys=[e^{\inftime\Lmod}q]_\modSubsys$. 
Since $p_A = [\projM p]_A$ (see the definition of $\projMbase$ in \cref{eq:projMbasedef}), 
\begin{align}
[e^{\inftime\Lmod}p]_\modSubsys=[e^{\inftime\Lmod}\klop p]_\modSubsys.
\label{eq:appG0}
\end{align}
In addition, given \cref{eq:n1}, we have 
$[e^{\inftime\Lmod}p]_\oSsi=p_\oSsi$. Given that $\otherS \subseteq \oSsi$ for each  $\otherS\ne \modSubsys$, we have
\begin{align}
[e^{\inftime\Lmod}p]_B&=p_B=\klop{p}_B= [e^{\inftime\Lmod}\klop{p}]_B.
\label{eq:appG1}
\end{align}
Similarly, $\olaps\subseteq \oSsi$ and therefore
\begin{align}
[e^{\inftime\Lmod}p]_{\olaps}&=[e^{\inftime\Lmod}\klop{p}]_{\olaps}.
\label{eq:appG2}
\end{align}

Now, observe that the distribution $\projM p$ does not depend
on the full distribution $p$, but only on the marginal distributions $p_{\olaps}$ and $\{p_\modSubsys\}_{\modSubsys\in\modDecomp}$.
By \cref{eq:appG0,eq:appG1,eq:appG2}, these marginals are the same for $\eTA p$ and  $\eTA \projM p$, which means that
\begin{align}
\projM{\eTA p} = \projM{\eTA\projM p}.
\label{eq:appG5}
\end{align}

Next, using \cref{eq:appHa} and some simple (but rather tedious) algebra, it can be shown that
\begin{align}
\eTA \projM p =  p_{\modSubsys\setminus O \vert \modSubsys \cap O}'\; p_O\; \prod_{{\otherS\ne \modSubsys}}p_{\otherS\setminus O \vert \otherS\cap O} \;,
\label{eq:appG6}
\end{align}
where 
\begin{multline}
p_{\modSubsys\setminus O \vert \modSubsys \cap O}'(x_{\modSubsys\setminus O}' \vert x_{\modSubsys \cap O}') = \\
\int \gAXbase{\modSubsys}({x_\subOnly'\vert x_\subOnly, x_{\modSubsys \cap O}}') p(x_{\subOnly}\vert x_{\modSubsys\cap O}')dx_{\subOnly},
\end{multline}
and we used the conditional distribution $\gAXbase{\modSubsys}({x_\subOnly'\vert x_\subOnly, x_{\modSubsys \cap O}})$ from  \cref{eq:appH0}. The right hand side of \cref{eq:appG6} has the form of the
right hand side of \cref{eq:projMbasedef}, so it is invariant under $\projMshortBase$:  
\begin{align}
\projMshort {e^{\inftime\Lmod}\projMshort {p}}=e^{\inftime\Lmod}\projMshort {p}.
\label{eq:appG6b}
\end{align}
\cref{eq:appHb} follows by combining \cref{eq:appG5,eq:appG6b}.

\subsubsection{Derivation of \cref{eq:appHc}}

Using \cref{eq:appHa} and some algebra, one can verify that for all $\inftime\ge0$ and $\modSubsys, \otherS\in\modDecomp$, 
\begin{multline}
\int \gAXbase{\modSubsys}(x''\vert x') \gAXbase{\otherS}(x'\vert x) \,dx'
\\=\int \gAXbase{\otherS}(x''\vert x') \gAX \, dx',
\end{multline}
which in operator notation can be written as
\begin{align}
\eTA \eTA[\otherS] \dX=  \eTA[\otherS] \eTA\dX.
\label{eq:appGG2}
\end{align}
Then, for any   function $f=\int f(x) \dX \,dx$, write
\begin{align*}
\eTA \eTA[\otherS]  f &= \eTA \eTA[\otherS] \int f(x) \dX\,dx \\
 &= \int f(x)\eTA \eTA[\otherS] \dX\,dx \\
&=\int f(x)  \eTA[\otherS] \eTA \dX \,dx\\
& = \eTA[\otherS] \eTA \int f(x)   \dX \,dx\\
&= \eTA[\otherS] \eTA f,
\end{align*}
where we exchanged the order of the bounded operators $\eTA \eTA[\otherS]$ and $ \eTA[\otherS] \eTA$ with the (Bochner) integral $\int f(x) \dX \,dx$, and used \cref{eq:appGG2}. 
This shows that $e^{\inftime\Lmod}$ and $e^{\inftime\Lmod[\otherS]}$ commute for all $\inftime \ge 0$, so their inverses $e^{-\inftime\Lmod}$ and $e^{-\inftime\Lmod[\otherS]}$ must also commute.  Given that $e^{\inftime\Lmod}$ and $e^{\inftime\Lmod[\otherS]}$ commute for all $\inftime \in \mathbb{R}$,  $\Lmod$ and $\Lmod[\otherS]$ must commute~\cite[p.~23]{engel_one-parameter_2000}.

\subsection{Szilard box: derivation of \cref{eq:cf01,eq:cf02}}
\label{app:szModDeriv}

We first derive \cref{eq:cf01}.  Using  \cref{eq:twirlMsz} and some rearrangement, write
\begin{align}
D(\projM{\ptheta}\Vert \po) =\ln 4 - S(\ptheta(X_1))-S({\ptheta}(X_2)),
\end{align}
where 
$S(\ptheta(X_1))$ and $S({\ptheta}(X_2))$ refer to the marginal entropies under $\ptheta$.  
It is easy to see that by symmetry,
\begin{align}
S(\ptheta(X_1))=S(p_{\frac{\pi}{2} - \theta}(X_2)).
\label{eq:XYsapp}
\end{align}
Therefore, we will derive a closed-form expression for $D(\projM{\ptheta}\Vert \po)$ by finding a closed-form expression for
\begin{align}
S(\ptheta(X_1)):=-\int_{-1}^{1} \ptheta(x_1) \ln \ptheta(x_1) \,dx_1.
\label{eq:Scf}
\end{align}

First, consider the case of $\theta\in [-\pi/2,\pi/2]$, and define $\Atheta:=|\tan\theta|$. It can be verified from \cref{eq:degMeas} that the marginal distribution $\ptheta(x_1)$ always has a piecewise linear form. In particular, if $\Atheta < 1$, then for any $x_1\in[-1,1]$,
\begin{align}
\ptheta(x_1) = \begin{cases}
1 & \text{if $-1 \le x_1\le -\Atheta$}\\
\frac{\Atheta - x_1}{2\Atheta} & \text{if $-\Atheta \le x_1 \le \Atheta$}\\
0 & \text{if $x_1> \Atheta$}
\end{cases}
\end{align}
Otherwise, if $\Atheta > 1$, then for any $x_1\in[-1,1]$,
\begin{align}
\ptheta(x_1) = \frac{\Atheta - x_1}{2\Atheta}.
\end{align}

Plugged into \cref{eq:Scf}, this gives
\begin{align}
S({\ptheta}(X_1)) &= \begin{cases}
-\int_{-1}^{1} \frac{\Atheta - x_1}{2\Atheta}\ln \frac{\Atheta - x_1}{2\Atheta}\, dx_1
& \text{if $\Atheta > 1$}\\
-\int_{-\Atheta}^{\Atheta} \frac{\Atheta - x_1}{2\Atheta}\ln \frac{\Atheta - x_1}{2\Atheta} \, dx_1
& \text{otherwise}
\end{cases}\nonumber
\end{align}
Integrating these two cases separately in \emph{Mathematica}, and plugging in the definition of $\Atheta$, gives
\begin{align}
S({\ptheta}(X_1))&=\frac{1}{2}\begin{cases}
f(|\tan \theta|) & \text{if $|\tan \theta| > 1$}\\
|\tan \theta|
& \text{otherwise}
\end{cases}
\label{eq:app321}
\end{align}
where for convenience we've defined
\begin{align}
f(x)=1-\frac{1+x^2}{2x} \ln \frac{x+1}{x-1} - \ln \frac{x^2-1}{4x^2}.
\end{align}
Recall that so far we assumed that  $\theta\in[-\pi/2,\pi/2]$. However, by \cref{eq:degMeas},  $\ptheta(x_1,x_2)=p_{\pm \pi-\theta}(-x_1,x_2)$, which implies that $\ptheta(x_1)=p_{\pi-\theta}(-x_1)=p_{-\pi-\theta}(-x_1)$ and $S({\ptheta}(X_1))=S(p_{\pi-\theta}(X_1))=S(p_{-\pi-\theta}(X_1))$. It can also be verified that $|\tan \theta|=|\tan (\pi-\theta)|=|\tan (-\pi-\theta)|$, so in fact \cref{eq:app321} holds for all $\theta\in[-\pi,\pi]$.

Finally, %
 if $|\theta|\in\angleRange$, then  \cref{eq:app321,eq:XYsapp} imply 
\begin{align*}
|\tan \theta|>1 ,\quad& S({\ptheta}(X_1)) = \frac{1}{2}f(|\tan \theta|)\\
|\tan ({\textstyle \frac{\pi}{2} - \theta})|\le 1,\quad& S({\ptheta}(X_2)) = \frac{1}{2}|\tan ({\textstyle \frac{\pi}{2} - \theta})|
\end{align*}
Conversely, if $|\theta|\in[0,\pi]\setminus \angleRange$, then
\begin{align*}
|\tan \theta|\le 1 ,\quad& S({\ptheta}(X_1)) = \frac{1}{2}|\tan \theta|\\
|\tan ({\textstyle \frac{\pi}{2} - \theta})|> 1,\quad& S({\ptheta}(X_2)) = \frac{1}{2}f(|\tan ({\textstyle \frac{\pi}{2} - \theta})|)
\end{align*}
\cref{eq:cf01} follows by combining these results  and rearranging.

To derive \cref{eq:cf02}, use $\projG{\projM{\ptheta}}(x_1,x_2)=\ptheta(x_1)\po(x_2)$ to write 
\begin{align}
D(\projG{\projM{\ptheta}}\Vert \po) &=\ln 4 - S(\ptheta(X_1))-S(\po(X_2)) \nonumber \\
&=\ln 2 - S(\ptheta(X_1)),\label{eq:app322}
\end{align}
where we used that $S(\po(X_2))=\ln 2$. \cref{eq:cf02} then follows by combining \cref{eq:app322,eq:app321}.

\subsection{Example: Feedback controlled flashing ratchet}
\label{app:ratchet}

Here we derive a closed-form expression for the accessible information in the feedback-controlled collective flashing ratchet.

For notational convenience, let $a=1/\alpha$ indicate the slope of the increasing part of $V$ in \cref{fig:genszilard}(b), and $b=-1/(1-\alpha)$ indicate the slope of the decreasing part of $V$.  
Note that the net force $\sum_v  V'(x_v)$ can be seen as the sum of $\numparticles$ random variables, where by assumption each $V'(x_v)$ is equal to $a=1/\alpha$ with probability $\alpha$ and equal to $b=-1/(1-\alpha)$ with probability $1-\alpha$.  This implies that the expectation of $V'(x_v)$ is 0 and the variance is $1/(\alpha(1-\alpha))$. 

We will first compute the accessible information $\IaccS[\projMbase]=\sum_v I(X_v;M)=\numparticles\cdot I(X_1;M)$. The mutual information between $M$ and the state of a single particle $X_1$ is given by 
\begin{align}
&I(X_1;M)=S(M)-S(M\vert X) \nonumber \\
&\quad=h_2(p(1))-\alpha h_2(p(1\vert a))-(1-\alpha)h_2(p(1\vert b))
\label{eq:singleMI},
\end{align}
where $p(1)$ is the probability that the net force is positive, $p(1\vert a)$ is the probability that the net force is positive given that particle $X_1$ experiences force $a$, and $p(1\vert b)$ is the probability that the net force is positive given that the particle $X_1$ experiences force $b$. 
We can compute $p(1)$ by considering the case when $k=0,1,2,\dots$ particles experience force $a$. Assuming the particles are independent, this is given by %
\begin{align}
p(1) = \sum_{\mathclap{k=0}}^{\numparticles} B_{\numparticles,\alpha}(k)\Theta(k a + (\numparticles-k)b) 
\label{eq:appP1}
\end{align}
where $B_{\numparticles,\alpha}$ is the binomial probability of $k$ successes, given  $\numparticles$ trials with success probability $\alpha$. 
To compute $p(1\vert a)$, note that, given that $X_1$  experiences force $a$, $M=1$ whenever the other $\numparticles-1$ particles experience a net force larger than $-a$. The probability of this event is
\begin{align}
p(1\vert a)=\sum_{\mathclap{k=0}}^{\numparticles-1} B_{\numparticles-1,\alpha}(k)\Theta(k a + (\numparticles-1-k)b + a).
\label{eq:appP2}
\end{align}
Conversely, if $X_1$ experiences force $b$, then $M=1$ if the other $\numparticles-1$ particles experience a net force larger than $-b$, which has probability 
\begin{align}
p(1\vert b)=\sum_{\mathclap{k=0}}^{\numparticles-1} B_{\numparticles-1,\alpha}(k)\Theta(k a + (\numparticles-1-k)b + b).
\label{eq:appP3}
\end{align}
Plugging \cref{eq:appP1,eq:appP2,eq:appP3} into \cref{eq:singleMI} gives $I(X_1;M)$. Multiplying by $\numparticles$ gives the accessible information,
\begin{align}
&\IaccS[\projMbase]=\numparticles \cdot I(X_1;M)=\label{eq:accFLR}\\
&\numparticles\Bigg[h_2 \Bigg(\sum_{\mathclap{k=0}}^{\numparticles} B_{\numparticles,\alpha}(k)\Theta(k a + (\numparticles-k)b)  \Bigg)-\nonumber\\
&\alpha h_2\Bigg(\sum_{\mathclap{k=0}}^{\numparticles-1} B_{\numparticles-1,\alpha}(k)\Theta(k a + (\numparticles-1-k)b + a) \Bigg)-\nonumber\\
&(1-\alpha) h_2\Bigg(\sum_{\mathclap{k=0}}^{\numparticles-1} B_{\numparticles-1,\alpha}(k)\Theta(k a + (\numparticles-1-k)b + b)\Bigg)\Bigg]\nonumber,
\end{align}
This is shown in \cref{fig:ratchet}(left) for different values of $\numparticles$ and $\alpha$.  

To compute the efficiency values in \cref{fig:ratchet}(right), we simply divide $\IaccS[\projMbase]$ by $\IXMs$ the total mutual information between the measurement and all particles. 
Since the measurement in \cref{eq:ratchetMeasurement} is deterministic, this mutual information is given by the entropy of $M$,
\begin{align}
\IXMs=S(M)=h_2(p(1)),
\end{align}
which can be computed using \cref{eq:appP1}.

We now compute the asymptotic value of accessible information and efficiency in the $\numparticles\to \infty$ limit. The sum of a large number of independent random variables with mean 0 and variance $1/(\alpha(1-\alpha))$ approaches a Gaussian with mean 0 and variance $\numparticles/(\alpha (1-\alpha))$. Thus, in the $\numparticles\to\infty$ limit, the probability that the force is positive converges to $p(1)=1/2$, so $\IXMs=S(M)$ converges to $\ln 2$.  Recall that $p(1\vert a)$ is given by the probability that $\numparticles-1$ particles experience a net force larger than $-a$.  In the $\numparticles\to\infty$ limit, this conditional probability converges to
\[
p(1\vert a)=1-\Phi_{\alpha,\numparticles-1}(-a)=\Phi_{\alpha,\numparticles-1}(a).
\]
where $\Phi_{\alpha,\numparticles-1}$ is the cumulative distribution function of a Gaussian with mean 0 and variance ${\numparticles}/(\alpha (1-\alpha))$. We can similarly calculate %
\[
p(1\vert b)=1-\Phi_{\alpha,\numparticles-1}(-b)=\Phi_{\alpha,\numparticles-1}(b).
\]
Plugging into \cref{eq:singleMI} gives
\begin{multline}
I(X_1;M)=\\\ln 2-\alpha h_2\big(\Phi_{\alpha,\numparticles-1}(a)\big)-(1-\alpha)h_2\big(\Phi_{\alpha,\numparticles-1}(b)\big).
\end{multline}
Using $a=1/\alpha$ and $b=-1/(1-\alpha)$ and some analysis (e.g., by taking limits in \emph{Mathematica}) shows that
\begin{align}
\lim_{\numparticles\to\infty} \numparticles\cdot I(X_1;M) = \frac{1}{\pi},
\end{align}
irrespective of $\alpha$. This is the asymptotic accessible information, which  appears as the dotted line in \cref{fig:ratchet}(left). The asymptotic efficiency, which appears as the dotted line in \cref{fig:ratchet}(right), is given by $1/(\pi\ln2)$  (since $\IXMs=\ln2$ in the $\numparticles\to\infty$ limit).

\section{Coarse-grained constraints}
\label{appsec:cg}

\subsection{Derivation of \cref{eq:closed0} from \cref{eq:closed0ME,eq:closed0FP}}

In general, the microstate distribution $p$ evolves according
to some generator $L$,  $\ppt p(t)=Lp(t)$, the macrostate distribution $\cgP$
evolves according to a coarse-grained generator $\cgLL^p$. In general, the coarse-grained dynamics will not be closed, meaning that $\cgLL^p$ can depend on the microstate distribution $p$.  In this section, we provide concrete conditions on the generators that guarantee that the coarse-grained dynamics are closed.
In the following derivations, for notational simplicity, we omit the dependence of $p(x,t)$ and $p(z,t)$ on $t$.

\def\cgLLp{\cgLL^p} 

For discrete-state master equations, the coarse-grained dynamics are given by~\citep{esposito2012stochastic}
\begin{equation}
\ppt \cgP(z)=\cgLLp \cgP(z)
=\sum_{\z}\Big[\sTrans[\cgLLp]{\zz}{\z}\cgP(\zz)-\sTrans[\cgLLp]{\z}{\zz}\cgP(\z)\Big],\label{eq:cgdiff0}
\end{equation}
 where $\sTrans[\cgLLp]{\zz}{\z}$ is the transition rate
from macrostate $\zz$ to  $\z$, 
\begin{equation}
\sTrans[\cgLLp]{\zz}{\z}=\sum_{x'}p({x'\vert\zz})\sum_{x}\deltaFunc[\cgf(x)](\z)\Lji.\label{eq:cg1}
\end{equation}
By
plugging \cref{eq:closed0ME} into \cref{eq:cg1} and simplifying, one can
verify that $\sTrans[\cgLLp]{\zz}{\z}$ does not depend on the microstate distribution $p$, therefore \cref{eq:closed0} holds.

A similar approach can be used for continuous-state master equations.

\def\cgDiff{\hat{\DL}}
\def\cgDrift{\hat{\uL}}

We now consider Fokker-Planck equations of
the form \cref{eq:cgFP}, given  a linear coarse-graining function $\cgf(x)=Wx$. %
Using 
\cite[Prop.~2.8]{duong2018quantification}, we write 
the evolution of the coarse-grained distribution $\cgP$ as
\begin{equation}
\ppt\cgP(z) =\nabla\cdot(\cgDrift(z) \cgP(z) )+\beta^{-1}\mathrm{tr}(H^T (\cgDiff(z) \cgP(z))),\label{eq:appcgfp1}
\end{equation}
where $H$ is the Hessian matrix of second derivative operators, and we've defined
\begin{align}
\cgDrift(z) & :=\int\left[p(x\vert z)W\nabla E(x)-\beta^{-1}\Delta\cgf(x)\right]dx \label{eq:appJ4}\\
&= \int\left[p(x\vert z)W\nabla E(x)\right]dx \label{eq:appJ5}\\
&= -\cgEnergy(z),\label{eq:appJ6}\\
\cgDiff(z) & :=\int p(x\vert z) WW^T\,dx=I\label{eq:appJ7}.
\end{align}
We used Eq.~2.29 from \cite{duong2018quantification} in \cref{eq:appJ4}, the linearity of $\cgf$ in \cref{eq:appJ5}, and \cref{eq:closed0FP} in \cref{eq:appJ6}. We used Eq.~2.30 from \cite{duong2018quantification} and the assumption that $WW^T=I$ in \cref{eq:appJ7}.  It is easy to check that $\mathrm{tr}(H^T (I \cgP))=\Delta \cgP$; combined with \cref{eq:appJ6,eq:appJ7,eq:appcgfp1}, this gives to \cref{eq:cgfp2}. Since the right hand side of \cref{eq:cgfp2} does not depend on the microstate distribution, the coarse-grained dynamics are closed.

\subsection{Derivation of \cref{eq:macroNaineq}}

\label{subsec:cgNabound}

\def\revTimeMarker{\tilde}
\global\long\def\eT{\inftime}%

\global\long\def\xF{\traj}%
\global\long\def\xR{\revTimeMarker{\bm{x}}^{\dagger}}%
\global\long\def\xRr{\revTimeMarker{\bm{x}}}%
\global\long\def\nuF{\bm{\nu}}%
\global\long\def\nuR{\revTimeMarker{\bm{\nu}}}%

\def\revP{\revTimeMarker{p}}
\global\long\def\sX{x}%
\global\long\def\sXx{x^\dagger}%
\global\long\def\eX{{x'}^{\dagger}}%
\global\long\def\sZ{z}%
\global\long\def\sZz{z^\dagger}%
\global\long\def\eZ{{z'}^\dagger}%
\global\long\def\eZz{z'}%
\global\long\def\eXx{x'}%

\global\long\def\ePp{\pf}%
\global\long\def\cgePp{\cgPf}%

\global\long\def\cgPi{\cgP[\pi]}%

Our derivation below does not assume isothermal protocols, so the inequality in \cref{eq:macroNaineq} holds both for isothermal protocols and for protocols connected to any number of thermodynamic reservoirs.

To derive this result for a given $\LL$, we make two assumptions. 
First, as described in the main text, we assume that the coarse-grained
dynamics are closed, \cref{eq:closed0}. 
Second, we assume that the coarse-grained stationary distribution
$\cgPi$ (where $\pi$ is the stationary distribution of $\LL$),  is invariant under conjugation of odd-parity variables,
\begin{equation}
\cgPi(\cgf(x))=\cgPi(\cgf(x^{\dagger}))\qquad\forall x\in X\label{eq:opvApp1}
\end{equation}
where $x^{\dagger}$ indicate the conjugation of state $x$ in which
all odd-parity variables (such as momentum) have their sign flipped. 
For an isothermal protocol, the stationary distributions are equilibrium distributions, and \cref{eq:opvApp1} is satisfied~\cite{lee_fluctuation_2013}.  For more general protocols, 
\cref{eq:opvApp1}  holds if there are no odd-parity
variables (e.g., overdamped dynamics), so $x=x^{\dagger}$. It  also holds if the coarse-graining function
maps each $x$ and its conjugate to the same macrostate, $\cgf(x)=\cgf(x^{\dagger})$, as well as some other cases.

Now imagine a system that starts from some initial distribution
$\ps$ at time $t=0$, and then undergoes free relaxation under $\LL$ towards a (possibly nonequilibrium)
stationary distribution $\pi$, reaching a final distribution $\pf$ by time $t=\eT$.  
Next, we use existing results in stochastic thermodynamics~\cite{esposito_three_2010,lee_fluctuation_2013} and write the EP incurred over time interval $t\in[0,\eT]$ as
\begin{align}
\EP(\eT)=D(p(\xF,\nuF)\Vert\revP(\xR,\nuR)),%
\label{eq:appS6}
\end{align}
(see also \cref{app:fluct}), where: 
\begin{enumerate}
\item $\xF = (x, \dots, x')$ indicate a continuous-time trajectory of system states over time interval $t\in[0,\eT]$, where $x$ and $x'$ indicate the initial and final system states respectively, and $\xR=(\eX,\dots,x^\dagger)$ is the corresponding
time-reversed and conjugated trajectory;
\item $\nuF$ is a sequence of reservoirs which exchange conserved quantities
with the system during $t\in[0,\eT]$ and $\nuR$ is the corresponding
time-reversed sequence~\cite{esposito_three_2010,van2010three,esposito2010three};
\item $p(\xF,\nuF)=P(\xF,\nuF\vert\sX)\ps(\sX)$ is the probability of forward
trajectory $(\xF,\nuF)$ given initial distribution $\ps$, where
$P(\xF,\nuF\vert\sX)$ is the conditional distribution generated by
the free relaxation;
\item $\revP(\xR,\nuR)=P(\xR,\nuR\vert\eX)\ePp(x')$ is the probability
of reverse trajectory $(\xR,\nuR)$ under a free relaxation that starts with the following distribution:
\begin{align}
p'(x')=\int P(x'\vert x)p(x) dx.\label{eq:appS4}
\end{align}
\end{enumerate}

Using the fact that EP decreases under state-space and temporal
coarse-graining~\cite{esposito2012stochastic,gomez2008cg}, we bound \cref{eq:appS6} as 
\begin{align}
\EP(\eT)\ge D(p(\xF)\Vert p(\xR))\ge D(p(\sZ,\eZz)\Vert\revP(\sZz,\eZ)),
\label{eq:appS8}
\end{align}
where $\sZ=\cgf(\sX)$, $\eZz=\cgf(\eXx)$, $\sZz=\cgf(\sXx)$, and
$\eZ=\cgf(\eX)$. The final KL divergence can be decomposed as
\begin{multline}
D(p(\sZ,\eZz)\Vert\revP(\sZz,\eZ))=\left[D(\cgP\Vert\cgPi)-D(\cgePp\Vert\cgPi)\right]+\\
\int p(\sZ,\eZz)\ln\left[\frac{p(\sZ,\eZz)\cgPi(\sZ)\cgePp(\eZz)}{\revP(\sZz,\eZ)\cgP(\sZ)\cgPi(\eZz)}\right]d\sZ \,d\eZz.\label{eq:appS0}
\end{multline}
Using Jensen's inequality, we lower bound the integral term as
\begin{align}
&\int p(\sZ,\eZz)\ln\left[\frac{p(\sZ,\eZz)\cgPi(\sZ)\cgePp(\eZz)}{\revP(\sZz,\eZ)\cgP(\sZ)\cgPi(\eZz)}\right]d\sZ \,d\eZz\nonumber\\
&\quad=-\int p(\sZ,\eZz)\ln\left[\frac{\revP(\sZz,\eZ)\cgP(\sZ)\cgPi(\eZz)}{p(\sZ,\eZz)\cgPi(\sZ)\cgePp(\eZz)}\right]d\sZ \,d\eZz\nonumber\\
&\quad\ge-\ln\left[\int\frac{\revP(\sZz,\eZ)\cgP(\sZ)\cgPi(\eZz)}{\cgPi(\sZ)\cgePp(\eZz)}d\sZ \,d\eZz\right].\label{eq:appS2}
\end{align}
Note that $\cgPi(\eZz)=\cgPi(\eZ)$ by \cref{eq:opvApp1},
and $\revP_Z(\eZ)=\cgePp(\eZz)$ by the definition of $\cgePp$ in \cref{eq:appS4}, allowing us to 
rewrite the RHS of \cref{eq:appS2} as 
\begin{equation}
-\ln\left[\int\frac{\cgP(\sZ)}{\cgPi(\sZ)}\left[\int\revP(\sZz\vert\eZ)\cgPi(\eZ)d\eZz\right]d\sZ\right].\label{eq:appS1}
\end{equation}
The inner integral can be further rewritten as
\begin{align*}
\int\revP(\sZz\vert\eZ)\cgPi(\eZ)d\eZz & =\int P(\sZz\vert\eX)\revP(\eX\vert\eZ)\cgPi(\eZ)d\eXx\\
 & =\cgPi(\sZz)\\
 & =\cgPi(\sZ),
\end{align*}
where in the second line we used the assumption of closed dynamics (\cref{eq:closed0}) and the stationarity of $\pi$ under $P(\cdot\vert\cdot)$, and in the third line we used \cref{eq:opvApp1}.
We can then rewrite \cref{eq:appS1} as \[
-\ln\left[\int\frac{\cgPi(\sZ)}{\cgPi(\sZ)}\cgPi(\sZ)\,d\sZ\right]=0.
\]
Combined with \cref{eq:appS2}, this implies that the integral term
in \cref{eq:appS0} is non-negative. Combining with \cref{eq:appS8} gives
\[
\EP(\eT)\ge D(\cgP\Vert\cgPi)-D(\cgePp\Vert\cgPi).
\]
Finally, using the definition of the EP rate and the results above,  
\begin{align}
\EPr(p,L)& := \lim_{\eT \to 0} \frac{1}{\eT} \EP(\eT) \nonumber\\
&\ge \lim_{\eT \to 0} \frac{1}{\eT} [D(\cgP\Vert\cgPi)-D(\cgePp\Vert\cgPi)]\nonumber \\
&= -\int \ppt \cgP(t)(z) \ln \frac{ \cgP(z)}{ \cgPi(z)}\,dz \ge0,%
\label{eq:gf76}
\end{align}
where $ \ppt \cgP(t)=\cgLL \cgP$. \cref{eq:gf76} follows from \cref{eq:app0,eq:app0b,eq:app0c,eq:app0d} above (with summations replaced by integrals). The discrete-state form of \cref{eq:gf76}, and also where $p$ and $\LL$ are explicitly time-dependent, appears in the main text as \cref{eq:macroNaineq}.